\documentclass[letterpaper,11pt]{article}

\usepackage{verbatim, url}
\usepackage{latexsym}
\usepackage{amsmath,amssymb,amsthm}
\usepackage{epsfig}
\usepackage{fullpage}

\usepackage{tikz, fp, subfigure, epsfig}
\usepackage{hyperref, cite}
\usepackage{color}

\newcommand\cD{\mathcal{D}}
\newcommand\cC{\mathcal{C}}

\newcommand\cA{\mathcal{A}}

\renewcommand{\paragraph}[1]{{\medskip\par\noindent\bf #1.}}

\newcommand\ptsubsection[1]{\subsection{{#1}}}

\newcommand{\figurewidthA}{0.475\columnwidth}

\let\myPushQED=\pushQED
\let\myPopQED=\popQED
\newcommand{\myignore}[1]{}
\newenvironment{proof*}
  {\let\pushQED=\myignore\begin{proof}\let\pushQED=\myPushQED}
  {\def\popQED{}\end{proof}\let\popQED=\myPopQED}

\newenvironment{description*}%
  {\vspace{-1ex}\begin{description}%
    \setlength{\itemsep}{-0.5ex}%
    \setlength{\parsep}{0pt}}%
  {\end{description}}
\newenvironment{itemize*}%
  {\vspace{-1ex}\begin{itemize}%
    \setlength{\itemsep}{-0.5ex}%
    \setlength{\parsep}{0pt}}%
  {\end{itemize}}
\newenvironment{enumerate*}%
  {\vspace{-1ex}\begin{enumerate}%
    \setlength{\itemsep}{-0.5ex}%
    \setlength{\parsep}{0pt}}%
  {\end{enumerate}}

{\makeatletter
 \gdef\xxxmark{%
   \expandafter\ifx\csname @mpargs\endcsname\relax 
     \expandafter\ifx\csname @captype\endcsname\relax 
       \marginpar{xxx}
     \else
       xxx 
     \fi
   \else
     xxx 
   \fi}
 \gdef\xxx{\@ifnextchar[\xxx@lab\xxx@nolab}
 \long\gdef\xxx@lab[#1]#2{{\bf [\xxxmark #2 ---{\sc #1}]}}
 \long\gdef\xxx@nolab#1{{\bf [\xxxmark #1]}}
}

\newtheorem{theorem}{Theorem}
\newtheorem{lemma}[theorem]{Lemma}
\newtheorem{corollary}[theorem]{Corollary}
\newtheorem{claim}[theorem]{Claim}
\newtheorem{fact}[theorem]{Fact}
\newtheorem{observation}[theorem]{Observation}

\newtheorem{proposition}[theorem]{Proposition}

\newcommand{\eps}{\varepsilon}
\newcommand{\twodots}{\mathinner{\ldotp\ldotp}}
\newcommand{\proc}[1]{\textnormal{\scshape#1}}

\newcommand{\E}{\mathbf{E}}

\newcommand{\evt}{\mathcal{E}}

\newcommand{\calH}{\mathcal{H}}

\let\phi=\varphi

\newcommand{\Id}{\proc{Id}}
\newcommand{\Hashes}{\proc{Hashes}}
\newcommand{\Coll}{\proc{Coll}}

\renewcommand{\th}{\ifmmode{^{\textrm{th}}}\else{\textsuperscript{th}\ }\fi}
\newcommand\drop[1]{}
\newcommand\req[1]{(\ref{#1})}

\author{
     Mihai P\v{a}tra\c{s}cu \\ AT\&T Labs
\and Mikkel Thorup \\ AT\&T Labs
}

\begin{document}

\title{The Power of Simple Tabulation Hashing} 

\maketitle

\begin{abstract}
Randomized algorithms are often enjoyed for their simplicity, but the
hash functions used to yield the desired theoretical guarantees are
often neither simple nor practical. Here we show that the simplest
possible tabulation hashing provides unexpectedly strong guarantees.

The scheme itself dates back to Carter and Wegman (STOC'77).  Keys are
viewed as consisting of $c$ characters. We initialize $c$ tables $T_1,
\dots, T_c$ mapping characters to random hash codes. A key
$x=(x_1, \dots, x_c)$ is hashed to $T_1[x_1] \oplus \cdots \oplus
T_c[x_c]$, where $\oplus$ denotes xor.

While this scheme is not even 4-independent, we show that it provides
many of the guarantees that are normally obtained via higher
independence, e.g., Chernoff-type concentration, min-wise hashing for
estimating set intersection, and cuckoo hashing.
\end{abstract}

\setcounter{tocdepth}{1}
\tableofcontents

\newpage

\section{Introduction}

An important target of the analysis of algorithms is to
determine whether there exist \emph{practical} schemes, which enjoy
mathematical \emph{guarantees} on performance. 

Hashing and hash tables are one of the most common inner loops in
real-world computation, and are even built-in ``unit cost'' operations
in high level programming languages that offer associative
arrays. Often, these inner loops dominate the overall computation
time.  Knuth gave birth to the analysis of algorithms in 1963
\cite{knuth63linprobe} when he analyzed linear probing, the most
popular practical implementation of hash tables. Assuming a perfectly
random hash function, he bounded the expected number of
probes. However, we do not have perfectly random hash functions. The
approach of algorithms analysis is to understand when simple and
practical hash functions work well. The most popular
multiplication-based hashing schemes maintain the $O(1)$ running times
when the sequence of operations has sufficient randomness
\cite{mitzenmacher08hash}. However, they fail badly even for very
simple input structures like an interval of consecutive
keys~\cite{pagh07linprobe, patrascu10kwise-lb,thorup10kwise}, giving linear probing
an undeserved reputation of being non-robust.

On the other hand, the approach of algorithm design (which may still
have a strong element of analysis) is to construct
(more complicated) hash functions providing the desired mathematical
properties. This is usually done in the influential $k$-independence
paradigm of Wegman and Carter~\cite{wegman81kwise}. It is known that
5-independence is sufficient \cite{pagh07linprobe} and necessary
\cite{patrascu10kwise-lb} for linear probing. Then one can use the
best available implementation of 5-independent hash functions, the
tabulation-based method of \cite{thorup04kwise,thorup10kwise}. 

Here we analyze simple tabulation hashing. This scheme views a key
$x$ as a vector of $c$ characters $x_1, \dots, x_c$. For each
character position, we initialize a totally random table $T_i$, and
then use the hash function 
\vspace{-.2cm}
\[h(x) = T_1[x_1] \oplus \dots \oplus T_c[x_c].\] 
This is a well-known
scheme dating back at least to Wegman and
Carter~\cite{wegman81kwise}. From a practical view-point, tables $T_i$
can be small enough to fit in fast cache, and the function is probably
the easiest to implement beyond the bare multiplication. However, the
scheme is only 3-independent, and was therefore assumed to have weak
mathematical properties.  We note that if the keys are drawn from a
universe of size $u$, and hash values are machine words, the space
required is $O(c u^{1/c})$ words.  The idea is to make
this fit in fast cache. We also note that the hash values are bit strings,
so when we hash into bins, the number of bins is generally understood
to be a power of two.

The challenge in analyzing simple tabulation is the significant
dependence between keys. Nevertheless, we show that the scheme works
in some of the most important randomized algorithms, including linear
probing and several instances when $\Omega(\lg n)$-independence was
previously needed. We confirm our findings by experiments: simple
tabulation is competitive with just one 64-bit multiplication, and the
hidden constants in the analysis appear to be very acceptable in practice.

In many cases, our analysis gives the first provably good 
\emph{implementation} of
an algorithm which matches the algorithm's conceptual simplicity if
one ignores hashing.

\paragraph{Desirable properties}
We will focus on the following popular properties of truly random hash
functions.
\begin{itemize*}
\item The worst-case query time of chaining is $O(\lg n / \lg\lg n)$
  with high probability (w.h.p.). More generally, when distributing
  balls into bins, the bin load obeys Chernoff bounds.

\item Linear probing runs in expected $O(1)$ time per operation. Variance
and all constant moments are also $O(1)$.


\item Cuckoo hashing: Given two tables of size $m \ge (1+\eps) n$, it
  is possible to place a ball in one of two randomly chosen locations
  without \emph{any} collision, with probability $1- O(\frac{1}{n})$.

\item Given two sets $A,B$, we have $\Pr_h[\min h(A) = \min h(B)] =
  \frac{|A\cap B|}{|A\cup B|}$. This can be used to quickly estimate the
  intersection of two sets, and follows from a property called
  \emph{minwise independence}: for any $x\notin S$, $\Pr_h[x < \min
    h(S)] = \frac{1}{|S|+1}$.
\end{itemize*}

As defined by Wegman and Carter~\cite{wegman81kwise} in 1977, a family
$\calH = \{ h : [u] \to [m] \}$ of hash functions is $k$-independent
if for any distinct $x_1, \dots, x_k \in [u]$, the hash codes $h(x_1),
\dots, h(x_k)$ are independent random variables, and the hash code of
any fixed $x$ is uniformly distributed in $[m]$.

Chernoff bounds continue to work with high enough
independence~\cite{schmidt95chernoff}; for instance, independence
$\Theta(\frac{\lg n}{\lg\lg n})$ suffices for the bound on the maximum
bin load. For linear probing, 5-independence is sufficient
\cite{pagh07linprobe} and necessary \cite{patrascu10kwise-lb}. For
cuckoo hashing, $O(\lg n)$-independence suffices and at least
6-independence is needed~\cite{cohen-cuckoo}. While minwise
independence cannot be achieved, one can achieve $\eps$-minwise
independence with the guarantee $(\forall) x\notin S, \Pr_h[x < \min
  h(S)] = \frac{1\pm \eps}{|S|+1}$. For this, $\Theta(\lg
\frac{1}{\eps})$ independence is sufficient \cite{indyk01minwise} and
necessary \cite{patrascu10kwise-lb}. (Note that the $\eps$ is a bias so
it is a lower
bound on how well set intersection can be approximated, with any
number of independent experiments.)

The canonical construction of $k$-independent hash functions is a
random degree $k-1$ polynomial in a prime field, which has small
representation but $\Theta(k)$ evaluation time. 
Competitive implementations of polynomial hashing simulate arithmetic
modulo Mersenne primes via bitwise operations. 
Even so, tabulation-based hashing with $O(u^{1/c})$ space and $O(ck)$
evaluation time is significantly faster~\cite{thorup04kwise}. The linear
dependence on  $k$ is problematic, e.g., when $k\approx \lg n$. 

Siegel~\cite{siegel04hash} shows that a family with superconstant
independence but $O(1)$ evaluation time requires $\Omega(u^\eps)$
space, i.e.~it requires tabulation. He also gives a solution that uses
$O(u^{1/c})$ space, $c^{O(c)}$ evaluation time, and achieves
$u^{\Omega(1/c^2)}$ independence (which is superlogarithmic, at least
asymptotically). The construction is non-uniform, assuming a certain
small expander which gets used in a graph product. Dietzfelbinger and
Rink \cite{dietzfel09splitting} use universe splitting to obtain
similar high independence with some quite different costs. Instead of being
highly independent on the whole universe, their goal is to be highly
independent on an unknown but fixed set $S$ of size $n$. 
For some constant parameter $\gamma$, they tolerate an error probability of
$n^{-\gamma}$. Assuming no error, their hash function is highly independent
on $S$. The evaluation time is constant and the space is sublinear. 
For error probability $n^{-\gamma}$, each hash computation calls
$O(\gamma)$ subroutines, each of which evaluates its 
own degree $O(\gamma)$ 
polynomial. The price for a lower error tolerance is therefore a slower hash 
function (even if we only count it as constant time in theory).

While polynomial hashing may perform better than its independence suggests,
we have no positive example yet. On the tabulation front, we have one
example of a good hash function that is not formally $k$-independent:
cuckoo hashing works with an ad hoc hash function that combines space
$O(n^{1/c})$ and polynomials of degree $O(c)$
\cite{dietzfel03tabhash}.

\subsection{Our results}

Here we provide an analysis of simple tabulation showing that it has
many of the desirable properties above. For most of our applications,
we want to rule out certain obstructions with high probability. This
follows immediately if certain events are independent, and the
algorithms design approach is to pick a hash function guaranteeing
this independence, usually in terms of a highly independent hash
function.

Instead we here stick with simple tabulation with all its
dependencies.  This means that we have to struggle in each individual
application to show that the dependencies are not fatal. However, from
an implementation perspective, this is very attractive, leaving us
with one simple and fast scheme for (almost) all our needs.

In all our results, we assume the number of characters is
$c=O(1)$. The constants in our bounds will depend on $c$. Our results
use a rather diverse set of techniques analyzing the table
dependencies in different types of problems. For chaining and linear
probing, we rely on some concentration results, which will also be
used as a starting point for the analysis of min-wise hashing.
Theoretically, the most interesting part is the analysis for cuckoo
hashing, with a very intricate study of the random graph constructed by the
two hash functions.

\paragraph{Chernoff bounds}
We first show that simple tabulation preserves Chernoff-type
concentration:
\begin{theorem}   \label{thm:chernoff} Consider hashing $n$ balls into $m\ge n^{1-1/(2c)}$ bins by simple tabulation.
Let $q$ be an additional
\emph{query ball}, and define $X_q$ as the number of regular balls
that hash into a bin chosen as a function of $h(q)$. Let $\mu =
\E[X_q] = \frac{n}{m}$. The following probability bounds hold
for any constant $\gamma$:
\begin{align}
(\forall) \delta\leq 1: & 
\Pr[|X_q-\mu|>\delta \mu]< 2e^{-\Omega(\delta^2\mu)} + m^{-\gamma}
\label{eq:normal} \\
(\forall) \delta=\Omega(1): &
  \Pr[X_q>(1+\delta)\mu]<(1+\delta)^{-\Omega((1+\delta)\mu)} + m^{-\gamma}
\label{eq:Poisson}
\end{align}
With $m\leq n$ bins, every bin gets
\begin{equation}\label{eq:few-bins}
n/m\pm O\left(\sqrt{n/m}\log^c n\right).
\end{equation}
keys with probability $1-n^{-\gamma}$.
\end{theorem}

Contrasting standard Chernoff bounds (see, e.g.,
\cite{motwani95book}), Theorem \ref{thm:chernoff} can only provide
polynomially small probability, i.e.~at least $n^{-\gamma}$ for any
desired constant $\gamma$. In addition, the exponential dependence on
$\mu$ in \req{eq:normal} and \req{eq:Poisson} is reduced by a constant
which depends (exponentially) on the constants $\gamma$ and $c$. It is
possible to get some super polynomially small bounds with super
constant $\gamma$ but they are not as clean.  An alternative way to
understand the bound is that our tail bound depends exponentially on
$\eps \mu$, where $\eps$ decays to subconstant as we move more than
inversely polynomial out in the tail. Thus, our bounds are sufficient
for any polynomially high probability guarantee. However, compared to
the standard Chernoff bound, we would have to tolerate a constant
factor more balls in a bin to get the same failure probability.

By the union bound (\ref{eq:normal}) implies that with $m=\Theta(n)$ bins,
no bin receives more
than $O(\lg n / \lg\lg n)$ balls w.h.p. This is the first realistic
hash function to achieve this fundamental property. Similarly, for
linear probing with fill bounded below $1$, (\ref{eq:Poisson}) shows
that the longest filled interval is of length $O(\log n)$ w.h.p.

\paragraph{Linear probing}
Building on the above concentration bounds, we show that if the table
size is $m=(1+\eps) n$, then the expected time per operation is
$O(1/\eps^2)$, which asymptotically matches the bound of Knuth
\cite{knuth63linprobe} for a truly random function. In particular,
this compares positively with the $O(1/\eps^{13/6})$ bound of
\cite{pagh07linprobe} for 5-independent hashing.  

Our proof is a combinatorial reduction that relates the
performance of linear probing to concentration bounds. 
The results hold for any hash function with concentration similar to
Theorem \ref{thm:chernoff}. To illustrate the generality of 
the approach, we also improve the $O(1/\eps^{13/6})$ bound from
\cite{pagh07linprobe} for 5-independent hashing to the optimal
$O(1/\eps^2)$. This was raised as an open problem in \cite{pagh07linprobe}. 

For simple tabulation, we get quite strong concentration results
for the time per operation, e.g,, constant variance
for constant $\eps$. For contrast, with 5-independent hashing,
the variance is only known to be $O(\log n)$
\cite{pagh07linprobe,thorup10kwise}.

\paragraph{Cuckoo hashing}
In general, the cuckoo hashing algorithm fails iff the random
bipartite graph induced by two hash functions contains a component
with more vertices than edges. With truly random hashing, this happens
with probability $\Theta(\frac{1}{n})$. Here we study the random
graphs induced by simple tabulation, and obtain a rather unintuitive
result: the optimal failure probability is inversely proportional to
the \emph{cube root} of the set size.

\begin{theorem}\label{thm:cuckoo}
Any set of $n$ keys can be placed in two table of size $m=(1+\eps)$ by
cuckoo hashing and simple tabulation with probability $1 -
O(n^{-1/3})$. There exist sets on which the failure probability is
$\Omega(n^{-1/3})$.
\end{theorem}

Thus, cuckoo hashing and simple tabulation are an excellent
construction for a static dictionary. The dictionary can be built (in
linear time) after trying $O(1)$ independent hash functions w.h.p.,
and later every query runs in constant worst-case time with two
probes. We note that even though cuckoo hashing requires two
independent hash functions, these essentially come for the cost of one
in simple tabulation: the pair of hash codes can be stored
consecutively, in the same cache line, making the running time
comparable with evaluating just one hash function.

In the dynamic case, Theorem \ref{thm:cuckoo} implies that 
we expect $\Omega(n^{4/3})$ updates between failures
requiring a complete rehash with new hash functions.

Our proof involves a complex understanding of the intricate, yet not
fatal dependencies in simple tabulation. The proof is a (complicated)
algorithm that assumes that cuckoo hashing has failed, and uses this
knowledge to compress the random tables $T_1, \dots, T_c$ below the
entropy lower bound. 

%

Using our techniques, it is also possible to show that if $n$ balls are
placed in $O(n)$ bins in an online fashion, choosing the least loaded
bin at each time, the maximum load is $O(\lg\lg n)$ in expectation.

\paragraph{Minwise independence}
In the full version, we show that simple tabulation is $\eps$-minwise
independent, for a vanishingly small $\eps$ (inversely polynomial in
the set size). This would require $\Theta(\log n)$ independence by
standard techniques.

\begin{theorem}
Consider a set $S$ of $n=|S|$ keys and $q\notin S$. Then with $h$
implemented by simple tabulation:
\[ \Pr[ h(q) < \min h(S)] = \frac{1\pm \eps}{n}, \qquad
\textrm{where } \eps = O\left( \frac{\lg^2 n}{n^{1/c}} \right). \]
\end{theorem}

This can be used to estimate the size of set intersection by estimating:
\begin{align*}
\Pr[ \min h(A) =& \min h(B) ] \\
&=~ \sum_{x\in A\cap B} \Pr[ x < \min h(A\cup B \setminus \{x\})] 
\\
&=~ \frac{|A\cap B|}{|A\cup B|} \cdot \left( 1 \pm \widetilde{O}\left(
\frac{1}{|A\cup B|^{1/c}} \right) \right).
\end{align*}
For good bounds on the probabilities, we would make multiple
experiments with independent hash functions. An alternative based on
a single hash function is that we for each set consider the $k$ elements 
with the smallest hash values. We will also present concentration bounds 
for this alternative.

\paragraph{Fourth moment bounds}
An alternative to Chernoff bounds in proving good concentration is to
use bounded moments. In the full version of the paper, we analyze the
4\th moment of a bin's size when balls are placed into bins by simple
tabulation. For a fixed bin, we show that the 4\th moment comes
extremely close to that achieved by truly random hashing: it deviates
by a factor of $1 + O(4^c / m)$, which is tiny except for a very large
number of characters $c$. This would require 4-independence by
standard arguments. This limited 4\th moment for a given bin was
discovered independently by \cite{braverman10kwise}.

If we have a designated query ball $q$, and we are interested in the
size of a bin chosen as a function of $h(q)$, the 4\th moment of
simple tabulation is within a constant factor of that achieved by
truly random hashing (on close inspection of the proof, that constant
is at most 2). This would require 5-independence by standard
techniques. (See \cite{patrascu10kwise-lb} for a proof that
4-independence can fail quite badly when we want to bound the size of
the bin in which $q$ lands.) Our proof exploits an intriguing
phenomenon that we identify in simple tabulation: in any fixed set of
5 keys, one of them has a hash code that is independent of the other
four's hash codes.

Unlike our Chernoff-type bounds, the constants in the 4\th moment
bounds can be analyzed quite easily, and are rather tame. Compelling
applications of 4\th moment bounds were given by \cite{karloff93prg}
and \cite{thorup09linprobe}. In \cite{karloff93prg}, it was shown that
any hash function with a good 4\th moment bound suffices for a
nonrecursive version of quicksort, routing on the hypercube, etc. In
\cite{thorup09linprobe}, linear probing is shown to have constant
expected performance if the hash function is a composition of
universal hashing down to a domain of size $O(n)$, with a strong
enough hash function on this small domain (i.e.~any hash function with
a good 4\th moment bound).

We will also use 4\th moment bounds to attain certain bounds of
linear probing not covered by our Chernoff-type bounds.  In the case
of small fill $\alpha=\frac nm=o(1)$, we use the 4\th moment bounds to
show that the probability of a full hash location is $O(\alpha)$.

\paragraph{Pseudorandom numbers}
The tables used in simple tabulation should be small to fit in the
first level of cache. Thus, filling them with truly random numbers
would not be difficult (e.g.~in our experiments we use atmospheric
noise from \url{random.org}). If the amount of randomness needs to be
reduced further, we remark that all proofs continue to hold if the
tables are filled by a $\Theta(\lg n)$-independent hash function
(e.g.~a polynomial with random coefficients).

With this modification, simple tabulation naturally lends itself to an
implementation of a very efficient pseudorandom number generator. We
can think of a pseudorandom generator as a hash function on range
$[n]$, with the promise that each $h(i)$ is evaluated once, in the
order of increasing $i$. To use simple tabulation, we break the
universe into two, very lopsided characters: $[\frac{n}{R}] \times [R]$, for $R$ chosen to be $\Theta(\lg
n)$. Here the second coordinate is least significant, that
is, $(x,y)$ represents $xR+y$. 
During initialization, we fill $T_2[1\twodots R]$ with $R$ truly
random numbers. The values of $T_1[1 \twodots n/R]$ are generated on
the fly, by a polynomial of degree $\Theta(\lg n)$, whose coefficients
were chosen randomly during initialization. Whenever we start a new row
of the matrix, we can spend a relatively large amount of time to
evaluate a polynomial to generate the next value $r_1$ which we store in a 
register. For the next $R$ calls, we run sequentially through $T_2$, xoring
each value with $r_1$ to provide a new pseudorandom number. With
$T_2$ fitting in fast memory and scanned sequentially, this will be
much faster than a single multiplication, and with $R$ large, the
amortized cost of generating $r_1$ is insignificant. The
pseudorandom generator has all the interesting properties discussed
above, including Chernoff-type concentration, minwise independence,
and random graph properties.

\paragraph{Experimental evaluation}
We performed an experimental
evaluation of simple tabulation. Our implementation uses tables of
$256$ entries (i.e.~using $c=4$ characters for 32-bit data and $c=8$
characters with 64-bit data). The time to evaluate the hash function
turns out to be competitive with multiplication-based 2-independent
functions, and significantly better than for hash functions
with higher independence. We also evaluated simple tabulation in applications, in an effort to
verify that the constants hidden in our analysis are not too large. Simple
tabulation proved very robust and fast, both for linear probing and
for cuckoo hashing.

\paragraph{Notation}
We now introduce some notation that will be used throughout the proofs.
We want to construct hash functions $h : [u] \to [m]$. We use simple
tabulation with an alphabet of $\Sigma$ and $c = O(1)$ characters.
Thus, $u = \Sigma^c$ and $h(x_1, \dots, x_c) = \bigoplus_{i=1}^c
T_i[x_i]$. It is convenient to think of each hash code $T_i[x_i]$ as a
fraction in $[0,1)$ with large enough precision. We always assume $m$
  is a power of two, so an $m$-bit hash code is obtained by keeping
  only the most significant $\log_2 m$ bits in such a fraction. We
  always assume the table stores long enough hash codes, i.e.~at least
  $\log_2 m$ bits.

Let $S \subset \Sigma^c$ be a set of $|S|=n$ keys, and let $q$ be a
query. We typically assume $q \notin S$, since the case $q\in S$ only
involves trivial adjustments (for instance, when looking at the load
of the bin $h(q)$, we have to add one when $q\in S$). Let $\pi(S, i)$
be the projection of $S$ on the $i$-th coordinate, $\pi(S,i) = \{ x_i
\mid (\forall) x \in S\}$.

We define a \emph{position-character} to be an element of $[c] \times
\Sigma$. Then, the alphabets on each coordinate can be assumed to be
disjoint: the first coordinate has alphabet $\{1\} \times \Sigma$, the
second has alphabet $\{2\} \times \Sigma$, etc. Under this view, we
can treat a key $x$ as a \emph{set} of $q$ position-characters (on
distinct positions). Furthermore, we can assume $h$ is defined on
position characters: $h((i,\alpha)) = T_i[\alpha]$. This definition is
extended to keys (sets of position-characters) in the natural way
$h(x) = \bigoplus_{\alpha \in x} h(\alpha)$.

When we say with high probability in $r$, we mean $1-r^a$ for any
desired constant $a$. Since $c=O(1)$, high probability in $|\Sigma|$
is also high probability in $u$. If we just say high probability, it
is understood to be in $n$.

\section{Concentration Bounds}\label{sec:chernoff}

This section proves Theorem~\ref{thm:chernoff}, except branch
\req{eq:few-bins} which is shown in the full version of the paper.

If $n$ elements are hashed into $n^{1+\eps}$ bins by a truly random
hash function, the maximum load of any bin is $O(1)$ with high
probability. First we show that simple
tabulation preserves this guarantee. Building on this,
we shows that the load of any fixed bin obeys
Chernoff bounds. Finally we show that the Chernoff bound holds
even for a bin chosen as a function of the query hash code, $h(q)$.

As stated in the introduction, the number of bins is always understood to be a
power of two. This is because our hash values are xor'ed bit
strings. If we want different numbers of bins we could view the hash
values as fractions in the unit interval and divide the unit interval
into subintervals.  Translating our results to this setting is
standard.

\ptsubsection{Hashing into Many Bins}
The notion of peeling lies at the heart of most work in tabulation
hashing. If a key from a set of keys contains one position-character
that doesn't appear in the rest of the set, its hash code will be
independent of the rest. Then, it can be ``peeled'' from the set, as
its behavior matches that with truly random hashing. More formally, we
say a set $T$ of keys is peelable if we can arrange the keys of $T$ in
some order, such that each key contains a position-character that
doesn't appear among the previous keys in the order.

\begin{lemma} \label{lem:manybins}
Suppose we hash $n \le m^{1-\eps}$ keys into $m$ bins, for some
constant $\eps>0$. For any constant $\gamma$, all bins get less than 
$d = \min \left\{ \left((1 + \gamma)/\eps\right)^c, 2^{(1 +
  \gamma) /\eps}\right\}$
keys with probability $\ge 1-m^{-\gamma}$.
\end{lemma}

\begin{proof}

We will show that among any $d$ elements, one can find a
\emph{peelable} subset of size $t \ge \max \{ d^{1/c}, \lg d
\}$. Then, a necessary condition for the maximum load of a bin to be
at least $d$ is that some bin contain $t$ peelable elements. There are
at most $\binom{n}{t} < n^t$ such sets. Since the hash codes of a
peelable set are independent, the probability that a fixed set lands
into a common bin is $1\big/ m^{t-1}$. Thus, an upper bound on the
probability that the maximum load is $d$ can be obtained: $n^t /
m^{t-1} = m^{(1-\eps)t}/m^{t-1}=m^{1 - \eps t}$.  To obtain failure
probability $m^{-\gamma}$, set $t = (1 + \gamma)/\eps$.

It remains to show that any set $T$ of $|T|=d$ keys contains a large
peelable subset. Since $T \subset \pi(T,1) \times \cdots \times
\pi(T,c)$, it follows that there exists $i\in [c]$ with $|\pi(T,i)|
\ge d^{1/c}$. Pick some element from $T$ for every character value in
$\pi(S,i)$; this is a peelable set of $t = d^{1/c}$ elements.

To prove $t \ge \log_2 d$, we proceed iteratively. Consider the
coordinate giving the largest projection, $j = \arg\max_i
|\pi(T,i)|$. As long as $|T| \ge 2$, $|\pi(T,j)| \ge 2$. Let $\alpha$
be the most popular value in $T$ for the $j$-th character, and let
$T^\star$ contain only elements with $\alpha$ on the $j$-th
coordinate. We have $|T^\star| \ge |T| / |\pi(T,j)|$. In the peelable
subset, we keep one element for every value in $\pi(T,j) \setminus
\{\alpha\}$, and then recurse in $T^\star$ to obtain more elements. In
each recursion step, we obtain $k \ge 1$ elements, at the cost of
decreasing $\log_2 |T|$ by $\log_2 (k+1)$. Thus, we obtain at least
$\log_2 d$ elements overall.
\end{proof}

We note that, when the subset of keys of interest forms a
combinatorial cube, the probabilistic analysis in the proof is sharp
up to constant factors. In other words, the exponential dependence on
$c$ and $\gamma$ is inherent.

\ptsubsection{Chernoff Bounds for a Fixed Bin}  \label{sec:fix-chernoff}
We study the number of keys ending up in a prespecified bin $B$. The
analysis will define a total ordering $\prec$ on the space of
position-characters, $[c]\times \Sigma$. Then we will analyze the
random process by fixing hash values of position-characters
$h(\alpha)$ in the order $\prec$. The hash value of a key $x\in S$
becomes known when the position-character $\max_\prec x$ is fixed. For
$\alpha \in [c]\times \Sigma$, we define the \emph{group} $G_\alpha =
\{x\in S \mid \alpha = \max_\prec x \}$, the set of keys for whom
$\alpha$ is the last position-character to be fixed. 

The intuition is that the contribution of each group $G_\alpha$ to the
bin $B$ is a random variable independent of the previous $G_\beta$'s,
since the elements $G_\alpha$ are shifted by a new hash code
$h(\alpha)$. Thus, if we can bound the contribution of $G_\alpha$ by a
constant, we can apply Chernoff bounds. 

\begin{lemma}\label{lem:group-size} 
There is an ordering $\prec$ such that the maximal group size
is $\max_\alpha |G_\alpha| \le n^{1-1/c}$.
\end{lemma}

\begin{proof} 
We start with $S$ being the set of all keys, and reduce $S$
iteratively, by picking a position-character $\alpha$ as next in the
order, and removing keys $G_\alpha$ from $S$. At each point in time,
we pick the position-character $\alpha$ that would minimize
$|G_\alpha|$. Note that, if we pick some $\alpha$ as next in the
order, $G_\alpha$ will be the set of keys $x \in S$ which contain
$\alpha$ and contain no other character that hasn't been fixed:
$(\forall) \beta \in x \setminus \{\alpha\}, \beta \prec \alpha$.

We have to prove is that, as long as $S\neq\emptyset$, there exists
$\alpha$ with $|G_\alpha| \le |S|^{1-1/c}$. If some position $i$ has
$|\pi(S,i)|> |S|^{1/c}$, there must be some character $\alpha$ on
position $i$ which appears in less than $|S|^{1-1/c}$ keys; thus
$|G_\alpha| \le S^{1-1/c}$. Otherwise, $\pi(S,i)\le |S|^{1/c}$ for all
$i$. Then if we pick an arbitrary character $\alpha$ on some position
$i$, have $|G_\alpha|\le \prod_{j\ne i} |\pi(S,j)| \le
(|S|^{1/c})^{c-1}=|S|^{1-1/c}$.
\end{proof}

From now on assume the ordering $\prec$ has been fixed as in the
lemma. This ordering partitions $S$ into at most $n$ non-empty groups,
each containing at most $n^{1-1/c}$ keys. We say a group $G_\alpha$ is
{\em $d$-bounded\/} if no bin contains more than $d$ keys from
$G_\alpha$.

\begin{lemma}\label{lem:d-bounded} 
Assume the number of bins is $m\ge n^{1-1/(2c)}$. For any constant
$\gamma$, with
probability $\ge 1-m^{-\gamma}$, all groups are $d$-bounded where
\begin{equation*}
d = \min \left\{(2c(3+\gamma)^c, 2^{2c(3+\gamma)}\right\}
\end{equation*}
\end{lemma}

\begin{proof}
Since $|G_\alpha| \le n^{1-1/c}\le m^{1-1/(2c)}$, 
by
Lemma~\ref{lem:manybins}, we get that there are at
most $d$ keys from $G_\alpha$ in any bin with probability 
$1-m^{-(2+\gamma)}\geq 1-m^{-\gamma}/n$. The conclusion follows by union bound over
the $\le n$ groups.
\end{proof}

Henceforth, we assume that $\gamma$ and $d$ are fixed as in Lemma \ref{lem:d-bounded}.
Chernoff bounds (see \cite[Theorem 4.1]{motwani95book}) consider
independent random variables $X_1, X_2, \dots \in [0,d]$. 
Let
$X=\sum_i X_i$, $\mu=\E[X]$, and $\delta>0$, the bounds are:

\begin{equation}
\begin{split}
\Pr[X\ge (1+\delta)\mu] &\le \left(
   \frac{e^\delta}{(1+\delta)^{(1+\delta)}} \right)^{\mu/d} \\
\Pr[X\le (1-\delta)\mu] &\le \left(
   \frac{e^{-\delta}} {(1-\delta)^{(1-\delta)}} \right)^{\mu/d}
\end{split}
\label{eq:both-chernoff}
\end{equation}

Let $X_\alpha$ be the number of elements from $G_\alpha$ landing in
the bin $B$. We are quite close to applying Chernoff bounds to the
sequence $X_\alpha$, which would imply the desired concentration
around $\mu = \frac{n}{m}$. Two technical problems remain: 
$X_\alpha$'s are not $d$-bounded in the worst case, and they are not
independent.

To address the first problem, we define the sequence of random
variables $\hat{X}_\alpha$ as follows: if $G_\alpha$ is $d$-bounded,
let $\hat{X}_\alpha = X_\alpha$; otherwise $\hat{X}_\alpha =
|G_\alpha| / m$ is a constant. Observe that $\sum_\alpha
\hat{X}_\alpha$ coincides with $\sum_\alpha X_\alpha$ if all groups
are $d$-bounded, which happens with probability $1-m^{-\gamma}$. Thus
a probabilistic bound on $\sum_\alpha \hat{X}_\alpha$ is a bound on
$\sum_\alpha X_\alpha$ up to an additive $m^{-\gamma}$ in the
probability.

Finally, the $\hat{X}_\alpha$ variables are not independent: earlier
position-character dictate how keys cluster in a later
group. Fortunately \req{eq:both-chernoff} holds even if the distribution of each $X_i$
is a function of $X_1, \dots, X_{i-1}$, 
as long as the mean $\E[X_i \mid X_1,
  \dots, X_{i-1}]$ is a fixed constant $\mu_i$ independent of 
$X_1,...,X_{i-1}$. A formal proof will be given
in Appendix \ref{sec:chernoff-appendix}.
We claim that our  means are fixed this way:
regardless of the hash codes for $\beta < \alpha$, we will argue
that $\E[\hat{X}_\alpha] = \mu_\alpha=|G_\alpha|/m$.

Observe that whether or not $G_\alpha$ is $d$-bounded is determined
before $h(\alpha)$ is fixed in the order $\prec$. Indeed,
$\alpha$ is the last position-character to be fixed for any key in
$G_\alpha$, so the hash codes of all keys in $G_\alpha$ have been
fixed up to an xor with $h(\alpha)$. This final shift by $h(\alpha)$
is common to all the keys, so it cannot change whether or not two
elements land together in a bin. Therefore, the choice of $h(\alpha)$ does 
not change if $G_\alpha$ is $d$-bounded.

After fixing all hash codes $\beta \prec \alpha$, we decide if
$G_\alpha$ is $d$-bounded. If not, we set $\hat{X}_\alpha=
|G_\alpha|/m$. Otherwise $\hat{X}_\alpha = X_\alpha$ is the number of
elements we get in $B$ when fixing $h(\alpha)$, and $h(\alpha)$ is a
uniform random variable sending each element to $B$ with probability
$1/m$. Therefore $\E[\hat{X}_\alpha]= |G_\alpha|/m$.  This completes
the proof that the number of keys in bin $B$ obeys Chernoff bounds from
\req{eq:both-chernoff}, which immediately imply \req{eq:normal} and \req{eq:Poisson} in Theorem~\ref{thm:chernoff}.

\ptsubsection{The Load of a Query-Dependent Bin} 
When we are dealing with a special key $q$ (a query), we may be
interested in the load of a bin $B_q$, chosen as a function of the
query's hash code, $h(q)$. We show that the above analysis also works
for the size of $B_q$, up to small constants. The critical change is
to insist that the query position-characters come first in our
ordering $\prec$:

\begin{lemma} \label{lem:group-size-q}  
There is an ordering $\prec$ placing the characters of $q$ first, in
which the maximal group size is $2\cdot n^{1-1/c}$.
\end{lemma}

\begin{proof} 
After placing the characters of $q$ at the beginning of the order, we
use the same iterative construction as in
Lemma~\ref{lem:group-size}. Each time we select the position-character
$\alpha$ minimizing $|G_\alpha|$, place $\alpha$ next in the order
$\prec$, and remove $G_\alpha$ from $S$. It suffices to prove that, as
long as $S\ne \emptyset$, there exists a position-character $\alpha
\notin q$ with $|G_\alpha| \le 2 \cdot |S|^{1-1/c}$. Suppose in some
position $i$, $|\pi(S,i)| > |S|^{1/c}$. Even if we exclude the query
character $q_i$, there must be some character $\alpha$ on position $i$
that appears in at most $|S| / (|\pi(S,i)|-1)$ keys. Since $S\ne
\emptyset$, $|S|^{1/c}>1$, so $|\pi(S,i)| \ge 2$. This means
$|\pi(S,i)| - 1 \ge |S|^{1/c}/2$, so $\alpha$ appears in at most $2
|S|^{1-1/c}$ keys. Otherwise, we have $\pi(S,i) \le |S|^{1/c}$ for all
$i$. Then, for any character $\alpha$ on position $i$, we have
$|G_\alpha| \le \prod_{j\ne i}|\pi(S,j)|\le |S|^{1-1/c}$.
\end{proof}

The lemma guarantees that the first nonempty group contains the query
alone, and all later groups have random shifts that are independent of
the query hash code.  We lost a factor two on the group size, which
has no effect on our asymptotic analysis. In particular, all groups
are $d$-bounded w.h.p. Letting $X_\alpha$ be the contribution of
$G_\alpha$ to bin $B_q$, we see that the distribution of $X_\alpha$ is
determined by the hash codes fixed previously (including the hash code
of $q$, fixing the choice of the bin $B_q$). But $\E[X_\alpha] =
|G_\alpha|/m$ holds irrespective of the previous choices. Thus, 
Chernoff bounds continue to apply to the size of $B_q$.
This completes the proof of \req{eq:normal} and \req{eq:Poisson} in Theorem~\ref{thm:chernoff}.

In Theorem~\ref{thm:chernoff} we limited ourselves to polynomially
small error bounds $m^{-\gamma}$ for constant $\gamma$. However,
we could also consider a super constant $\gamma=\omega(1)$ using
the formula for $d$ in Lemma \ref{lem:d-bounded}. For
the strongest error bounds, we would balance $m^{-\gamma}$
with the Chernoff bounds from (\ref{eq:both-chernoff}). Such
balanced error bounds would be messy, and we found it
more appealing to elucidate the standard Chernoff-style behavior
when dealing with polynomially small errors.

\subsection{Few bins}
We will now settle Theorem \ref{thm:chernoff} \req{eq:few-bins},
proving some high probability bounds for the concentration
with $m\leq n$ bins. As stated in \req{eq:few-bins}, we  will show, w.h.p., that the
number of keys in each bin is
\[n/m\pm O(\sqrt{n/m}\log^c n).\]
Consider any subset $S$ of $s\leq n$ keys that only vary in $b$
characters.  Generalizing \req{eq:few-bins}, we will show for 
any $L\geq 32\log n$, that with
probability $1-\exp(-\Omega(L))$, the keys in
$S$ get distributed with 
\begin{equation}\label{eq:ind-L}
\left\{\begin{array}{ll}
s/m\pm \sqrt{s/m}\,L^b &\mbox{if }s\geq mL^b/2\\
\leq L^{b}&\mbox{if }s\leq mL^b/2\\
\end{array}\right.
\end{equation}
keys in each of the $m$ bins. This is trivial for $m=1$, so
we can assume $m\geq 2$. The proof is by induction on $(b,s)$.
First we will prove that each inductive step fails with small
probability. Later we will conclude that the combined
failure probability for the whole induction is small.

For the base case of the induction, if $s\leq L^{b}$, the
result is trivial since it holds even if some bin gets
all the keys from $S$. 
This case includes if we have no characters
to vary, that is, when $s=1$ and $b=0$.
We may therefore assume that $s>L^{b}$, and $b>0$. 
The characters positions where $S$ do not vary will
only shuffle the bins, but not affect which keys from $S$ go
together, so we can ignore them when giving bounds for
the sizes of all bins.

Considering the varying characters in $S$, we apply
the ordering from Section \ref{sec:fix-chernoff} leading
to a grouping of $S$. By Lemma \ref{lem:group-size}, there is an 
ordering $\prec$ such that the  maximal group size
is $\max_\alpha |G_\alpha| \le {s}^{1-1/b}$. In particular,
$\max_\alpha|G_\alpha|<s/L$.

First, assume that $s\leq mL^b/2$. Each group has one less
free character, so by induction, each group
has at most $L^{b-1}$ keys in each bin, that is, each
group is $L^{b-1}$-bounded. Now as in Section~\ref{sec:fix-chernoff},
for any fixed bin, we can apply the Chernoff upper-bound 
from (\ref{eq:both-chernoff}) with $d=L^{b-1}$. We have $\mu=s/m\leq L^b/2$, and 
we want to bound the probability of getting a bin of size at least 
$x=L^b\geq 2\mu$. For an upper bound, we use $\mu'=x/2\geq\mu$ and $\delta'=1$,
and get a probability bound of 
\[\left( \frac{e^{\delta'}}{(1+\delta')^{(1+\delta')}}\right)^{\mu'/d}=
(e/4)^{\mu'/d}\leq (e/4)^{L/2}.\]
With the union bound,
the probability that any bin has more than $L^b$ keys is bounded by
$m (e/4)^{L/2}$.

\paragraph{Partitioning many keys}
Next, we consider the more interesting case where $s\geq mL^b/2$.  As
stated in \req{eq:ind-L}, we want to limit the probability that the
contribution $S$ to any bin deviates by more than $\sqrt {s/m}L^b$ from
the mean $s/m$.
We partition the groups into levels $i$ based on their sizes. On level
$0$ we have the groups of size up to $mL^{b-1}/2$. On level 
$i>0$, we have the groups of size
between $t_i=mL^{b-1}2^{i-2}$ and $2t_i$. For each $i$, we let
$S_i$ denote the union of the level $i$ groups. We are going to handle each
level $i$ separately, providing a high probability bound on how
much the contribution of $S_i$ to a given bin can deviate from the
mean. Adding the deviations from all levels, we bound the total
deviation in the contribution from $S$ to this bin.  
Let $s_i$ be the number of keys in
$S_i$ and define
\begin{equation}\label{eq:size-classes}
\Delta_i=\sqrt{s_i/m}\,L^{b-1/2}.
\end{equation}
For level $i>0$, we will use $\Delta_i$ as our deviation bound, while
we for $i=0$, will use $\bar\Delta_0=\max\{\Delta_0,L^b\}$. 

\paragraph{The total deviation}
We will now show that the above level deviation bounds
provide the desired total deviation bound of $\sqrt{s/m} L^b$ from
\req{eq:ind-L}. Summing over the levels, the total deviation is bounded by 
$L^b+\sum_i\sqrt{s_i/m}\,L^{b-1/2}$. 
To bound the sum, we first consider the smaller terms where $s_i<s/\log n$. 
Then $\sqrt{s_i/m}L^{b-1/2}\leq \sqrt{s/m}L^b/\sqrt{L\log n}$.
We have at most $\log n$ values of $i$, so these
smaller terms sum to at most $\sqrt{s/m}L^b\sqrt{(\log n)/L}$.

Next we consider the larger terms where $s_i\geq s/\log n$. Each
such term can be bounded as 
\begin{align*}
\sqrt{s_i/m}\,L^{b-1/2}&=\left((s_i/m)/\sqrt{s_i/m}\right)L^{b-1/2} \\
&\leq \left((s_i/m)/\sqrt{s/m}\right)L^b\sqrt{\log n/L}.
\end{align*}
The sum of the larger terms is therefore also bounded by 
$\sqrt{s/m}L^b\,\sqrt{(\log n)/L}$.
Thus the total deviation is bounded by 
$L^b+\sqrt{s/m}\, L^b2\sqrt{(\log n)/L}$. Assuming 
$L\geq 9\log n$, we have $2\sqrt{(\log n)/L}\leq 2/3$.
Moreover, with $n\geq 4$ and $b\geq 1$, we have $s/m>L^b/2\geq 9$.
It follows that $L^b+2\sqrt{(\log n)/L}\sqrt{s/m} L^b\leq
\sqrt{s/m} L^b$, as desired.

\paragraph{Deviation from small groups} We now consider
the contribution to our bin from the small groups in $S_0$.
These groups have size most $(mL^b/2)^{1-1/b}\leq
mL^{b-1}/2$,
and $b-1$ free character positions, so inductively from \req{eq:ind-L}, 
each group
contributes at most $d_0\leq L^{b-1}$ to each bin. We want
to bound the probability that the deviation from the
mean $\mu_0=s_0/m$ exceeds $\bar\Delta_0=\max\{\Delta_0,L^b\}$.

Suppose $\mu_0\leq \bar\Delta_0$. 
For a Chernoff upper bound, we use $\mu'=
\bar \Delta_0\geq\mu_0$ and $\delta'=1$,
and get a probability bound of
\[\left(e^{\delta'}/(1+\delta')^{(1+\delta')}\right)^{\mu'/d_0}=
(e/4)^{\mu'/d_0}\leq (e/4)^{L}.\]

On the other hand, if $\mu_0\geq \Delta_0$, we have a relative
deviation of $\delta_0=\Delta_0/\mu_0=\sqrt{m/s_0} L^{b-1/2}\leq 1$.
The probability of this deviation for any fixed bin is
bounded by 
\[\left(e^{\delta_0}/(1+\delta_0)^{(1+\delta_0)}\right)^{\mu_0/d_0}
\leq \exp(-(\mu_0/d_0)\delta_0^2/3)=\exp(-L^b/3)\leq \exp(-L/3).\]

\paragraph{Larger groups}
To deal with a larger group level $i>1$, we will use
a standard symmetric Chernoff bound, which is easily derived from
the negative version of \req{eq:both-chernoff}.
We consider independent $n$ random variables $X_1,....,X_n\in [-d,d]$,
each with mean zero. Let $X=\sum_i X_i$. For
any $\delta>0$, 
\begin{equation}
\Pr[|X|\ge \delta dn] \le 2\exp(-n\delta^2/4)
\label{eq:symmetric}
\end{equation}
As we did for \req{eq:both-chernoff} in Section~\ref{sec:fix-chernoff}, we note that \req{eq:symmetric} also
holds when $X_i$ depends on the previous $X_j$, $j<i$, as long
as $|X_i|\leq d$ and $\E[X_i]=0$.
Back to our problem, let $s_i$ be the total size. Each group $G$
has size at least $t_i=mL^{b-1}2^{i-2}$, so we have at most $n_i=s_i/t_i$ 
groups. The group $G$ has only $b-1$ varying characters and 
$t_i\geq t_1=mL^{b-1}/2$, so inductively from \req{eq:ind-L}, the contribution
of $G$ to any bin deviates by at most 
$d_i=\sqrt{|G|/m}L^{b-1}<\sqrt{2t_i/m}L^{b-1}$ from the mean $|G|/m$. 
We let $X_G$ denote the contribution of $G$ to our bin minus
the mean $|G|/m$. Thus, regardless of the distribution of previous groups,
we
have $\E[X_G]=0$ and $|X_G|\leq d_i$.
We want to bound the probability that $|\sum_G X|\geq \Delta_i$.
We therefore apply (\ref{eq:symmetric}) with 
\[\delta_i =\Delta_i/(d_in_i)
=\sqrt{s_i/m}\,L^{b-1/2}/\left(\sqrt{2t_i/m}L^{b-1} s_i/t_i\right) = \sqrt{t_iL/(2s_i)}.\]
The probability that the contribution to our bin deviates by more than
$\Delta_i$ is therefore bounded by 
\[2\exp(-n_i\delta_i^2/4)=2\exp(-s_i/t_i\cdot\sqrt{t_iL/(2s_i)}^2/4)=2\exp(-L/8).\]
Conveniently, this dominates the error probabilities
of $(e/4)^{L}$ and $\exp(-L/3)$ from level $0$.
There are less $\log n$ levels, so 
by the union bound, the probability of a too large deviation from any level to
any bin is bounded by $m(\log n)2\exp(-L/8)$.

\paragraph{Error probability for the whole induction}
Above we proved that any particular inductive step fails with
probability at most $m(\log n)2\exp(-L/8)$. We want to conclude that
the probability of any failure in the whole induction is bounded by
$nm\exp(-L/8)$.

First we note that the all the parameters of the inductive steps
are determined deterministically. More precisely, the inductive
step is defined via the deterministic grouping from 
Lemma \ref{lem:group-size}. This grouping corresponds to
a certain deterministic ordering of the position characters, and
we use this ordering to analyze the failure probability of the
inductive step. However, there is no relation between the ordering used to
analyze different inductive steps.
Thus, we are dealing with
a recursive deterministic partitioning. Each partitioning results in
groups that are at least $L$ times smaller, so the recursion corresponds to a 
tree with degrees at least $L$. At the bottom we have base cases, each
containing at least one key. The internal nodes correspond to
inductive steps, so we have less than $2n/L$ of these. If $L\geq 4\log n$,
we conclude that the combined failure probability is at
most $2n/L\,m(\log n)2\exp(-L/8)\leq nm\exp(-L/8)$. With $L\geq 32\log n$,
we get that the overall failure probability is bounded by $\exp(-L/64)$.
This completes our proof that \req{eq:ind-L} is satisfied with
high probability, hence the proof of 
Theorem \ref{thm:chernoff} \req{eq:few-bins}.

\section{Linear Probing and the Concentration in Arbitrary Intervals}   \label{sec:linprobe}

We consider linear probing using simple tabulation hashing to store a
set $S$ of $n$ keys in an array of size $m$ (as in the rest of our
analyses, $m$ is a power of two). Let $\alpha = 1-\eps= \frac{n}{m}$
be the fill.  We will argue that the performance with simple
tabulation is within constant factors of the performance with a truly
random function, both in the regime $\eps \geq 1/2$ (high fill) and
$\alpha \leq 1/2$ (low fill). With high fill, the expected number of
probes when we insert a new key is $O(1/\eps^2)$ and with low fill, it
is $1+O(\alpha)$.

Pagh et al.~\cite{pagh07linprobe} presented an analysis of
linear probing with 5-independent hashing using 4\th moment
bounds. They got a bound of $O(1/\eps^{13/6})$ on the expected
number of probes. We feel that our analysis, which is centered around dyadic
intervals, is simpler, tighter, and more generic. Recall that a dyadic interval,
is an interval of the form $[j2^i,(j+1)2^i)$ for integers $i$ and $j$. In fact, as we shall
see later in Section \ref{sec:fourth-lp}, our analysis also
leads to an optimal $O(1/\eps^2)$ for 5-independent hashing,
settling an open problem from ~\cite{pagh07linprobe}. 
However, with simple tabulation, we get much stronger concentration
than with 5-independent hashing, e.g., constant variance with
constant $\eps$ whereas the variance is only known to be
$O(\log n)$ with 5-independent hashing.

When studying the complexity of linear probing, the basic measure is
the length $R=R(q,S)$ of the longest run of filled positions starting
from $h(q)$, that is, positions $h(q),...,h(q)+\ell-1$ are filled with
keys from $S$ while $h(q)+R$ is empty. This is the case if and only if
$R$ is the largest number such there is an interval $I$ which contains
$h(q)$ and $h(q)+R-1$ and such that $I$ is full in the sense that at
least $|I|$ keys from $S$ hash to $I$. In our analysis, we assume that
$q$ is not in the set. An insert or unsuccessful search with $q$ will
consider exactly $R+1$ positions. A successful search for $q$ will
consider at most $R(q,S\setminus\{q\})+1$ positions. For
deletions, the cost is $R(q,S)+1$ but where $q\in S$. For now
we assume $q\not\in S$, but we shall return to the case $q\in S$ in Section \ref{sec:lp-cost}.

Aiming for upper bounds on $R(q,S)$, it is simpler to
study the symmetric length $L(q,S)$ of the longest filled interval containing $h(q)$. 
Trivially $R(q,S)\leq L(q,S)$. We have $n=|S|$ keys hashed
into $m$ positions. We defined the fill $\alpha=n/m$ and $\eps=(1-\alpha)$.
The following theorem considers the case of
general relative deviations $\delta$. To bound $\Pr[L(q,S)\geq\ell]$, we can
apply it with $p=h(q)$ and $\delta=\eps$ or $(1+\delta)=1/\alpha$.
\begin{theorem}\label{thm:linprobe}
  Consider hashing a set of $n$
  keys into $\{0,...,m-1\}$ using simple tabulation (so $m$ is a power
  of two). Define the fill $\alpha=n/m$.  Let $p$ be any point which
  may or may not be a function of the hash value of specific query key
  not in the set. Let $\cD_{\ell,\delta,p}$ be the event that there
  exists an interval $I$ containing $p$ and of length at least $\ell$
  such that the number of keys
  $X_I$ in $I$ deviates at least $\delta$ from the mean, that is,
  $|X_I-\alpha|I||\geq \delta\alpha|I|$. Suppose $\alpha\ell\leq
  n^{1/(3c)}$, or equivalently, $m/\ell\geq n^{1-1/(3c)}$. Then for any constant $\gamma$,
\begin{equation}\label{eq:lp}
\Pr[\cD_{\ell,\delta,p}]\leq \left\{\begin{array}{ll}
2e^{-\Omega(\alpha\ell\delta^2)} + (\ell/m)^{\gamma}&\mbox{if }\delta\leq 1\\
(1+\delta)^{-\Omega((1+\delta)\alpha\ell)} + (\ell/m)^{\gamma}&\mbox{if }\delta=\Omega(1)
\end{array}\right.
\end{equation}
Moreover, with probability $1-n^{-\gamma}$, for every interval $I$, if $\alpha|I|\geq 1$,
the number of keys in $I$ is 
\begin{equation}\label{eq:few-bins-lp}
\alpha|I|\pm O\left(\sqrt{\alpha|I|}\log^c n\right).
\end{equation}
\end{theorem}
Theorem~\ref{thm:linprobe} is a very strong generalization
of Theorem~\ref{thm:chernoff}. A bin from Theorem~\ref{thm:linprobe}  corresponds to a specific dyadic interval of length
$\ell=2^i$ (using $m'=m/2^i$ in Theorem~\ref{thm:linprobe}). In Theorem~\ref{thm:linprobe} we consider every interval of length
at least $\ell$ which contains a specific point,
yet we get the same deviation bound modulo a change in the constants hidden in
the $\Omega$-notation.

To prove the bound on $\cD_{\ell,\delta,p}$, we first consider the
weaker event $\cC_{i,\delta,p}$ for integer $i$ that there exists an
interval $I\ni p$, $2^i\leq |I|<2^{i+1}$,  such
that the relative deviation $\delta$ in the number of keys $X_I$ is at
least $\delta$. As a start we will prove that the
bound from \req{eq:lp} holds for $\cC_{i,\delta,p}$. Essentially Theorem~\ref{thm:linprobe} 
will follow because the probability bounds decrease exponentially in
$i$.  

When bounding the probability of $\cC_{i,\delta,p}$, we will consider any
$i$ such that $\alpha2^i\leq n^{1/(2c)}$ whereas we in Theorem
\ref{thm:linprobe} only considered $\alpha\ell\leq n^{1/(3c)}$.  The
constraint $\alpha2^i\leq n^{1/(2c)}$ matches that in 
Theorem~\ref{thm:chernoff} with $m'=m/2^i$. In Theorem~\ref{thm:chernoff} 
we required $m'\geq n^{1-1/(2c)}\iff n/m'=\alpha 2^i\leq n^{1/(2c)}$.

Our proof is based on 
decompositions of intervals into dyadic intervals. To simplify
the terminology and avoid confusion
between intervals and dyadic intervals, we let a {\em bin on level $j$},
or for short, a {\em $j$-bin}, denote a dyadic interval of length
$2^j$. The expected number of keys in a $j$-bin is $\mu_j=\alpha 2^j$.
The $j$-bins correspond to the bins in Theorem \ref{thm:chernoff} with $m'=m/2^j$.
For any $j\leq i$, consider the $j$-bin containing $p$, and the
$2^{i+1-j}$ $j$-bins on either side. We say that these $2^{i+2-j}+1$
consecutive $j$-bins are {\em relevant\/} to $\cC_{i,\delta,p}$ noting
that they cover any $I\ni p$, $|I|\leq 2^{i+1}$.

\paragraph{\underline{$\delta=\Omega(1)$}}
To handle $\delta=\Omega(1)$, we will use the following
combinatorial claim that holds for any $\delta$.
\begin{claim}\label{claim:large-delta}
Let $j$ be maximal such that $2^j < \frac{\delta}{1+\delta/2} 2^{i-2}$. If
$\cC_{i,\delta,p}$ happens, then
one of the relevant $j$-bins contains more than 
$(1+\frac{\delta}{2})\alpha 2^j$ keys.
\end{claim}

\begin{proof}
  Assume that all the relevant $j$-bins have relative deviation
  at most $\frac{\delta}{2}$. Let $I$ be an interval witnessing
  $\cC_{i,\delta,p}$, that is, $p\in I$, $2^i\leq |I|<2^{i+1}$,  and
  the number of keys in $I$ deviates by $\delta\alpha|I|$ from
  the mean $\alpha|I|$.
  The interval $I$ contains some number of
  the $j$-bins, and properly intersects at most two in 
  the ends. The relative deviation within $I$ is
  $\delta$, but the $j$-bins have only half this
  relative deviation. This means that the $j$-bins
  contained in $I$ can contribute at most half the deviation. The
  remaining $\frac{\delta}2\alpha|I|$ has to
  come from the two $j$-bins intersected in the ends. 
  Those could contribute all or none of their keys to the deviation
  (e.g.~all keys are on the last/first position of the
  interval). However, together they have at most 
  $2(1+\frac{\delta}{2})\alpha2^j<\delta\alpha2^{i-1}\leq \frac{\delta}{2}\alpha|I|$ keys.
\end{proof}
Let $\delta=\Omega(1)$ and define $j$ as in Claim \ref{claim:large-delta}.
Then $j=i-\Omega(1)$. To bound the probability of $\cC_{i,\delta,p}$
it suffices to bound the probability that none of the $2^{i+2-j}+1=O(1)$
relevant $j$-bins has relative deviation beyond $\delta'=\delta/2$.
We will apply Theorem~\ref{thm:chernoff} \req{eq:Poisson} 
with $m'=m/2^j$ and $\mu'=\alpha2^j$ to each of these $j$-bins. Checking the
conditions of Theorem~\ref{thm:chernoff}, we note that
the $k$'th relevant $j$-bin can specified as a function of $p$ which again 
may be a function of the hash of 
the query. Also, as noted above, $m'>m/2^i\geq n/(\alpha2^i)\geq n^{1-1/(2c)}$. 
From \req{eq:Poisson} we get that
\[\Pr[\cC_{i,\delta,p}]=O(1)\left((1+\delta/2)^{-\Omega((1+\delta/2)\alpha 2^j)} 
+ (2^j/m)^{\gamma}\right)
=(1+\delta)^{-\Omega((1+\delta)\alpha 2^i)}+O\left((2^i/m)^{\gamma}\right).\]

\paragraph{\underline{$\delta\leq 1$}}
We now consider the case $\delta\leq 1$. In particular, this covers
the case $\delta=o(1)$ which was not covered above. The issue
is that if we apply Claim \ref{claim:large-delta}, we could
get $j=i-\omega(1)$, hence $\omega(1)$ relevant $j$-bins, and
then the applying the union bound would lead to a loss.
To circumvent the problem we will consider a tight decomposition
involving bins on many levels below $i$ but with bigger deviations
on lower levels.
For any level $j\leq i$, we 
say that a $j$-bin is ``dangerous'' for level
$i$ if it has deviation at  least:
\[ \Delta_{j,i} ~=~ \tfrac{\delta\alpha2^i}{24} / 2^{(i-j)/5} 
~=~ \tfrac{\delta\alpha}{24}\cdot 2^{\frac45 i + \frac{1}{5} j}.
\]
\begin{claim}\label{claim:small-delta}
Let $j_0$ be the smallest non-negative integer satisfying $\Delta_{j_0,i} \le
\alpha 2^{j_0}$. If
$\cC_{i,\delta,p}$ happens, then for some $j\in\{j_0,...,i\}$,
there is a relevant $j$-bin which is
dangerous for level $i$.
\end{claim}

\begin{proof}
Witnessing $\cC_{i,\delta,p}$, let $I\ni p$, $2^i\leq |I|<2^{i+1}$,
have at least  $(1+\delta)\alpha |I|$ keys. First we make the
standard dyadic decomposition of $I$ into 
maximal level bins: at most two $j$-bins on each 
level $j=0 \twodots i$. For technical reasons, if $j_0>0$, the
decomposition is ``rounded to level $j_0$''. 
Formally, the decomposition rounded to level $j_0$ is
obtained by discarding all the bins on
levels below $j_0$, and including one $j_0$-bin on both sides 
(each covering the discarded bins on lower levels). Note that
all the level bins in the decomposition of $I$ are
relevant to $\cC_{i,\delta,p}$.

Assume for a contradiction that no relevant bin on levels
$j_0,...,i$ is dangerous for level $i$. In particular, this
includes all the level bins from our decomposition. We will 
sum their deviations, and show
that $I$ cannot have the required deviations. In case of rounding,
all keys in the two rounding $j_0$-bins can potentially be in or out of 
$I$ (all keys in such
intervals can hash to the beginning/end), contributing at most
$\Delta_{j_0,i}+\alpha 2^{j_0}$ keys to the deviation in $I$. 
By choice of $j_0$, we have $\alpha 2^{j_0-1}<\Delta_{j_0-1,i}$.
It follows that the total
contribution from the rounding bins is at most
\[2(\Delta_{j_0,i}+\alpha 2^{j_0})\leq 2(\Delta_{j_0,i}+2\Delta_{j_0-1,i})
< 6\Delta_{i,i}=\tfrac{\delta\alpha2^i}{4}.\]
The other bins from the decomposition are internal to $I$. This
includes discarded ones in case of rounding. A $j$-bin
contributes at most $\Delta_{j,i}$ to
the deviation in $I$, and there are at most 2 such $j$-bins for each
$j$. The combined internal contribution is therefore bounded
by
\begin{equation}\label{eq:internal-contr}
2\sum_{j=0}^i \Delta_{j,i}
=  2\sum_{j=0}^i\left( \frac{\delta\alpha 2^i}{24} / 2^{(i-j)/5}\right)
=\frac{\delta\alpha 2^i}{12} \sum_{h=0}^i 1/2^{h/5}< 
\frac{\delta\alpha 2^i}{12}/(1-2^{-1/5})<\frac{7.73\delta\alpha 2^i}{12}
\end{equation}
The total deviation is thus at $(\frac14+\frac{7.73}{12})\delta\alpha 2^i$,
contradicting that $I$ had deviation $\delta\alpha 2^i$.
\end{proof}

For each $j=j_0,...,i$, we bound the probability that there exists a
relevant $j$-bin which is dangerous for level $i$. There are $2^{2+i-j}+1$
such intervals. We have mean
$\mu_j=\alpha2^j$ and
deviation $\delta_{i,j}\mu_j = \Delta_{i,j}=\Theta(\delta 2^{\frac{4}{5} i
  +\frac{1}{5} j})$. Therefore $\delta_{i,j} = \Delta_{i,j} / \mu_j =
\Theta(\delta 2^{\frac{4}{5}(i-j)})$. Note that $\delta_{i,j} < 1$ by choice
of $j_0$. We can therefore apply
\eqref{eq:normal} from Theorem \ref{thm:chernoff}.  Hence, for any constant $\gamma$, the
probability that there exists a relevant $j$-bin which is dangerous for $i$ is
bounded by
\begin{eqnarray*}
(2^{2+i-j}+1)\left(2e^{-\Omega(\mu_j\delta_{i,j}^2)} + (2^j/m)^{\gamma}\right)
&\le& O(2^{i-j})\left(e^{-\Omega\left(\alpha 2^j\delta 2^{4(i-j)/5})^2\right)} + (2^j/m)^{\gamma}\right)
\\
&=& O\left(2^{i-j}e^{-\Omega\left(\alpha2^i\delta^2 2^{\frac35(i-j)}\right)} + (2^i/m)^{\gamma}/2^{(i-j)(\gamma-1)}\right).
\end{eqnarray*}
To bound the probability of $\cC_{i,\delta,p}$,  we sum the above bound for 
$j=j_0,...,i$. We will argue that the $j=i$ dominates. If $\gamma>2$, then
clearly this is the case for the term
$O\left((2^i/m)^{\gamma}/2^{(i-j)(\gamma-1)}\right)$. It
remains to argue that 
\begin{equation}\label{eq:lp-sum}
\sum_{h=0}^{i-j_0} O\left(2^he^{-\Omega(\alpha2^i\delta^2 2^{\frac35 h})}\right)=
O(e^{-\Omega(\alpha 2^i\delta^2)}).
\end{equation}
At first this may seem obvious since the increase with $h$ is exponential while
the decrease is doubly
exponential. The statement is, however, not true 
if $\alpha2^i\delta^2 =o(1)$, for $e^{-\Omega\left(\alpha2^i\delta^2 2^{\frac35 h}\right)}\approx 1$ 
as long as  $\alpha2^i\delta^2 2^{\frac35 h}=o(1)$.
We need to argue that $\alpha2^i\delta^2 =\Omega(1)$. Then for $h=\omega(1)$,
the bound will decrease super-exponentially in $h$.
Recall that our final goal for $\delta\leq 1$ is to prove that 
$\Pr[\cC_{i,\delta,p}]\leq 2 \exp(-\Omega(\alpha2^i\delta^2))+O\left((2^i/m)^\gamma\right)$. This statement is trivially true if
$\exp(-\Omega(\alpha2^i\delta^2))\geq 1/2$. Thus we may
assume $\exp(-\Omega(\alpha2^i\delta^2))<1/2$ 
and this implies $\alpha2^i\delta^2=\Omega(1)$, as desired. Therefore
the sum in \req{eq:lp-sum} is dominated in by the case $h=(i-j)=0$. Summing
up, for any constant $\gamma>2$ and $\delta\leq 1$, we have proved that
\begin{eqnarray*}
\Pr[\cC_{i,\delta,p}]&=&\sum_{h=0}^{i-j_0}
O\left(2^he^{-\Omega\left(\alpha2^i\delta^2 2^{\frac35h}\right)} + (2^i/m)^{\gamma}/2^{h(\gamma-1)}\right)\\
&=&O\left(e^{-\Omega\left(\alpha2^i\delta^2\right)} + (2^i/m)^{\gamma}\right)\\
&=&2e^{-\Omega\left(\alpha2^i\delta^2 \right)} + O((2^i/m)^{\gamma}).
\end{eqnarray*}
The constraint $\gamma>2$ has no effect, since we get
better bounds with larger $\gamma$ as long as $\gamma$ remains constant. All
together, for $\alpha2^i\leq n^{1/(2c)}$ or equivalently, $m/2^i\geq n^{1-1/(2c)}$, we have proved 
\begin{equation}\label{eq:C}
\Pr[\cC_{i,\delta,p}]\leq \left\{\begin{array}{ll}
2e^{-\Omega(\alpha 2^i\delta^2)} + (2^i/m)^{\gamma}&\mbox{if }\delta\leq 1\\
(1+\delta)^{-\Omega((1+\delta)\alpha 2^i)} + (2^i/m)^{\gamma}&\mbox{if }\delta=\Omega(1)
\end{array}\right.
\end{equation}
We now want to bound $\Pr[\cD_{\ell,\delta,p}]$ as in \req{eq:lp}. For
our asymptotic bound, it suffices to consider cases where $\ell=2^k$
is a power of two. Essentially we will use the trivial
bound $\Pr[\cD_{2^k,\delta,p}]\leq
\sum_{h\geq 0}\Pr[\cC_{k+h,\delta,p}]$. First we want to argue that
the terms 
$e^{-\Omega(\alpha 2^{k+h}\delta^2)}=e^{-\Omega(\alpha 2^k\delta^2)2^h}$, $\delta\leq 1$, and
$(1+\delta)^{-\Omega((1+\delta)\alpha 2^{k+h})}=
(1+\delta)^{-\Omega((1+\delta)\alpha 2^{k})2^{h}}$, $\delta=\Omega(1)$, are dominated by
the case $h=0$.  Both terms are of the form $1/a^{2^h}$ and we
want to show that $a=1+\Omega(1)$. For the case $\delta\leq 1$, we can
use the same trick before: to prove \req{eq:lp} it
suffices to consider $\exp(-\Omega(\alpha2^k\delta^2))<1/2$ which
implies $\alpha 2^{k}\delta^2=\Omega(1)$ and 
$e^{\Omega(\alpha 2^k\delta^2)}=1+\Omega(1)$. When it comes
to $(1+\delta)^{-\Omega((1+\delta)\alpha 2^{k+h})}$, we have $\delta=\Omega(1)$.
Moreover, to get the strongest probability bound on 
$\Pr[\cD_{2^k,\delta,p}]$, we can assume $(1+\delta)\alpha2^k\geq 1$.
More precisely, suppose $(1+\delta)\alpha2^k<1$ and
define $\delta'>\delta$ such that $(1+\delta')\alpha 2^k=1$.
If an interval is non-empty, it has at least 1 key, so
$\cD_{2^k,\delta,p}\iff \cD_{2^k,\delta',p}$, and the probability
bound from \req{eq:lp} is better with the larger $\delta'$.
Thus we can assume
$(1+\delta)^{\Omega((1+\delta)\alpha 2^k)}=1+\Omega(1)$.
We have now established that for any relevant $\delta$, the bound from
\req{eq:C} is of the form 
\[1/a^{2^h}+(2^{k+h}/m)^{\gamma}\mbox{ where }a=1+\Omega(1).\]
As desired the first term is dominated by the smallest $h=0$.
Two issues remain: the second term is dominated by larger $h$ and
\req{eq:C} only applies when $m/2^{k+h}\geq n^{1-1/(2c)}$.
Define $\bar h$ as the smallest value such that $a^{2^{\bar h}}\geq m/2^k$. 
We have $2^{\bar h}=\lceil \log_a (m/2^k)\rceil=O(\log (m/2^k))$ and the
condition for \req{eq:lp} is that $n^{1-1/(3c)}\leq m/2^k$, so
$m/2^{k+\bar h}=(m/2^k)/O(\log (m/2^k))=\tilde \Omega(n^{1-1/(3c)})>
n^{1-1/(2c)}$. We conclude that
\req{eq:C} applies for any $h\leq\bar h$.

To handle $h\geq\bar h$ we consider the more general event $\cA_{i,\delta}$ 
that the contribution to {\em any\/}
interval of length at least $2^{i}$ has relative deviation at least
$\delta$. It is easy to see that 
\begin{equation}\label{eq:A}
\Pr[\cA_{i,\delta}]\leq m/2^i\cdot\Pr[\cC_{i,\delta,p}].
\end{equation}
More precisely, consider the $m/2^i$ points $p$ that are
multiples of $2^i$. Any interval $I$ of length $\geq 2^i$
can be partitioned into intervals $I_j$ such
that $2^i\leq |I_j|<2^{i+1}$ and $p_j=j m/2^i\in I_j$. If $I$ has
relative deviation $\delta$, then so does some $I_j$, and
then $\cC_{i,\delta,p_j}$ is satisfied. Thus \req{eq:A} follows.
With our particular value of $i=k+\bar h$, for any $\gamma>1$, we get
\begin{eqnarray*}
\Pr[\cA_{k+\bar h,\delta}]\ \leq\ m/2^{k+\bar h}\Pr[\cC_{k+\bar h,\delta,p}]
&=&m/2^{k+\bar h}\left(1/a^{2^{\bar h}}+(2^{k+\bar h}/m)^{\gamma}\right)\\
&\leq &m/2^{k+\bar h}\left(2^k/m+(2^{k+\bar h}/m)^{\gamma}\right)\\
&=&(2^{k+\bar h+1}/m)^{\gamma-1}=\tilde\Omega(2^k/m)^{\gamma-1}
\end{eqnarray*}
Finally we are ready to compute $\Pr[\cD_{2^k,\delta,p}]\leq
\sum_{h=0}^{\bar h-1}\Pr[\cC_{k+h,\delta,p}]+\Pr[\cA_{k+\bar h,\delta}]$.
In $\Pr[\cC_{k+h,\delta,p}]=1/a^{2^h}+(2^{k+h}/m)^{\gamma}$ the
terms $1/a^{2^h}$ were dominated by $h=0$, and the
terms $(2^{k+h}/m)^{\gamma}$ are dominated by $h=\bar h$
which is covered by $\Pr[\cA_{k+\bar h,\delta}]$. We conclude
that 
\[\Pr[\cD_{2^k,\delta,p}]\leq \left\{\begin{array}{ll}
2e^{-\Omega(\alpha2^k\delta^2)} + \tilde\Omega(2^k/m)^{\gamma-1}&\mbox{if }\delta\leq 1\\
(1+\delta)^{-\Omega((1+\delta)\alpha2^k)} + \tilde\Omega(2^k/m)^{\gamma-1}&\mbox{if }\delta=\Omega(1)
\end{array}\right.\]
Since $\gamma$ can always be picked larger, this completes the
proof of \req{eq:lp} in Theorem \ref{thm:linprobe}.

\subsection{The cost of linear probing}\label{sec:lp-cost}
We now return to the costs of the different operations with
linear probing and simple tabulation hashing. 
We have stored a set $S$ of $n$ keys in a table of size $m$. 
Define the fill $\alpha=n/m$ and $\eps=1-\alpha$.  
For any key $q$ we let $R=R(q,S)$ be
the number of filled positions from the hash location of $q$ to the nearest 
empty slot. For insertions and unsuccessful searches, we have
$q\not\in S$, and then the number of cells probed is exactly
$R(q,S)+1$. This also expresses the number of probes when
we delete, but in deletions, we have $q\in S$. Finally,
in a successful search, the number
of probes is bounded by $R(q,S\setminus\{q\})+1$. 
From Theorem \ref{thm:linprobe} we get tail bounds on $R(q,S)$ including
the case where $q\in S$:
\begin{corollary}\label{cor:linprobe}
 For any
  $\gamma=O(1)$ and $\ell\leq n^{1/(3c)}/\alpha$,
\begin{equation}\label{eq:lp*}
\Pr[R(q,S)\geq \ell]\leq \left\{\begin{array}{ll}
2e^{-\Omega(\ell\eps^2)} + (\ell/m)^{\gamma} &\mbox{if }\alpha\geq 1/2\\
\alpha^{-\Omega(\ell)} + (\ell/m)^{\gamma} &\mbox{if }\alpha\leq 1/2
\end{array}\right.
\end{equation}
\end{corollary}
\begin{proof}
When $q\not\in S$, we simply apply \req{eq:lp} from
Theorem \ref{thm:linprobe} with $p=h(q)$.
If $\eps\leq 1/2$, we use $\delta=\eps$, and if $\alpha\leq 1/2$, we
use $(1+\delta)=1/\alpha$ implying $\delta\geq 1/2$. In fact,
we can do almost the same if $q\in S$. We will only apply 
Theorem \ref{thm:linprobe} to $S'=S\setminus\{q\}$ which has
fill $\alpha'<\alpha$.
If $\alpha\geq 1/2$ we
note that \req{eq:lp*} does not
provide a non-trivial bounds if $\ell=O(1/\eps^2)$, so we can
easily assume $\ell\geq 2/\eps$. For any interval of this
length to be filled by $S$, the contribution from $S'$ has
to have a relative deviation of at least $\eps/2$. For
$\alpha\leq 1/2$, we note that we can assume $\ell\geq 4$, and 
we choose $\delta$ such that $(1+2\delta)=1/\alpha$. Since
$1/\alpha\geq 2$, we have $(1+\delta)\leq (3/4)/\alpha$. For
an interval of length $\ell$ to be full, it needs
$(1+2\delta)\alpha\ell\geq 1+(1+\delta)\alpha\ell$ keys from $S$,
so it needs at least $(1+\delta)\alpha\ell$ keys from $S'$.
Now \req{eq:lp*} follows from \req{eq:lp} since $(1+\delta)>
\sqrt{1+2\delta}=\sqrt{1/\alpha}$.
\end{proof}
From Corollary \ref{cor:linprobe} it follows that we for $\alpha\geq 1/2$ get
a tight concentration of $R(q,S)$ around $\Theta(1/\eps^2)$, e.g., for any
moment $p=O(1)$, $\E[R(q,S)^p]=O(1/\eps^{2p})$.

Now consider smaller fills $\alpha\leq 1/2$. 
Corollary \ref{cor:linprobe} does not offer strong bound on
$\Pr[R(q,S)>0]$. It works better when $R(q,S)$ exceeds some
large enough constant. However, in Section~\ref{sec:moment}, we 
show that simple tabulation satisfies a certain 4\th moment bounds, and in
Section~\ref{sec:fourth-lp} \req{eq:fourth-lp}, we show that when
$q\not\in S$, this implies that linear probing fills a location
depending on $h(q)$ with 
probability $O(\alpha)$. 
Thus we add to Corollary \ref{cor:linprobe} that for $q\not\in S$,
\begin{equation}\label{eq:non-empty}
\Pr[R(q,S)>0]=O(\alpha).
\end{equation}
Combining this with the exponential drop for larger $R(q,S)$ in 
Corollary \ref{cor:linprobe}, it follows for any constant moment $p$ that
$\E[R(q,S)^p]=O(\alpha)$ when $q\not\in S$. 

Now consider $q\in S$ as in deletions. The
probability that $S'=S\setminus\{q\}$ fills either
$h(q)$ or $h(q)+1$ is $O(\alpha)$. Otherwise
$S$ fills $h(q)$ leaving $h(q)+1$ empty, and
then $R(q,S)=1$. Therefore, for $q\in S$,
\begin{equation}\label{eq:non-empty-del}
\Pr[R(q,S)>1]=O(\alpha).
\end{equation}
Combining this with the exponential drop for larger $R(q,S)$ in 
Corollary \ref{cor:linprobe}, it follows for any constant moment $p$ that
$\E[R(q,S)^p]=1+O(\alpha)$ when $q\in S$.

\subsection{Larger intervals}
To finish the proof of Theorem \ref{thm:linprobe}, we
need to consider the case of larger intervals. We want to show that, 
with probability $1-n^{-\gamma}$ for any $\gamma=O(1)$, 
for every interval $I$ where the mean number of keys is $\alpha|I|\geq 1$,
the deviation is at most 
\[O\left(\sqrt{\alpha|I|}\log^c n\right).\]
Consider an interval $I$ with $\alpha|I|\geq 1$. 
As in the proof of Claim \ref{claim:small-delta},
we consider a maximal dyadic decomposition into level bins
with up to two $j$-bins for each $j\leq i=\lfloor\log_2 |I|\rfloor$.
Let $j_0=\lceil \log_2 (1/\alpha)\rceil$. Again we
round to level $j_0$, discarding the lower level bins, but
adding a $j_0$-bin on either side.  The deviation in $I$ is
bounded by total deviation of the internal bins plus the total
contents of the side bins. The expected number of keys
in each side bins is $\alpha2^{j_0}\leq 2$. 

For each $j\in\{j_0,...,i\}$, we apply
Theorem  \ref{thm:chernoff} with $m'=m/2^j\leq \alpha m=n$ bins.
W.h.p., the maximal deviation for any $j$-bins is 
$O\left(\sqrt{n/m'}\log^c n\right)=O\left(\sqrt{\alpha 2^j}\log^c n\right)$.
This gives a total deviation of at most
\[2\left(2+O\left(\sqrt{\alpha 2^{j_0}}\log^c n\right)+
\sum_{j=j_0}^i O\left((\sqrt{\alpha 2^j}\log^c n\right)\right)=
O\left((\sqrt{\alpha|I|}\log^c n\right)\textnormal,\]
as desired. For each $j$ there is an error probability of $n^{-\gamma'}$ for any $\gamma'=O(1)$. The error probability over all $j\in\{j_0,...,i\}$ 
is $(i-j_0+1)n^{-\gamma'}$. Here 
$i-j_0\leq \log_2 m-\log_2(1/\alpha)=\log_2 m-\log_2\frac mn=\log_2n$,
so $(i-j_0+1)n^{-\gamma}\leq n^{-\gamma'}(1+\log n)$. 
This completes the proof Theorem \ref{thm:linprobe}.

\subsection{Set estimation}
We can easily apply our results for set estimation where one
saves a bottom-$k$ sketch. More precisely, suppose we for
a set $A$ store a sample $S$ consisting of the $k$ keys with
the smallest hash values. Consider now some subset $B\subseteq A$.
We then use $|B\cap S|/k$ as an estimator for $|B|/|A|$.
We can use the above bounds to bound the probability that
this estimator is wrong by more than a factor $\frac{1+\delta}{1-\delta}$.
Let $\tau$ be the $k$th hash value of $A$. First we use
the bounds to argue that $\tau=(1\pm \delta)k/|A|$. Next
we use them to argue that the number of elements from $B$
below any given $\tau'$ is $(1\pm\delta)\tau'|B|$. Applying
this with $\tau'=(1-\delta)k/|A|,(1+\delta)k/|A|$, we get
the desired bound.

\newcommand\cqed{\hfill\ensuremath{\Diamond}}

\section{Cuckoo Hashing}

We are now going to analyze cuckoo hashing. In our
analysis of chaining and linear probing, we did not worry so much
about constants, but with Cuckoo hashing, we do have to worry about obstructions
that could be stem from the hashing of just a constant
number of keys, e.g., as an extreme case we could have three keys
sharing the same two hash locations. It is, in fact, a constant
sized obstruction that provides the  negative side of our result:

\begin{observation}
There exists a set $S$ of $n$ keys such that cuckoo hashing with simple
tabulation hashing cannot place $S$ into two tables of size $2n$ with
probability $\Omega(n^{-1/3})$.
\end{observation}

\begin{proof}
The hard instance is the 3-dimensional cube $[n^{1/3}]^3$. Here is a
sufficient condition for cuckoo hashing to fail:
\begin{itemize*}
\item there exist $a,b,c \in [n^{1/3}]^2$ with $h_0(a)=h_0(b)=h_0(c)$;
\item there exist $x,y \in [n^{1/3}]$ with $h_1(x) = h_1(y)$.
\end{itemize*}

If both happen, then the elements $ax, ay, bx, by, cx, cy$ cannot be
hashed. Indeed, on the left side $h_0(a)=h_0(b)=h_0(c)$ so they only
occupy 2 positions. On the right side, $h_1(x)=h_1(y)$ so they only
occupy 3 positions. In total they occupy $5<6$ positions.

The probability of 1.~is asymptotically $(n^{2/3})^3 / n^2 =
\Omega(1)$. This is because tabulation (on two characters) is
3-independent. The probability of 2.~is asymptotically $(n^{1/3})^2 /
n = \Omega(1 / n^{1/3})$. So overall cuckoo hashing
fails with probability $\Omega(n^{-1/3})$.
\end{proof}

Our positive result will effectively show that this is the worst
possible instance: for any set $S$, the failure probability is
$O(n^{-1/3})$.

The proof is an encoding argument. A tabulation hash function from
$\Sigma^c \mapsto [m]$ has entropy $|\Sigma|^c \lg m$ bits; we have
two random functions $h_0$ and $h_1$. If, under some event $\evt$, one
can encode the two hash functions $h_0, h_1$ using $\big( 2 |\Sigma|^c
\lg m \big) - \gamma$ bits, it follows that $\Pr[\evt] =
O(2^{-\gamma})$. Letting $\evt_S$ denote the event that cuckoo hashing
fails on the set of keys $S$, we will demonstrate a saving of $\gamma
= \frac{1}{3} \lg n - f(c,\eps)=\frac{1}{3} \lg n - O(1)$ bits in the
encoding. Note that we are analyzing simple tabulation on a
\emph{fixed} set of $n$ keys, so both the encoder and the decoder know
$S$.

We will consider various cases, and give algorithms for encoding some
subset of the hash codes (we can afford $O(1)$ bits in the beginning
of the encoding to say which case we are in). At the end, the encoder
will always list all the remaining hash codes in order. If the
algorithm chooses to encode $k$ hash codes, it will use space at most
$k\lg m - \frac{1}{3} \lg n + O(1)$ bits. That is, it will save
$\frac{1}{3} \lg n - O(1)$ bits in the complete encoding of $h_0$ and
$h_1$.

\ptsubsection{An easy way out}
A {\em subkey\/} is a set of position-characters on distinct
positions.  If $a$ is a subkey, we let $C(a) = \{ x \in S \mid a
\subseteq x \}$ be the set of ``completions'' of $a$ to a valid key.

We first consider an easy way out: there subkeys $a$ and
$b$ on the positions such that $|C(a)| \ge n^{2/3}, |C(b)|
\ge n^{2/3},$ and $h_i(a) = h_i(b)$ for some $i \in \{0,1\}$. Then we
can easily save $\frac{1}{3} \lg n - O(1)$ bits. 
First we write the
set of positions of $a$ and $b$, and the side of the collision ($c+1$
bits). There are at most $n^{1/3}$ subkeys on those positions
that have $\ge n^{2/3}$ completions each, so we can write the
identities of $a$ and $b$ using $\frac{1}{3}\lg n$ bits each. We write
the hash codes $h_i$ for all characters in $a \Delta b$ (the symmetric
difference of $a$ and $b$), skipping the last one, since it can be
deduced from the collision. This uses $c + 1 + 2\cdot \frac{1}{3}\lg n
+ (|a\Delta b| - 1) \lg m$ bits to encode $|a\Delta b|$ hash codes, so
it saves $\frac{1}{3} \lg n - O(1)$ bits.

The rest of the proof assumes that there is no easy way out.

\ptsubsection{Walking Along an Obstruction}
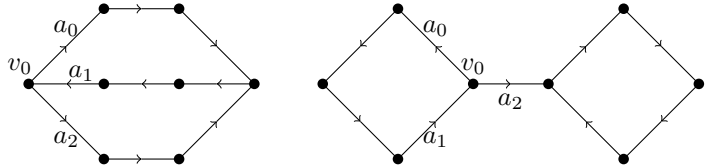
\begin{figure}
\begin{center}
\centering
\begin{tabular}{c@{\qquad}c}
  \begin{tikzpicture} 
    \newcommand{\arrowline}[4]{{
      \FPeval{\midx}{({#1}+{#3}) / 2}
      \FPeval{\midy}{({#2}+{#4}) / 2}
      \draw[->] (#1, #2) -- (\midx, \midy);
      \draw (\midx, \midy) -- (#3, #4);
      }}
    \fill (0,2) circle(0.07) node[above] {\small $v_0$~~~}
          (1,1) circle(0.07) 
          (2,1) circle(0.07) (3,2) circle(0.07) (1,2) circle(0.07)
          (2,2) circle(0.07) (1,3) circle(0.07) (2,3) circle(0.07);
    \arrowline{0}{2}{1}{1}
    \draw (0.5, 1.5) node[below] {\small $a_2$};
    \draw (0.5, 2.5) node[above] {\small $a_0$};
    \draw (0.7, 1.9) node[above] {\small $a_1$};
    \arrowline{1}{1}{2}{1}
    \arrowline{2}{1}{3}{2}
    \arrowline{1}{2}{0}{2}
    \arrowline{0}{2}{1}{3}
    \arrowline{1}{3}{2}{3}
    \arrowline{2}{3}{3}{2}
    \arrowline{3}{2}{2}{2}
    \arrowline{2}{2}{1}{2}
  \end{tikzpicture}
&
  \begin{tikzpicture} 
    \newcommand{\arrowline}[4]{{
      \FPeval{\midx}{({#1}+{#3}) / 2}
      \FPeval{\midy}{({#2}+{#4}) / 2}
      \draw[->] (#1, #2) -- (\midx, \midy);
      \draw (\midx, \midy) -- (#3, #4);
      }}
    \fill (0,2) circle(0.07) (1,1) circle(0.07) 
          (2,2) circle(0.07) node[above] {\small $v_0$} 
          (1,3) circle(0.07)
          (3,2) circle(0.07) (4,1) circle(0.07)
          (5,2) circle(0.07) (4,3) circle(0.07);
    \arrowline{2}{2}{1}{3}
    \draw (1.5, 2.5) node[above] {\small $a_0$};
    \arrowline{1}{3}{0}{2}
    \arrowline{0}{2}{1}{1}
    \arrowline{1}{1}{2}{2}
    \draw (1.5, 1.5) node[below] {\small $a_1$};
    \arrowline{5}{2}{4}{1}
    \arrowline{4}{3}{5}{2}
    \arrowline{3}{2}{4}{3}
    \arrowline{2}{2}{3}{2}
    \draw (2.5, 2) node[below] {\small $a_2$};
    \arrowline{4}{1}{3}{2}
  \end{tikzpicture}
\end{tabular}
\end{center}
\caption{Minimal obstructions to cuckoo hashing.}
  \label{fig:cycle}
\end{figure}
Consider the bipartite graph with $m$ nodes on each side and $n$ edges
going from $h_0(x)$ to $h_1(x)$ for all $x\in S$. Remember that cuckoo
hashing succeeds if and only if no component in this graph has more edges
than nodes. Assuming cuckoo hashing failed, the encoder can find
a subgraph with one of two possible obstructions: (1) a cycle with a chord;
or (2) two cycles connected by a path (possibly a trivial path,
i.e. the cycles simply share a vertex).

Let $v_0$ be a node of degree $3$ in such an obstruction, and let its
incident edges be $a_0, a_1, a_2$. The obstruction can be traversed by
a walk that leaves $v_0$ on edge $a_0$, returns to $v_0$ on edge $a_1$,
leaves again on $a_2$, and eventually meets itself. Other than visiting
$v_0$ and the last node twice, no node or edge is repeated. See
Figure~\ref{fig:cycle}.

Let $x_1, x_2, \dots$ be the sequence of keys in the walk. The first
key is $x_1 = a_0$. Technically, when the walk meets itself at the end,
it is convenient to expand it with an extra key, namely the one it
first used to get to the meeting point. This repeated key marks the
end of the original walk, and we chose it so that it is not identical
to the last original key.  Let $x_{\le i} = \bigcup_{j\le i} x_j$ be
the position-characters seen in keys up to $x_i$. Define $\hat{x}_i =
x_i \setminus x_{<i}$ to be the position-characters of $x_i$ not seen
previously in the sequence. Let $k$ be the first position such that
$\hat{x}_{k+1} = \emptyset$. Such a $k$ certainly exists, since the
last key in our walk is a repeated key.

At a high level, the encoding algorithm will encode the hash codes of
$\hat{x}_1, \dots, \hat{x}_k$ in this order. Note that the
obstruction, hence the sequence $(x_i)$, depends on the hash functions
$h_0$ and $h_1$. Thus, the decoder does not know the sequence, and it
must also be written in the encoding.

For notational convenience, let $h_i = h_{i \bmod 2}$. This means that
in our sequence $x_i$ and $x_{i+1}$ collide in their $h_i$ hash code,
that is $h_i(x_i) = h_i(x_{i+1})$. Formally, we define 3 subroutines:
\begin{description*}
\item[\Id($x$):] Write the identity of $x\in S$ in the
  encoding, which takes $\lg n$ bits.

\item[\Hashes($h_i, x_k$):] Write the hash codes $h_i$ of the
  characters $\hat{x}_k$. This takes $|\hat{x}_k| \lg m$ bits.

\item[\Coll($x_i, x_{i+1}$):] Document the collision $h_i(x_i) =
  h_i(x_{i+1})$. We write all $h_i$ hash codes of characters
  $\hat{x}_i \cup \hat{x}_{i+1}$ in some fixed order. The last hash
  code of $\hat{x}_i \Delta \hat{x}_{i+1}$ is redundant and will be
  omitted. Indeed, the decoder can compute this last hash code from
  the equality $h_i(x_i) = h_i(x_{i+1})$. Since $\hat{x}_{i+1} =
  x_{i+1} \setminus x_{\le i}$, $\hat{x}_{i+1} \setminus \hat{x}_i \ne
  \emptyset$, so there exists a hash code in $\hat{x}_i \Delta
  \hat{x}_{i+1}$. This subroutine uses $\big( |\hat{x}_i \cup
  \hat{x}_{i+1}| - 1\big) \lg m$ bits, saving $\lg m$ bits compared to
  the trivial alternative: $\proc{Hashes}(h_i, x_i);
  \proc{Hashes}(h_i, x_{i+1})$.
\end{description*}
To decode the above information, the decoder will need enough context
to synchronize with the coding stream. For instance, to decode
$\Coll(x_i, x_{i+1})$, one typically needs to know $i$, and the
identities of $x_i$ and $x_{i+1}$.

Our encoding begins with the value $k$, encoded with $O(\lg k)$ bits,
which allows the decoder to know when to stop. The encoding proceeds
with the output of the stream of operations:
\begin{align*}
 \Id(x_1);& \Hashes(h_0, x_1); \Id(x_2); \Coll(x_1, x_2); \\
   \dots~ & \Id(x_k); \Coll(x_k, x_{k-1}); \Hashes(h_k, x_k) 
\end{align*}
We observe that for each $i > 1$, we save $\eps$ bits of
entropy. Indeed, $\Id(x_i)$ uses $\lg n$ bits, but $\Coll(x_{i-1},
x_i)$ then saves $\lg m = \lg((1+\eps)n) \ge \eps + \lg n$ bits.

The trouble is $\Id(x_1)$, which has an upfront cost of $\lg n$
bits. We must devise algorithms that modify this stream of operations
and save $\frac{4}{3} \lg n - O(1)$ bits, giving an overall saving of
$\frac{1}{3} \lg n - O(1)$. (For intuition, observe that a saving that
ignores the cost of $\Id(x_1)$ bounds the probability of an
obstruction at some fixed vertex in the graph. This probability must
be much smaller than $1/n$, so we can union bound over all
vertices. In encoding terminology, this saving must be much more than
$\lg n$ bits.)

We will use modifications to all types of operations. For instance, we
will sometimes encode $\Id(x)$ with much less than $\lg n$ bits. At
other times, we will be able to encode $\Coll(x_i, x_{i+1})$ with the
cost of $|\hat{x}_i \cup \hat{x}_{i+1}| - 2$ characters, saving $\lg
n$ bits over the standard encoding. 

Since we will make several such modifications, it is crucial to verify
that they only touch distinct operations in the stream.  Each
modification to the stream will be announced at the beginning of the
stream with a pointer taking $O(\lg k)$ bits. This way, the decoder
knows when to apply the special algorithms. We note that terms of
$O(\lg k)$ are negligible, since we are already saving $\eps k$ bits
by the basic encoding ($\eps$ bits per edge). For any $k$, $O(\lg k)
\le \eps k + f(c,\eps) = k + O(1)$. Thus, if our overall saving is
$\frac{1}{3} \lg n - O(\lg k) + \eps k$, it achieves the stated bound
of $\lg n - O(1)$.

\ptsubsection{Safe Savings}
Remember that $\hat{x}_{k+1} = \emptyset$, which suggests that we can
save a lot by local changes towards the end of the encoding.  We have
$x_{k+1} \subset x_{\le k}$, so $x_{k+1} \setminus x_{<k} \subseteq
\hat{x}_k$.  We will first treat the case when $x_{k+1} \setminus
x_{<k}$ is a proper subset of $\hat{x}_k$ (including the empty
subset). This is equivalent to $\hat{x}_k \not\subset x_{k+1}$.

\begin{lemma}[safe-strong]    \label{lem:safe-strong}
If $\hat{x}_k \not\subset x_{k+1}$, we can save $\lg n - O(c\lg k)$
bits by changing $\Hashes(x_k)$.
\end{lemma}

\begin{proof}
We can encode $\Id(x_{k+1})$ using $c\lg k$ extra bits, since it consists
only of known characters from $x_{\le k}$. For each position
$1\twodots c$, it suffices to give the index of a previous $x_i$ that
contained the same position-character. Then, we will write all hash
codes $h_k$ for the characters in $\hat{x}_k$, except for some $\alpha
\in \hat{x}_k \setminus x_{k+1}$. From $h_k(x_k) = h_k(x_{k+1})$, we
have $h_k(\alpha) = h_k(x_k \setminus \{\alpha\}) \oplus
h_k(x_{k+1})$. All quantities on the right hand side are known (in
particular $\alpha \notin x_{k+1}$), so the decoder can compute
$h_k(\alpha)$.
\end{proof}

It remains to treat the case when the last revealed characters of
$x_{k+1}$ are precisely $\hat{x}_k$: $\hat{x}_k \subset x_{k+1}$.
That is, both $x_k$ and $x_{k+1}$ consist of $\hat{x}_k$ and some
previously known characters. In this case, the collision $h_k(x_k) =
h_k(x_{k+1})$ does not provide us any information, since it reduces to
the trivial $h_k(\hat{x}_k) = h_k(\hat{x}_k)$. Assuming that we didn't
take the ``easy way out'', we can still guarantee a more modest saving
of $\frac{1}{3}\lg n$ bits:

\begin{lemma}[safe-weak]  \label{lem:safe-weak}
Let $K$ be the set of position-characters known before encoding
$\Id(x_i)$, and assume there is no easy way out. If $x_i\Delta
x_{i+1}\subseteq x_{<i}$, then we can encode both $\Id(x_i)$ and
$\Id(x_{i+1})$ using a total of $\frac{2}{3} \lg n + O(c \lg |K|)$
bits.
\end{lemma}

A typical case where we apply the lemma is $i=k$ and $K = x_{<k}$. If
$\hat{x}_k \subset x_{k+1}$, we have $x_k \Delta x_{k+1} \subset K$.
Thus, we can obtain $\Id(x_k)$ for roughly $\frac{2}{3} \lg n$ bits,
which saves $\frac{1}{3} \lg n$ bits.

\begin{proof}[Proof of Lemma \ref{lem:safe-weak}]
With $O(c\lg k)$ bits, we can code the subkeys $x_i \cap x_{<i}$ and
$x_{i+1} \cap x_{<i}$. It remains to code $z=x_i \setminus x_{<i} =
x_{i+1} \setminus x_{<i}$. Since $z$ is common to both keys $x_i$ and
$x_{i+1}$, we have that $x_i\setminus z$ and $x_{i+1}\setminus z$ are
subkeys on the same positions. With no easy way out and $h_i(x_i
\setminus z)=h_i(x_{i+1} \setminus z)$, we must have $|C(x_i \setminus
z)| \le n^{2/3}$ \underline{or} $|C(x_{i+1} \setminus z)| \le
n^{2/3}$. In the former case, we code $z$ as a member of
$C(x_i\setminus z)$ with $\lceil \frac{2}{3}\lg n\rceil$ bits;
otherwise we code $z$ as member of $C(x_{i+1} \setminus z)$.
\end{proof}

\ptsubsection{Piggybacking}
Before moving forward, we present a general situation when we can save
$\lg n$ bits by modifying a $\Coll(x_i, x_{i+1})$ operation:
\begin{lemma}   \label{lem:coll}
We can save $\lg n - O(\lg k)$ bits by modifying $\Coll(x_i, x_{i+1})$
if we have identified two (sub)keys $e$ and $f$ satisfying:
\[
h_i(e) = h_i(f);~~   e\Delta f \subset x_{\le i+1};~~
\emptyset \ne (e\Delta f) \setminus x_{<i} 
          \ne (x_i \Delta x_{i+1}) \setminus x_{<i}.
\]
\end{lemma}

\begin{proof}
In the typical encoding of $\Coll(x_i, x_{i+1})$, we saved one
redundant character from $h_i(x_i) = h_i(x_{i+1})$, which is an
equation involving $(x_i \Delta x_{i+1}) \setminus x_{<i}$ and some
known characters from $x_{<i}$. The lemma guarantees a second linearly
independent equation over the characters $\hat{x}_i \cup
\hat{x}_{i+1}$, so we can save a second redundant character.

Formally, let $\alpha$ be a position-character of $(e\Delta f)
\setminus x_{<i}$, and $\beta$ a position-character in $(x_i \Delta
x_{i+1}) \setminus x_{<i}$ but outside $(e\Delta f) \setminus
x_{<i}$. Note $\beta \ne \alpha$ and such a $\beta$ exists by
assumption. We write the $h_i$ hash codes of position characters
$(\hat{x}_i \cup \hat{x}_{i+1}) \setminus \{\alpha, \beta\}$.  The
hash $h_i(\alpha)$ can be deduced since $\alpha$ is the last unknown
in the equality $h_i(e \setminus f) = h_i(f \setminus e)$. The hash
$h_i(\beta)$ can be deduced since it is the last unknown in the
equality $h_i(x) = h_i(x_{i+1})$.
\end{proof}

While the safe saving ideas only require simple local modifications to
the encoding, they achieve a weak saving of $\frac{1}{3}\lg n$ bits
for the case $\hat{x}_k \subset x_{k+1}$. A crucial step in our proof
is to obtain a saving of $\lg n$ bits for this case.  We do this by
one of the following two lemmas:

\begin{lemma}[odd-size saving]   \label{lem:bad-pig}
Consider two edges $e,f$ and an $i \le k-2$ satisfying:
\[ h_{i+1}(e) = h_{i+1}(f); \qquad
   e \setminus x_{\le i} \ne f \setminus x_{\le i}; \qquad
   e \setminus x_{\le i+1} = f \setminus x_{\le i+1}.
\]
We can save $\lg n - O(c\lg k)$ bits by changing $\Coll(x_{i+1},
x_{i+2})$.
\end{lemma}

\begin{proof}
We apply Lemma~\ref{lem:coll} with the subkeys $\tilde{e} = e
\setminus f$ and $\tilde{f} = f\setminus e$. We can identify these in
$O(c\lg k)$ bits, since they only contain characters of $x_{\le
  i+1}$. Since $e$ and $f$ have different free characters before
$\hat{x}_{i+1}$, but identical free characters afterward, it must be
that $\tilde{e} \cup \tilde{f} \subset x_{i+1}$ by $\tilde{e} \cup
\tilde{f} \not\subseteq x_{\le i}$. To show $(e\Delta f) \setminus
x_{<i} \ne (x_{i+1} \Delta x_{i+2}) \setminus x_{\le i}$, remark that
$\hat{x}_{i+2} \ne \emptyset$ and $\hat{x}_{i+2}$ cannot have
characters of $\tilde{e} \cup \tilde{f}$. Thus, Lemma~\ref{lem:coll}
applies.
\end{proof}

\begin{lemma}[piggybacking]   \label{lem:good-pig}
Consider two edges $e,f$ and an $i \le k-1$ satisfying:
\[ h_i(e) = h_i(f); \qquad
   e \setminus x_{\le i} \ne f \setminus x_{\le i}; \qquad
   e \setminus x_{\le i+1} = f \setminus x_{\le i+1}.
\]
We can encode $\Id(e)$ and $\Id(f)$ using only $O(c\lg k)$ bits, after
modifications to $\Id(x_i)$, $\Id(x_{i+1})$, and $\Coll(x_i,
x_{i+1})$.
\end{lemma}

The proof of this lemma is more delicate, and is given below.  The
difference between the two lemmas is the parity (side in the
bipartite graph) of the collision of $x_i$ and $x_{i+1}$ versus the
collision of $e$ and $f$. In the second result, we cannot actually
save $\lg n$ bits, but we can encode $\Id(e)$ and $\Id(f)$ almost for
free: we say $e$ and $f$ piggyback on the encodings of $x_i$ and
$x_{i+1}$.

Through a combination of the two lemmas, we can always achieve a
saving $\lg n$ bits in the case $\hat{x}_k \subset x_{k+1}$, improving
on the safe-weak bound:

\begin{lemma}  \label{lem:piggy-overall}
Assume $k$ is minimal such that $\hat{x}_k \subset x_{k+1}$. We can
save $\lg n - O(c\lg k)$ bits if we may modify any operations in the
stream, up to those involving $x_{k+1}$.
\end{lemma}

\begin{proof}
We will choose $e = x_k$ and $f = x_{k+1}$. We have $e \setminus x_{<
  k} = f \setminus x_{<k} = \hat{x}_k$. On the other hand, $e
\setminus x_1 \ne f \setminus x_1$ since $x_1$ only reveals one
character per position. Thus there must be some $1 \le i < k-1$ where
the transition happens: $e \setminus x_{\le i} \ne f \setminus x_{\le
  i}$ but $e \setminus x_{\le i+1} = f \setminus x_{\le i+1}$. If $i$
has the opposite parity compared to $k$, Lemma~\ref{lem:bad-pig} saves
a $\lg n$ term. (Note that $i \le k-2$ as required by the lemma.)

If $i$ has the same parity as $k$, Lemma~\ref{lem:good-pig} gives us
$\Id(x_k)$ at negligible cost. Then, we can remove the operation
$\Id(x_k)$ from the stream, and save $\lg n$ bits. (Again, note that
$i \le k-2$ as required.)
\end{proof}

\begin{proof}[Proof of Lemma~\ref{lem:good-pig}]
The lemma assumed $e \setminus x_{\le i} \ne f \setminus x_{\le i}$
but $e \setminus x_{\le i+1} = f \setminus x_{\le i+1}$. Therefore, $e
\Delta f \subset x_{\le i+1}$ and $(e \Delta f) \cap \hat{x}_{i+1} \ne
\emptyset$. Lemma~\ref{lem:coll} applies if we furthermore have
$(e\Delta f) \setminus x_{<i} \ne (x_i \Delta x_{i+1}) \setminus
x_{<i}$. If the lemma applies, we have a saving of $\lg n$, so we can
afford to encode $\Id(e)$. Then $\Id(f)$ can be encoded using $O(c\lg
k)$ bits, since $f$ differs from $e$ only in position-characters from
$x_{\le i+1}$.

If the lemma does not apply, we have a lot of structure on the
keys. Let $y = \hat{x}_i \setminus (e\cup f)$ and $g = e \setminus
x_{\le i+1} = f \setminus x_{\le i+1}$. We must have $y \subset
x_{i+1}$, for otherwise $\hat{x}_i \setminus x_{i+1}$ contains an
elements outside $e\Delta f$ and the lemma applies. We must also have
$\hat{x}_{i+1} \subset e \cup f$.

We can write $\Id(x_i)$, $\Id(x_{i+1})$, $\Id(e)$, and $\Id(f)$ using
$2\lg n + O(c\lg k)$ bits in total, as follows:
\begin{itemize*}
\item the coordinates on which $y$ and $g$ appear, taking $2c$ bits.

\item the value of $y$ using Huffman coding. Specifically, we consider
  the projection of all $n$ keys on the coordinates of $y$. In this
  distribution, $y$ has frequency $\frac{C(y)}{n}$, so its Huffman
  code will use $\lg \frac{n}{C(y)} + O(1)$ bits.

\item the value of $g$ using Huffman coding. This uses $\lg
  \frac{n}{C(g)} + O(1)$ bits.

\item if $C(y) \le C(g)$, we write $x_i$ and $x_{i+1}$. Each of these
  requires $\lceil \log_2 C(y) \rceil$ bits, since $y \subset x_i,
  x_{i+1}$ and there are $C(y)$ completions of $y$ to a full
  key. Using an additional $O(c\lg k)$ bits, we can write $e \cap
  x_{\le i+1}$ and $f \cap x_{\le i+1}$. Remember that we already
  encoded $g = e \setminus x_{\le i+1} = f \setminus x_{\le
    i+1}$, so the decoder can recover $e$ and $f$.

\item if $C(g) < C(y)$, we write $e$ and $f$, each requiring $\lceil
  \log_2 C(g) \rceil$ bits. Since we know $y = \hat{x}_i \setminus
  (e\cup f)$, we can write $x_i$ using $O(c\lg k)$ bits: write the old
  characters outside $\hat{x}_i$, and which positions of $e\cup f$ to
  reuse in $\hat{x}_i$. We showed $\hat{x}_{i+1} \subset e \cup f$, so
  we can also write $x_{i+1}$ using $O(c\lg k)$.
\end{itemize*}

\noindent
Overall, the encoding uses space:
$ \lg \tfrac{n}{C(\xi)} + \lg \tfrac{n}{C(\hat{e}_{i+1})}
+ 2\lg \min \big\{ C(\xi), C(\hat{e}_{i+1}) \big\} + O(c\lg k)
~\le~ 2\lg n + O(c\lg k) \hfill$.
\end{proof}

\ptsubsection{Putting it Together}
We now show how to obtain a saving of at least $\frac{4}{3} \lg n -
O(c\lg k)$ bits by a careful combination of the above techniques.
Recall that our starting point is three edges $a_0, a_1, a_2$ with
$h_0(a_0)=h_0(a_1)=h_0(a_2)$. The walk $x_1,...,x_{k+1}$ started with
$x_1 = a_0$ and finished when $\hat{x}_{k+1} = \emptyset$. We will now
involve the other starting edges $a_1$ and $a_2$. The analysis will
split into many cases, each ended by a '$\Diamond$'.

\newcommand{\case}[1]{{\bf{#1}.}~}

\case{Case 1: One of $a_1$ and  $a_2$ contains a free character}
Let $j\in \{1,2\}$ such that $a_j\not\subseteq x_{\le k}$.  Let $y_1 =
a_j$. We consider a walk $y_1, y_2, \dots$ along the edges of the
obstruction.Let $\hat{y}_i = y_i \setminus x_{\le k} \setminus y_{\le
  i}$ be the free characters of $y_i$ (which also takes all $x_i$'s
into consideration). We stop the walk the first time we observe
$\hat{y}_{\ell+1} = \emptyset$. This must occur, since the graph is
finite and there are no leaves (nodes of degree one) in the
obstruction. Thus, at the latest the walk stops when it repeats an
edge.

We use the standard encoding for the second walk:
\begin{align*}
 \Id(y_1); & \Coll(a_0, y_1); \Id(y_2); \Coll(y_2, y_1); \\
    \dots; & \Id(y_\ell); \Coll(y_{\ell-1}, y_\ell); \Hashes(h_\ell,
    y_\ell) 
\end{align*} 

Note that every pair $\Id(y_j), \Coll(y_{j-1}, y_j)$ saves $\eps$
bits, including the initial $\Id(y_1), \Coll(a_0, y_1)$. To end the
walk, we can use one of the safe savings of Lemmas \ref{lem:safe-strong}
and \ref{lem:safe-weak}. These give a saving of $\frac{1}{3}\lg n -
O(c \lg (\ell+k))$ bits, by modifying only $\Hashes(h_\ell, y_\ell)$
or $\Id(y_\ell)$. These local changes cannot interfere with the first
walk, so we can use any technique (including piggybacking) to save
$\lg n-O(c\log k)$ bits from the first walk. We obtain a total saving
of $\frac{4}{3}\lg n - O(1)$, as required. \cqed


We are left with the situation $a_1 \cup a_2 \subseteq x_{\le
  k}$. This includes the case when $a_1$ and $a_2$ are actual edges
seen in the walk $x_1, \dots, x_k$. 

Let $t_j$ be the first time $a_j$ becomes known in the walk; that is,
$a_j \not\subseteq x_{<t_j}$ but $a_j \subseteq x_{\le t_j}$. By
symmetry, we can assume $t_1\leq t_2$. We begin with two simple cases.

\case{Case 2: For some $j\in \{1,2\}$, $t_j$ is even and $t_j<k$}
We will apply Lemma~\ref{lem:coll} and save $\lg n - O(c\lg k)$ bits
by modifying $\Coll(x_{t_j}, x_{t_j+1})$. Since $t_j<k$, this does not
interact with safe savings at the end of the stream, so we get total
saving of at least $\frac{4}{3}\lg n - O(c\lg k)$.

We apply Lemma~\ref{lem:coll} on the keys $e=a_0$ and $f=a_j$.  We
must first write $\Id(a_j)$, which takes $O(c\lg k)$ bits given
$x_{\leq k}$. We have $a_0 \cup a_j \subseteq x_{\le t_j}$ by
definition of $t_j$. Since $a_j \cap \hat{x}_{t_j} \ne \emptyset$ and
$\hat{x}_{t_j + 1} \cap (a_j \cup a_0) = \emptyset$, the lemma applies.
\cqed

\case{Case 3: For some $j \in \{1,2\}$, $t_j$ is odd and
 $a_j\setminus x_{< t_j-1}\ne \hat{x}_{t_j-1} \Delta\hat{x}_{t_j}$} 
This assumption is exactly what we need to apply Lemma~\ref{lem:coll}
with $e=a_0$ and $f=a_j$. Note that $h_0(e) = h_0(f)$ and $t_j$ is
odd, so the lemma modifies $\Coll(x_{t_j-1}, x_{t_j})$.  The lemma can
be applied in conjunction with any safe saving, since the safe savings
only require modifications to $\Id(x_k)$ or $\Hashes(h_k, x_k)$. \cqed


We now deal with two cases when $t_1 = t_2$ (both being odd or
even). These require a combination of piggybacking followed by
safe-weak savings. Note that in the odd case, we may assume
$a_1\setminus x_{< t-1}= a_2\setminus x_{< t-1} =
\hat{x}_{t-1}\Delta\hat{x}_t$ (due to case 3 above), and in the even
case we may assume $t_1 = t_2 = k$ (due to case 2 above).

\case{Case 4: $t_1=t_2=t$ is odd and $a_1\setminus x_{< t-1}=
  a_2\setminus x_{< t-1} = \hat{x}_{t-1}\Delta\hat{x}_t$} 
We first get $a_1$ and $a_2$ by piggybacking or odd-side saving. Let
$i$ be the largest value such that $a_1\setminus x_{\le i} \ne
a_2\setminus x_{\le i}$.  Since $a_1 \setminus x_{<t-1} = a_2
\setminus x_{<t-1}$, we have $i \le t-3$. The last key that
piggybacking or odd-side saving can interfere with is $x_{t-2}$.

We will now use the safe-weak saving of Lemma~\ref{lem:safe-weak} to
encode $\Id(x_{t-1})$ and $\Id(x_t)$. The known characters are $K =
x_{< t-1} \cup a_1 \cup a_2$, so $x_{t-1}\Delta x_t \subseteq K$.
Lemma \ref{lem:safe-weak} codes both $\Id(x_{t-1})$ and $\Id(x_t)$
with $\frac{2}{3} \lg n + O(c\lg k)$ bits, which represents a saving
of roughly $\frac{4}{3} \lg n$ over the original encoding of the two
identities. We don't need any more savings from the rest of the walk
after $x_t$.  \cqed

\case{Case 5: $t_1=t_2=k$ is even}
Thus, $k$ is even and the last characters of $a_1$ and $a_2$ are
only revealed by $\hat{x}_k$. 

\begin{lemma}
We can save $2\lg n - O(c\lg k)$ bits by modifying $\Hashes(h_k,
x_k)$, unless both: (1) $a_1 \cap \hat{x}_k = a_2 \cap \hat{x}_k$; and
(2)$\hat{x}_k \setminus x_{k+1}$ is the empty set or equal to $a_1
\cap \hat{x}_k$.
\end{lemma}

\begin{proof}
The $h_0$ hash codes of the following 3 subkeys are known from the
hash codes in $x_{<k}$: $a_1 \cap \hat{x}_k$, $a_2 \cap \hat{x}_k$
(both because we know $h_0(a_0) = h_0(a_1) = h_0(a_2)$), and
$\hat{x}_k \setminus x_{k+1}$ (since $x_k$ and $x_{k+1}$ collide).  If
two of these subsets are distinct and nonempty, we can choose two
characters $\alpha$ and $\beta$ from their symmetric difference. We
can encode all characters of $\hat{x}_k$ except for $\alpha$ and
$\beta$, whose hash codes can be deduced for free.

Since $a_j \cap \hat{x}_k \ne \empty$ in the current case, the
situations when we can find two distinct nonempty sets are: (1)
$a_1\cap \hat{x}_k \ne a_2 \cap \hat{x}_k$; or (2) $a_1\cap \hat{x}_k
= a_2 \cap \hat{x}_k$ but $\hat{x}_k \setminus x_{k+1}$ is nonempty
and different from them.
\end{proof}

From now on assume the lemma fails. We can still save $\lg n$ bits by
modifying $\Hashes(h_k, x_k)$. We reveal all hash codes of
$\hat{x}_k$, except for one position-character $\alpha \in a_1 \cap
\hat{x}_k$. We then specify $\Id(a_1)$, which takes $O(c\lg k)$
bits. The hash $h_0(\alpha)$ can then be deduced from $h_0(a_1) =
h_0(a_0)$.

We will now apply piggybacking or odd-side saving to $a_1$ and $a_2$.
Let $i$ be the largest value with $a_1 \setminus x_{\le i} \ne a_2
\setminus x_{\le i}$. Note that $a_1 \setminus x_{<k} = a_2 \setminus
x_{<k}$, so $i < k-1$. If $i$ is odd, Lemma~\ref{lem:bad-pig}
(odd-side saving) can save $\lg n$ bits by modifying $\Coll(x_{i+1},
x_{i+2})$; this works since $i+2 \le k$. If $i$ is even,
Lemma~\ref{lem:good-pig} (piggybacking) can give use $\Id(a)$ and
$\Id(b)$ at a negligible cost of $O(c \lg k)$ bits. This doesn't touch
anything later than $\Id(x_{i+1})$, where $i+1< k$.

When we arrive at $\Id(x_k)$, we know the position characters $K =
x_{<k} \cup a_1 \cup a_2$. This means that $x_k \Delta x_{k+1}
\subseteq K$, because $\hat{x}_k \setminus x_{k+1}$ is either empty or
a subset of $a_1$. Therefore, we can use weak-safe savings from
Lemma~\ref{lem:safe-weak} to code $\Id(x_k)$ in just $\frac{1}{3} \lg
n + O(c\lg k)$ bits. In total, we have save at least $\frac{4}{3}\lg n
- O(c\lg k)$ bits. \cqed


It remains to deal with distinct $t_1, t_2$, i.e.~$t_1 < t_2 \le
k$. If one of the numbers is even, it must be $t_2 = k$, and
then $t_1$ must be odd (due to case 2). By Case 3, if $t_j$ is odd, we
also know $a_j\setminus x_{< t_j-1}=\hat{x}_{t_j-1}\Delta\hat{x}_{t_j}$.
Since these cases need to deal with at least one odd $t_j$, the
following lemma will be crucial:

\begin{lemma}\label{lem:dbl-code} 
If $t_j\le k$ is odd and $a_j\setminus x_{< t_j-1} = \hat{x}_{t_j-1}
\Delta \hat{x}_{t_j}$, we can code $\Id(x_{t_j-1})$ and $\Id(x_{t_j})$
with $\frac{3}{2} \lg n + O(c\lg k)$ bits in total.
\end{lemma}

\begin{proof}
Consider the subkey $y=\hat{x}_{t_j-1}\setminus x_{t_j}$. We first
specify the positions of $y$ using $c$ bits. If $C(y)\ge \sqrt{n}$,
there are at most $\sqrt{n}$ possible choices of $y$, so we can
specify $y$ with $\frac{1}{2}\lg n$ bits. We can also identify
$x_{t_j}$ with $\lg n$ bits. Then $\Id(x_{t_j-1})$ requires $O(c\lg
k)$ bits, since $x_{t_j-1} \subseteq y\cup x_{t_j}\cup x_{<t_j-1}$.

If $C(y)\leq \sqrt{n}$, we first specify $\Id(x_{t_j-1})$ with $\lg n$
bits. This gives us the subkey $y\subseteq x_{t_j-1}$. Since
$a_j\setminus x_{< t_j-1} = \hat{x}_{t_j-1} \Delta \hat{x}_{t_j}$, it
follows that $y\subset a_j$. Thus, we can write $\Id(a_j)$ using $\lg
C(y) \le \lg \frac{1}{2} \lg n$ bits. Since $x_{t_j}\subseteq x_{\leq
  t_j-1}\cup a_j$, we get $\Id(x_{t_j})$ for an additional $O(c\lg k)$
bits.
\end{proof}

\case{Case 6: Both $t_1$ and $t_2$ are odd, $t_1<t_2<k$, and for all 
  $j\in \{1,2\}$, $a_j\setminus x_{< t_j-1} = \hat{x}_{t_j-1} \Delta
  \hat{x}_{t_j}$}
We apply Lemma \ref{lem:dbl-code} for both $j=1$ and $j=2$, and save
$\lg n$ bits in coding $\Id(x_{t_1-1})$, $\Id(x_{t_1})$,
$\Id(x_{t_2-1})$, and $\Id(x_{t_2})$. These are all distinct keys,
because $t_1<t_2$ and both are odd. Since $t_2<k$, we can combine this
with any safe saving.\cqed

\case{Case 7: $t_2=k$ is even and $t_1<k$ is odd with
  $a_1\setminus x_{< t_1-1}=\hat{x}_{t_1-1}\Delta\hat{x}_{t_1}$}
We apply Lemma \ref{lem:dbl-code} for $j=1$, and save $\frac{1}{2}\lg
n - O(c\lg k)$ bits in coding $\Id(x_{t_1-1})$, $\Id(x_{t_1})$.  We
also save $\lg n$ bits by modifying $\Hashes(h_0, x_k)$. We reveal all
hash codes of $\hat{x}_k$, except for one position-character $\alpha
\in a_2 \cap \hat{x}_k$ (which is a nonempty set since $t_2=k$). We
then specify $\Id(a_2)$, which takes $O(c\lg k)$ bits. The hash
$h_0(\alpha)$ can then be deduced from $h_0(a_2) = h_0(a_0)$. \cqed

\case{Case 8: Both $t_1$ and $t_2$ are odd, $t_1<t_2=k$, and for
  all $j\in \{1,2\}$, $a_j\setminus x_{< t_j-1} = \hat{x}_{t_j-1}
  \Delta \hat{x}_{t_j}$}
To simplify notation, let $t_1=t$. This case is the most difficult. If
we can apply strong-safe saving as in Lemma~\ref{lem:safe-strong}, we
save $\lg n$ by modifying $\Hashes(h_k, x_k)$. We also save $\lg n$ by
two applications of Lemma \ref{lem:dbl-code}, coding $\Id(x_{t-1})$,
$\Id(x_t)$, $\Id(x_{k-1})$, and $\Id(x_k)$. These don't interact since
$t < k$ and both are odd.

The strong-safe saving fails if $\hat{x}_k \subset x_{k+1}$. We will
attempt to piggyback for $x_k$ and $x_{k+1}$. Let $i$ be the largest
value such that $x_k\setminus x_{\le i} \ne x_{k+1} \setminus x_{\le
  i}$.  If $i$ is even, we get an odd-side saving of $\lg n$
(Lemma~\ref{lem:bad-pig}). Since this does not affect any identities,
we can still apply Lemma \ref{lem:dbl-code} to save $\frac{1}{2} \lg
n$ on the identities $\Id(x_{t-1})$ and $\Id(x_t)$.

Now assume $i$ is odd. We have real piggybacking, which may affect the
coding of $\Id(x_i)$, $\Id(x_{i+1})$ and $\Id(x_k)$. Since both $i$
and $t$ are odd, there is at most one common key between $\{x_i,
x_{i+1} \}$ and $\{ x_{t-1},x_t \}$. We consider two cases:
\begin{itemize*}
\item Suppose $x_{t-1} \notin \{x_i, x_{i+1} \}$. Let
  $y=\hat{x}_{t-1}\setminus x_t$. After piggybacking, which in
  particular encodes $x_t$, we can encode $\Id(x_{t-1})$ in $\lg
  \frac{n}{C(y)} + O(c\lg k)$ bits. Indeed, we can write the positions
  of $y$ with $c$ bits and then the identity of $y$ using Huffman coding
  for all subkeys on those positions. Finally the identity of
  $x_{t-1}$ can be written in $O(c\lg k)$ bits, since $x_{t-1}
  \subset x_{<t-1} \cup y \cup x_t$.

\item Suppose $x_t \notin \{x_i, x_{i+1} \}$. Let $y=\hat{x}_t
  \setminus x_{t-1}$. As above, we can write $\Id(x_t)$ using $\lg
  \frac{n}{C(y)} + O(c\lg k)$ bits, after piggybacking.
\end{itemize*}

If $C(y)\ge n^{1/3}$, we have obtained a total saving of $\frac{4}{3}
\lg n - O(c\lg k)$: a logarithmic term for $\Id(x_k)$ from
piggybacking, and $\frac{1}{3}\lg n$ for $\Id(x_{t-1})$ or $\Id(x_t)$.

Now assume that $C(y)\le n^{1/3}$. In this case, we do \emph{not} use
piggybacking. Instead, we use a variation of Lemma \ref{lem:dbl-code}
to encode $\Id(x_{t-1})$ and $\Id(x_t)$. First we code the one
containing $y$ with $\lg n$ bits. Since $a_1 \setminus x_{< t-1} =
\hat{x}_{t-1} \Delta \hat{x}_t$, and therefore $y\subset a_1$, we have
$y\subset a_1$. We code $\Id(a_1)$ with $\lg C(y) \le \frac{1}{3} \lg
n$ bits. We obtain the other key among $x_{t-1}$ and $x_t$ using
$O(c\lg k)$ bits, since all its characters are known. Thus we have
coded $\Id(x_{t-1})$ and $\Id(x_t)$ with $\frac{4}{3} \lg n + O(c\lg
k)$ bits, for a saving of roughly $\frac{2}{3} \lg n$ bits.

Next we consider the coding of $\Id(x_{k-1})$ and $\Id(x_k)$. We know
that $a_2 \setminus x_{<k-1} = \hat{x}_{k-1} \Delta \hat{x}_k$ and
$\hat{x}_k \subset x_{k+1}$. Lemma \ref{lem:dbl-code} would guarantee
a saving of $\frac{1}{2}\lg n$ bits. However, we will perform an
analysis like above, obtaining a saving of $\frac{2}{3}\lg n$ bits.

Let $y=\hat{x}_{k-1}\setminus x_k$. First assume $C(y)\ge n^{1/3}$.
We use the safe-weak saving of Lemma~\ref{lem:safe-weak} to encode
$\Id(x_k)$ using $\frac{2}{3}\lg n$ bits. We then encode the subkey
$y$ using $\lg \frac{n}{C(y)} + O(c) \le \frac{2}{3} \lg n + O(c)$
bits, and finally $x_{k-1}$ using $O(c\lg k)$ bits. This obtains both
$\Id(x_{k-1})$ and $\Id(x_k)$ using $\frac{4}{3}\lg n + O(c\lg k)$
bits.

Now assume $C(y)\le n^{1/3}$. We first code $\Id(x_{k-1})$ using $\lg
n$ bits. This gives us $y$ for the price of $c$ bits. But $a_2
\setminus x_{<k-1} = \hat{x}_{k-1} \Delta \hat{x}_k$, so $y \subset
a_2$, and we can code $\Id(a_2)$ using $\lg C(y) \le \frac{1}{3} \lg
n$ bits. Then $\Id(x_k)$ can be coded with $O(c\lg k)$ bits.  Again,
we obtain both $\Id(x_{k-1})$ and $\Id(x_k)$ for the price of
$\frac{4}{3}\lg n + O(c\lg k)$ bits.  \cqed

This completes our analysis of cuckoo hashing.

\section{Minwise Independence}

We will prove that:
\begin{equation} \label{eq:minwise}
\frac{1}{n} \cdot \left(1 - \frac{O(\lg n)}{n^{1/c}} \right) 
~\le~
\Pr[h(q) < \min h(X)] 
~\le~ \frac{1}{n} \cdot \left(1 + \frac{O(\lg^2 n)}{n^{1/c}} \right) 
\end{equation}
The lower bound is relatively simple, and is shown in
\S\ref{sec:mw-lower}. The upper bound is significantly more involved
and appears in \S\ref{sec:mw-upper}.

For the sake of the analysis, we divide the output range $[0,1)$ into
  $\frac{n}{\ell}$ bins, where $\ell = \gamma \lg n$ for a large
  enough constant $\gamma$. Of particular interest is the minimum bin
  $[0, \frac{\ell}{n})$. We choose $\gamma$ sufficiently large for the
    Chernoff bounds of Theorem~\ref{thm:chernoff} to guarantee that
    the minimum bin in non-empty w.h.p.: $\Pr[ \min h(X) <
      \frac{\ell}{n} ] \ge 1 - \frac{1}{n^2}$.

In \S\ref{sec:mw-lower} and \S\ref{sec:mw-upper}, we assume that hash
values $h(x)$ are binary fractions of infinite precision (hence, we
can ignore collisions). It is easy to see that \eqref{eq:minwise}
continues to hold when the hash codes have $(1+ \frac{1}{c})\lg n$
bits, even if ties are resolved adversarially. Let $\tilde{h}$ be a
truncation to $(1 + \frac{1}{c})\lg n$ bits of the infinite-precision
$h$. We only have a distinction between the two functions if $q$ is the
minimum and $(\exists) x\in S : \tilde{h}(x) = \tilde{h}(q)$. The
probability of a distinction is bounded from above by:
\[ \Pr \big[ \tilde{h}(q) \le \tfrac{\ell}{n} 
      ~\land~ (\exists) x\in S:\tilde{h}(x) = \tilde{h}(q) \big]
~\le~ \tfrac{\ell}{n} \cdot \big( n \cdot \tfrac{1}{n^{1+1/c}} \big)
~\le~ \tfrac{O(\lg n)}{n^{1+1/c}} \]
We used 2-independence to conclude that $\{ h(q) < \frac{\ell}{n} \}$
and $\{ \tilde{h}(x) = \tilde{h}(q) \}$ are independent.

Both the lower and upper bounds start by expressing:
\[
 \Pr[h(q) < \min h(S)] = \int_0^1 f(p) \mathrm{d}p,
         \qquad \textrm{where } 
f(p) = \Pr[p < \min h(S) \mid h(q) = p].
\]
For truly random hash functions, $\Pr[p < \min h(S) \mid h(q) = p] =
(1-p)^n$, since each element has an independent probability of $1-p$
of landing about $p$.

\subsection{Lower bound}   \label{sec:mw-lower}

For a lower bound, it suffices to look at the case when $q$ lands in
the minimum bin:
\[
 \Pr[h(q) < \min h(S)] \ge \int_0^{\ell/n} f(p) \mathrm{d}p, \qquad
    \textrm{where } f(p) = \Pr[p < \min h(S) \mid h(q) = p]
\]

We will now aim to understand $f(p)$ for $p \in [0, \frac{\ell}{n}]$.
In the analysis, we will fix the hash codes of various
position-characters in the order $\prec$ given by
Lemma~\ref{lem:group-size-q}. Let $h(\prec\!\alpha)$ done the choice for
all position-characters $\beta \prec \alpha$.

Remember that $\prec$ starts by fixing the characters of $q$ first,
so: $q_1 \prec \cdots \prec q_c \prec \alpha_0 \prec \alpha_1 \prec
\cdots$ Start by fixing $h(q_1), \dots, h(q_c)$ subject to $h(q) = x$.

When it is time to fix some position-character $\alpha$, the hash code
of any key $x\in G_\alpha$ is a constant depending on $h(\prec\!
\alpha)$ {\bf xor} the random quantity $h(\alpha)$. This final xor
makes $h(x)$ uniform in $[0,1)$. Thus, for any choice of $h(\prec\!
  \alpha)$, $\Pr[ h(z) < p \mid h(\prec\! \alpha) ] = p$. By the union
  bound, $\Pr[p < \min h(G_\alpha) \mid h(\prec\! \alpha)] \ge 1 - p
  \cdot |G_\alpha|$. This implies that:
\begin{equation}\label{eq:lower-prod}
f(p) ~=~ \Pr[p < \min h(S) \mid h(q) = p] 
~\ge~ \prod_{\alpha \succ q_c} (1- p \cdot |G_\alpha|).
\end{equation}

To bound this product from below, we use the following lemma:
\begin{lemma}\label{lem:Gerber}
Let $p \in [0,1]$ and $k \ge 0$, where $p\cdot k \le \sqrt{2} - 1$. Then 
$1- p\cdot k > (1-p)^{(1+pk) k}$.
\end{lemma}
\begin{proof*} 
First we note a simple proof for the weaker statement
$(1-pk)<(1-p)^{\lceil (1+pk)k \rceil}$. However, it will be crucial
for our later application of the lemma that we can avoid the ceiling.

Consider $t$ Bernoulli trials, each with success probability $p$.  The
probability of no failures occurring is $(1-p)^t$. By the
inclusion-exclusion principle, applied to the second level, this is
bounded from above by:
\[ (1-p)^t ~\le~ 1 - t\cdot p + \binom{t}{2} p^2 
~<~ 1 - (1 - \tfrac{pt}{2}) t \cdot p \]
Thus, $1-kp$ can be bounded from below by the probability that no
failure occurs amount $t$ Bernoulli trials with success probability
$p$, for $t$ satisfying $t \cdot (1 - \frac{pt}{2}) \ge k$. This holds
for $t \ge (1+kp) k$.

We have just shown $1- p\cdot k > (1-p)^{\lceil (1+pk) k
  \rceil}$. Removing the ceiling requires an ``inclusion-exclusion''
inequality with a non-integral number of experiments $t$. Such an
inequality was shown by Gerber~\cite{gerber68bernoulli}: $(1-p)^t \le
1-\alpha t + (\alpha t)^2/2$, even for fractional $t$. Setting $t =
(1+pk)k$, our result is a corollary of Gerber's inequality:
\begin{eqnarray*}
(1-p)^{t} &\le& 1 - pt + \tfrac{(pt)^2}{2}
~=~ 1- p (1+pk)k+ \tfrac{1}{2}(p (1+pk)k)^2 \\
&=&1- pk - (1-\tfrac{(1+pk)^2}{2})(pk)^2 
~\le~ 1-pk. \hfill \qed
\end{eqnarray*}
\end{proof*}

The lemma applies in our setting, since $p < \frac{\ell}{n} =
O(\frac{\lg n}{n})$ and all groups are bounded $|G_\alpha| \le 2\cdot
n^{1-1/c}$. Note that $p \cdot |G_\alpha| \le \frac{\ell}{n} \cdot 2
n^{1-1/c} = O(\ell / n^{1/c})$. Plugging into \eqref{eq:lower-prod}:
\[
f(p) ~\ge~ \prod_{\alpha \succ q_c} (1- p \cdot |G_\alpha|)
~\ge~ \prod_{\alpha \succ q_c} (1 - p)^{|G_\alpha| (1 + \ell / n^{1/c})}
~\ge~ (1-p)^{n\cdot (1+\ell/n^{1/c})}.
\]
Let $m = n\cdot (1 + \ell / n^{1/c})$. The final result follows by
integration over $p$:
\begin{eqnarray*}
\Pr[h(q) < \min h(S)] &\ge& \int_0^{\ell/n} f(p) \mathrm{d}p
 \ge \int_{0}^{\ell/n} (1-p)^m \mathrm{d}p \\
&=& \left. \frac{-(1-p)^{m+1}}{m+1}~ \right|_{p=0}^{\ell/n}
~=~ \frac{1-(1-\ell/n)^{m+1}}{m+1} 
\\ &> & \frac{1-e^{-\ell}}{m+1}
~>~ \frac{1-1/n}{n(1+\ell/n^{1/c})}
~=~ \frac{1}{n} \cdot \left( 1 - \frac{O(\lg n)}{n^{1/c}} \right)
\end{eqnarray*}

\subsection{Upper bound}   \label{sec:mw-upper}

As in the lower bound, it will suffice to look at the case when $q$
lands in the minimum bin:
\[ \Pr[ h(q) < h(S)] ~\le~  \Pr[ \min h(S) \ge \tfrac{\ell}{n}]
+ \Pr[ h(q) < h(S) \land h(q) < \tfrac{\ell}{n}]
~\le~ \tfrac{1}{n^2} + \int_0^{\ell/n} f(p) \mathrm{d}p
\]

To bound $f(p)$, we will fix position-characters in the order $\prec$
from Lemma~\ref{lem:group-size-q}, subject to $h(q) = p$. In the lower
bound, we could analyze the choice of $h(\alpha)$ even for the
worst-case choice of $h(\prec\! \alpha)$. Indeed, no matter how the
keys in $G_\alpha$ arranged themselves, when shifted randomly by
$h(\alpha)$, they failed to land below $p$ with probability $1 -
p|G_\alpha| \ge (1-p)^{(1+o(1)) |G_\alpha|}$.

For an upper bound, we need to prove that keys from $G_\alpha$ do land
below $p$ often enough: $\Pr[p < \min h(G_\alpha) \mid
  h(\prec\!\alpha) ] \le (1-p)^{(1-o(1)) |G_\alpha|}$. However, a
worst-case arrangement of $G_\alpha$ could make all keys equal, which
would give the terrible bound of just $1-p$.

To refine the analysis, we can use Lemma~\ref{lem:manybins}, which
says that for $d = O(1)$, all groups $G_\alpha$ are $d$-bounded with
probability $\ge 1 - \frac{1}{n^2}$. If $G_\alpha$ is $d$-bounded, its
keys cannot cluster in less than $\lceil |G_\alpha| / d \rceil$
different bins. 

When a group $G_\alpha$ has more than one key in some bin, we pick one
of them as a \emph{representative}, by some arbitrary (but fixed)
tie-breaking rule. Let $R_\alpha$ be the set of representatives of
$G_\alpha$. Observe that the set $R_\alpha \subseteq G_\alpha$ is
decided once we condition on $h(\prec\! \alpha)$. Indeed, the hash codes
for keys in $G_\alpha$ are decided up to a shift by $h(\alpha)$, and
this common shift cannot change how keys cluster into bins.  We
obtain:
\[ \Pr[p < \min h(G_\alpha)\mid h(\prec\!\alpha)]
 ~\le~ \Pr[p < \min h(R_\alpha) \mid h(\prec\! \alpha)] 
~=~ 1- p |R_\alpha| ~\le~ (1-p)^{|R_\alpha|}
\]
To conclude $\Pr[p < \min h(R_\alpha)] = 1 - p|R_\alpha|$ we used that
the representatives are in different bins, so at most one can land
below $p$. Remember that $|R_\alpha|$ is a function of $h(\prec\!
\alpha)$. By $d$-boundedness, $|R_\alpha| \ge |G_\alpha|/d$, so we get
$\Pr[p < \min h(G_\alpha)\mid h(\prec\!\alpha)] \le (1-p)^{|G_\alpha|
  / d}$ for almost all $h(\prec\!  \alpha)$. Unfortunately, this is a
far cry from the desired exponent, $|G_\alpha| \cdot \big( 1 -
\widetilde{O}(n^{-1/c}) \big)$.

To get a sharper bound, we will need a dynamic view of the
representatives. After fixing $h(\prec\! \alpha)$, we know whether two
keys $x$ and $y$ collide whenever the symmetric difference $x\Delta y
= (x \setminus y) \cup (y \setminus x)$ consists only of
position-characters $\prec \alpha$. Define $R_\beta(\alpha)$ to be our
understanding of the representatives $R_\beta$ just before character
$\alpha$ is revealed: from any subset of $G_\beta$ that is known to
collide, we select only one key. After the query characters get
revealed, we don't know of any collisions yet (we know only one
character per position), so $R_\beta(\alpha_0) = G_\beta$. The set of
representatives decreases in time, as we learn about more collisions,
and $R_\beta(\beta) = R_\beta$ is the final value (revealing $\beta$
doesn't change the clustering of $G_\beta$).

Let $C(\alpha)$ be the number of key pairs $(x,y)$ from the same group
$G_\beta$ ($\beta \succ \alpha$) such that $\alpha = \max_\prec (x
\Delta y)$. These are the pairs whose collisions is decided when
$h(\alpha)$ is revealed, since $h(\alpha)$ is the last unknown hash
code in the keys, besides the common ones. Let $\alpha^+$ be the
successor of $\alpha$ in the order $\prec$. Consider the total number
of representatives before and after $h(\alpha)$ is revealed:
$\sum_\beta |R_\beta(\alpha)|$ versus $\sum_\beta
|R_\beta(\alpha^+)|$. The maximum change between these quantities is
$\le C(\alpha)$, while the expected change is $\le C(\alpha) \cdot
\frac{\ell}{n}$. This is because $h(\alpha)$ makes every pair $(x,y)$
collide with probability $\frac{\ell}{n}$, regardless of the previous
hash codes in $(x\Delta y) \setminus \{\alpha\}$. Note, however, that
the number of colliding pairs may overestimate the decrease in the
representatives if the same key is in multiple pairs.

Let $n(\succ\! \alpha) = \sum_{\beta \succ \alpha} |G_\alpha|$ and
define $n(\succeq\! \alpha)$ simmilarly. Our main inductive claim is:
\begin{lemma} \label{lem:mess}
For any setting $h(\prec\! \alpha)$ such that $h(q)=p$ and
$\sum_{\beta \succeq \alpha} |R_\beta (\alpha)| = r$, we have:
\[
\Pr\Big[ \Big( p < \min \bigcup_{\beta \succeq \alpha} h(G_\beta) \Big)
         \land (\forall \alpha) G_\alpha~d\textrm{-bounded}
  ~\big|~ h(\prec\!\alpha) \Big] ~\le~ P(\alpha, p, r)
\]
where we define $\displaystyle P(\alpha, p, r)
= (1-p)^r + (1-p)^{n(\succeq \alpha) / (2d)} \cdot
     \sum_{\beta \succeq \alpha} \frac{4C(\beta)\cdot (\ell/n)}
         {n(\succ\!\beta) / d}$.
\end{lemma}

As the definition $P(\alpha, p, r)$ may look intimidating, we first
try to demystify it, while giving a sketch for the lemma's proof (the
formal proof appears in \S\ref{sec:full-mess}.) The
lemma looks at the worst-case probability, over prior choices
$h(\prec\!\alpha)$, that $p = h(q)$ remains the minimum among groups
$G_\alpha, G_{\alpha^+}, \dots$. After seeing the prior hash codes,
the number of representatives in these groups is $r = \sum_{\beta
  \succeq \alpha} |R_\beta(\alpha)|$. In the ideal case when
$h(\alpha), h(\alpha^+), \dots$ do not introduce any additional
collisions, we have $r$ representatives that could beat $p$ for the
minimum. As argued above, the probability that $p$ is smaller than all
these representatives is $\le (1-p)^r$. Thus, the first term of
$P(\alpha, p, r)$ accounts for the ideal case when no more collisions
occur.

On the other hand, the factor $(1-p)^{n(\succeq \alpha) / (2d)}$
accounts for the worst case, with no guarantee on the representatives
except that the groups are $d$-bounded (the 2 in the exponent is an
artifact). Thus, $P(\alpha, p, r)$ interpolates between the best case
and the worst case. This is explained by a convexity argument: the
bound is maximized when $h(\alpha)$ mixes among two extreme strategies
--- it creates no more collisions, or creates the maximum it could.

It remains to understand the weight attached to the worst-case
probability. After fixing $h(\alpha)$, the maximum number of remaining
representatives is $\hat{r} = \sum_{\beta \succ \alpha}
|R_\beta(\alpha)|$.  The expected number is $\ge \hat{r} - C(\alpha)
\frac{\ell}{n}$, since every collision happens with probability
$\frac{\ell}{n}$. By a Markov bound, the worst case (killing most
representatives) can only happen with probability $O\big(
\frac{\ell}{n} C(\alpha) \big/ \hat{r} \big)$. The weight of the worst
case follows by $\hat{r} \ge n(\succ\! \alpha) / d$ and letting these
terms accrue in the induction for $\beta \succ \alpha$.

\paragraph{Deriving the upper bound.}
We now prove the upper bound on $\Pr[h(q) < h(S)]$ assuming
Lemma~\ref{lem:mess}. Let $\alpha_0$ be the first position-character
fixed after the query. Since fixing the query cannot eliminate
representatives,
\[ \Pr[ p < \min h(S) \land (\forall \alpha) G_\alpha
  ~d\textrm{-bounded} \mid h(q) = p ] \le P(\alpha_0, p, n) \]

\begin{lemma}
$P(\alpha_0, p, n) \le (1-p)^n + (1-p)^{n/(2d)} \cdot \frac{O(\lg^2 n)}{n^{1/c}}$.
\end{lemma}

\begin{proof}
We will prove that $A = \sum_{\beta \succ \alpha_0}
\frac{C(\beta)}{n(\succ\! \beta)} \le n^{1-1/c} \cdot H_n$, where
$H_n$ is the Harmonic number. 

Consider all pairs $(x,y)$ from the same group $G_\gamma$, and order
them by $\beta = \max_\prec (x \Delta y)$. This is the time when the
pair gets counted in some $C(\beta)$ as a potential collision. The
contribution of the pair to the sum is $1 / n(\succ\! \beta)$, so this
contribution is maximized if $\beta$ immediately precedes $\gamma$ in
the order $\prec$. That is, the sum is maximized when $C(\beta) =
\binom{|G_{\beta^+}|}{2}$. We obtain $A \le \sum_\beta \frac{|G_\beta|^2}{2} /
n(\succeq\!\beta) \le n^{1-1/c} \cdot \sum_\beta |G_\beta| /
n(\succeq\!\beta)$.  In this sum, each key $x \in G_\beta$ contributes
$1 / n(\succeq\!\beta)$, which is bounded by one over the number of
keys following $x$. Thus $A \le H_n$.
\end{proof}

To achieve our original goal, bounding $\Pr[h(q) < h(S)]$, we proceed
as follows:
\begin{eqnarray*}
\Pr[ h(q) < h(S)]
&\le& \tfrac{1}{n^2} + 
   \int_0^{\ell/n} \Pr[ p < \min h(S) \mid h(q)=p] \mathrm{d}p \\
&\le& \tfrac{1}{n^2} + \Pr[ (\exists \alpha) G_\alpha \textrm{ not $d$-bounded}]
 + \int_0^{\ell/n} P(\alpha_0, p, n) \mathrm{d}p
\end{eqnarray*}
By Lemma~\ref{lem:manybins}, all groups are $d$-bounded with
probability $1 - \frac{1}{n^2}$. We also have
\[
\int_{0}^{\ell/n} (1-p)^n \mathrm{d}p = 
\left. \frac{-(1-p)^{n+1}}{n+1} \right|_{p=0}^{\ell/n}
\le \frac{1}{n+1}
\]
Thus: 
\[
\Pr[ h(q) < h(S)] \le \frac{O(1)}{n^2} +
\frac{1}{n+1} + \frac{1}{n/(2d) + 1} \cdot \frac{O(\lg^2 n)}{n^{1/c}}
= \frac{1}{n} \cdot \left( 1 + \frac{O(\lg^2 n)}{n^{1/c}} \right).
\]

\subsection{Proof of Lemma~\ref{lem:mess}}   \label{sec:full-mess}

Recall that we are fixing some choice of $h(\prec\!\alpha)$ and
bounding:
\[ A = \Pr \big[ p < \min \bigcup_{\beta \succeq \alpha} h(G_\beta)
  ~\land~ (\forall \alpha) G_\alpha~d\textrm{-bounded} \mid
  h(\prec\!\alpha) \big] 
\]
If for some $\beta$, $|R_\beta(\alpha)| < |G_\beta| / d$, it means not
all groups are $d$-bounded, so $A=0$. If all groups are $d$-bounded
and we finished fixing all position-characters, $A=1$. These form the
base cases of our induction.

The remainder of the proof is the inductive step. We first break the
probability into:
\[ 
A_1\cdot A_2 = 
\Pr \big[ p < \min h(G_\alpha) \mid h(\prec\!\alpha) \big]
\cdot \Pr \big[\bigcup_{\beta \succ \alpha} h(G_\beta)
  ~\land~ (\forall \alpha) G_\alpha~d\textrm{-bounded} 
\mid h(\prec\!\alpha), p > \min h(G_\alpha) \big]
\]

As $h(\alpha)$ is uniformly random, each representative in $R_\alpha$
has a probability of $p$ of landing below $p$. These events are
disjoint because $p$ is in the minimum bin, so $A_1 = 1 - p\cdot
|R_\alpha| \le (1-p)^{|R_\alpha|}$.

After using $R_\alpha$, we are left with $\hat{r} = r - |R_\alpha| =
\sum_{\beta \succ \alpha} |R_\beta(\alpha)|$ representatives. After
$h(\alpha)$ is chosen, some of the representative of $\hat{r}$ are
lost. Define the random variable $\Delta = \sum_{\beta \succ \alpha}
\big( |R_\beta (\alpha)| - |R_\beta(\alpha^+)| \big)$ to measure this
loss. 

Let $\Delta^{\max} \ge \hat{r} - \frac{n(\succ\!\alpha)}{d}$ be a
value to be determined. We only need to consider $\Delta \le
\Delta^{\max}$. Indeed, if more than $\Delta^{\max}$ representatives
are lost, we are left with less than $n(\succ\!\alpha) / d$
representatives, so some group is not $d$-bounded, and the probability
is zero. We can now bound $A_2$ by the induction hypothesis:
\[
A_2 \le \sum_{\delta = 0}^{\Delta^{\max}} \Pr[\Delta = \delta 
\mid h(\prec\!\alpha), p > \min h(G_\alpha) \big]
\cdot P(\alpha^+, p, \hat{r} - \delta)
\]
where we had $\displaystyle P(\alpha^+, p, \hat{r} - \delta)
= (1-p)^{\hat{r} - \delta} + (1-p)^{n(\succ\! \alpha) / (2d)} \cdot
     \sum_{\beta \succ \alpha} \frac{4C(\beta)\cdot (\ell/n)}
         {n(\succ\!\beta) / d}$.

Observe that the second term of $P(\alpha^+, p, \hat{r} - \delta)$
does not depend on $\delta$ so:
\[
A_2 \le A_3 + (1-p)^{n(\succ\! \alpha) / (2d)} \cdot
     \sum_{\beta \succ \alpha} \frac{4C(\beta)\cdot (\ell/n)}
         {n(\succ\!\beta) / d} \]
where $\displaystyle A_3 = 
\sum_{\delta = 0}^{\Delta^{\max}} \Pr[\Delta = \delta 
\mid h(\prec\!\alpha), p > \min h(G_\alpha) \big]
\cdot (1-p)^{\hat{r} - \delta}$.

It remains to bound $A_3$. We observe that $(1-p)^{\hat{r} - \delta}$
is convex in $\delta$, so its achieves the maximum value if all the
probability mass of $\Delta$ is on $0$ and $\Delta^{\max}$, subject to
preserving the mean.

\begin{observation}
We have:
$\E[\Delta \mid h(\prec\!\alpha), p > \min h(G_\alpha)] \le 2\cdot
  C(\alpha) \cdot \frac{\ell}{n}$.
\end{observation}

\begin{proof}
As discussed earlier, a representative disappears when we have a pair
$x,y\in R_\beta(\alpha)$ that lands in the same bin due to
$h(\alpha)$. This can only happen if $(x,y)$ is counted in
$C(\alpha)$, i.e.~$\alpha = \max_\prec (x \Delta y)$. If $h(\alpha)$
is uniform, such a pair $(x,y)$ collides with probability
$\frac{\ell}{n}$, regardless of $h(\prec\!\alpha)$. By linearity of
expectation $\E[\Delta \mid h(\prec\!\alpha)] \le C(\alpha) \cdot
\frac{\ell}{n}$.

However, we have to condition on the event $p > \min h(G_\alpha)$,
which makes $h(\alpha)$ non-uniform. Since $p < \frac{\ell}{n}$ and
$|G_\alpha| \le n^{1-1/c}$, we have $\Pr[p< \min h(G_\alpha)] <
1/2$. Therefore, conditioning on this event can at most double the
expectation of positive random variables.
\end{proof}

A bound on $A_3$ can be obtained by assuming $\Pr[\Delta =
  \Delta^{\max}] = \big( 2\cdot C(\alpha) \cdot \frac{\ell}{n} \big)
\big/ \Delta^{\max}$, and all the rest of the mass is on $\Delta = 0$.
This gives:
\[
 A_3 \le (1-p)^{\hat{r}}  + \frac{2\cdot C(\alpha) \cdot (\ell/n)}{\Delta^{\max}}
\cdot (1-p)^{\hat{r} - \Delta^{\max}}
\]
Remember that we promised to choose $\Delta^{\max} \ge \hat{r} -
\frac{n(\succ\!\alpha)}{d}$. We now fix $\Delta^{\max} = \hat{r} -
\frac{n(\succ\!\alpha)}{2d}$. We are guaranteed that $\hat{r} \ge
\frac{n(\succ\!\alpha)}{d}$, since otherwise some group is not
$d$-bounded. This means $\Delta^{\max} \ge
\frac{n(\succ\!\alpha)}{2d}$. We have obtained a bound on $A_3$:
\begin{eqnarray*}
 A_3 &\le& (1-p)^{\hat{r}}  + \frac{2\cdot C(\alpha) \cdot (\ell/n)}
   {n(\succ\!\alpha) / (2d)} \cdot (1-p)^{n(\succ\!\alpha) / (2d)} \\
\Longrightarrow\quad A_2 &\le&
(1-p)^{\hat{r}} + (1-p)^{n(\succ\! \alpha) / (2d)} \cdot
     \sum_{\beta \succeq \alpha} \frac{4C(\beta)\cdot (\ell/n)}
         {n(\succ\!\beta) / d} \\
\Longrightarrow\quad A &\le&
(1-p)^{|R_\alpha|} \cdot (1-p)^{r - |R_\alpha|} + (1-p)^{n(\succ\! \alpha) / (2d)} \cdot
     \sum_{\beta \succeq \alpha} \frac{4C(\beta)\cdot (\ell/n)}
         {n(\succ\!\beta) / d}
\end{eqnarray*}
This completes the proof of Lemma~\ref{lem:mess}, and the bound on
minwise independence.

\section{Fourth Moment Bounds}   \label{sec:moment}

Consider distributing a set $S$ of $n$ balls into $m$ bins truly
randomly.  For the sake of generality, let each element have a weight
of $w_i$. We designate a query ball $q \notin S$, and let $W$ be the
total weight of the elements landing in bin $F(h(q))$, where $F$ is an
arbitrary function. With $\mu = \E[W] = \frac{1}{m}\sum w_i$, we are
interested in the 4\th moment of the bin size: $\E[(W-\mu)^4]$.

Let $X_i$ be the indicator that ball $i \in S$ lands in bin $F(h(q))$,
and let $Y_i = X_i - \frac{1}{m}$. We can rewrite $W - \mu = \sum_i
Y_i w_i$, so:
\begin{equation}  \label{eq:terms}
\E[(W-\mu)^4] ~=~ \sum_{i,j,k,l \in S} w_i w_j w_k w_l\cdot \E[ Y_i Y_j
  Y_k Y_l ].
\end{equation}

The terms in which some element appears exactly once are zero. Indeed,
if $i \notin \{j,k,l\}$, then $\E[ Y_iY_j Y_k Y_l ] = \E[Y_i] \cdot
\E[Y_j Y_kY_l ]$, which is zero since $\E [ Y_i] = 0$.  Thus, the only
nonzero terms arise from:
\begin{itemize*}
\item four copies of one element ($i=j=k=l$), giving the term
  $\big( \frac{1}{m} \pm O(\frac{1}{m^2}) \big) w_i^4$.

\item two distinct elements $s\ne t$, each appearing twice. There are
  $\binom{4}{2} = 6$ terms for each $s,t$ pair, and each term is
  $O(\frac{1}{m^2}) w_s^2 w_t^2$.
\end{itemize*}

\noindent
This gives the standard 4\th moment bound:
\begin{equation}   \label{eq:ideal-moment}
\E\big[ (W-\mu)^4 \big] ~=~ 
\frac{1}{m} \sum_i w_i^4 + \frac{O(1)}{m^2} \big( \sum_i w_i^2 \big)^2.
\end{equation}

This bound holds even if balls are distributed by 5-independent
hashing: the balls in any 4-tuple hit the bin chosen by $h(q)$
independently at random. On the other hand, with 4-independent
hashing, this bound can fail quite badly~\cite{patrascu10kwise-lb}.

If the distribution of balls into bins is achieved by simple
tabulation, we will show a slightly weaker version of
\eqref{eq:ideal-moment}:
\begin{equation} \label{eq:moment-simple}
\E\big[ (W-\mu)^4 \big] ~=~ 
\frac{1}{m} \sum_i w_i^4 + O\left( \frac{1}{m^2} +  \frac{4^c}{m^3} \right) 
				\cdot \Big( \sum_i w_i^2 \Big)^2.
\end{equation}

In \S\ref{sec:fixed-bin}, we show how to analyze the 4\th moment of a
fixed bin (which requires 4-independence by standard techniques). Our
proof is a combinatorial reduction to Cauchy--Schwarz. In
\S\ref{sec:query-bin}, we let the bin depend on the hash code
$h(q)$. This requires 5-independence by standard techniques. To handle
tabulation hashing, \S\ref{sec:five-indep} shows a surprising result:
among any 5 keys, at least one hashes independently of the rest.

We note that the bound on the 4\th moment of a fixed bin has been
indendently discovered by \cite{braverman10kwise} in a different
context. However, that work is not concerned with a query-dependent
bin, which is the most surprising part of our proof.

\subsection{Fourth Moment of a Fixed Bin}   \label{sec:fixed-bin}

We now attempt to bound the terms of \eqref{eq:terms} in the case of
simple tabulation. Since simple tabulation is 3-independent
\cite{wegman81kwise}, any terms that involve only 3 distinct keys
(i.e.~$|\{i,j,k,l\}| \le 3$) have the same expected value as
established above. Thus, we can bound:
\[ 
\E[(W-\mu)^4] ~=~ \frac{1}{m} \sum_i w_i^4 + 
\frac{O(1)}{m^2} \big( \sum_i w_i^2 \big)^2
~+~ \sum_{i \ne j \ne k \ne l} w_i w_j w_k w_l \cdot \E[ Y_i Y_j Y_k Y_l ].
\]
Unlike the case of 4-independence, the contribution from distinct
$i,j,k,l$ will not be zero. We begin with the following simple bound
on each term:

\begin{claim}   \label{clm:bcube}
For distinct $i,j,k,l$, $\E[ Y_i Y_j Y_k Y_l ] ~=~ O(\frac{1}{m^3})$.
\end{claim}

\begin{proof}
We are looking at the expectation of $Z = (X_i - \frac{1}{m}) (X_j -
\frac{1}{m}) (X_k - \frac{1}{m}) (X_l- \frac{1}{m})$. Note that $Z$ is
only positive when an \emph{even} number of the four $X$'s are 1:
\begin{enumerate*}
\item the case $X_i = X_j = X_k = X_l = 1$ only happens with
  probability $\frac{1}{m^3}$ by 3-independence. The contribution to
  $Z$ is $(1 - \frac{1}{m})^4 < 1$.

\item the case of two 1's and two 0's happens with probability at most
  $\binom{4}{2} \frac{1}{m^2}$, and contributes $\frac{1}{m^2} (1 -
  \frac{1}{m})^2 < \frac{1}{m^2}$ to $Z$.

\item the case of $X_i = X_j = X_k = X_l = 0$ contributes
  $\frac{1}{m^4}$ to $Z$.
\end{enumerate*}
\noindent
Thus, the first case dominates and $\E[Z] = O(\frac{1}{m^3})$.
\end{proof}

If one of $\{i,j,k,l\}$ contains a unique position-character, its hash
code is independent of the other three. In this case, the term is
zero, as the independent key factors out of the expectation and $\E[
  Y_i] = 0$. We are left with analyzing 4-tuples with no unique
position-characters; let $A \subseteq S^4$ contain all such 4-tuples. Then:
\[
\sum_{i \ne j \ne k \ne l} w_i w_j w_k w_l\cdot \E[ Y_i Y_j Y_k Y_l ]
~=~ O\left( \tfrac{1}{m^3} \right) \cdot 
        \sum_{(i,j,k,l) \in A} w_i w_j w_k w_l.
\]

Imagine representing a tuple from $A$ as a $4 \times q$ matrix, with
every key represented in a row. There are four types of columns that
we may see: columns that contain a single character in all rows (type
1), and columns that contain two distinct characters, each appearing
in two rows (type $j\in \{2,3,4\}$ means that row $j$ has the same
character as row 1). According to this classification, there are $4^q$
possible matrix types.

\begin{claim}    \label{clm:onetype}
Fix a fixed matrix type, and let $B \subseteq A$ contain all tuples
conforming to this type. Then, $\sum_{(i,j,k,l) \in B} w_i w_j w_k w_l
\le \big(\sum_i w_i^2 \big)^2$.
\end{claim}

\begin{proof}
We first group keys according to their projection on the type-1
characters. We obtain a partition of the keys $S = S_1 \cup S_2
\cup \cdots$ such that $S_t$ contains keys that are identical in
the type-1 coordinates. Tuples that conform to the fixed matrix type,
$(i,j,k,l) \in B$, must consist of four keys from the same set,
i.e.~$i,j,k,l \in S_t$. Below, we analyze each $S_t$ separately and
bound the tuples from $(S_t)^4$ by $\big( \sum_{i\in S_t} w_i^2
\big)^2$. This implies the lemma by convexity, as $\sum_t \left(
\sum_{i\in S_t} w_i^2 \right)^2 ~\le~ \left( \sum_i w_i^2 \right)^2$.

For the remainder, fix some $S_t$. If $|S_t| < 4$, there is nothing to
prove. Otherwise, there must exist at least one character of type
different from 1, differentiating the keys. By permuting the set
$\{i,j,k,l\}$, we may assume a type-2 character exists.  Group
keys according to their projection on all type-2 characters. We
obtain a partition of the keys $S_t = T_1 \cup T_2 \cup \cdots$
such that $T_a$ contains keys that are identical in the type-2
coordinates.

A type-conforming tuple $(i,j,k,l) \in B$ must satisfy $i,j \in T_a$
and $k,l \in T_b$ for $a \ne b$. We claim a stronger property: for any
$i,j \in T_a$ and every $b \ne a$, there exists at most one pair $k,l
\in T_b$ completing a valid tuple $(i,j,k,l) \in B$. Indeed, for
type-1 coordinates, $k$ and $l$ must be identical to $i$ on that
coordinate. For type 3 and 4 coordinates, $k$ and $l$ must reuse the
characters from $i$ and $j$ ($k \gets i, l \gets j$ for type 3; $k
\gets j, l \gets i$ for type 4).

Let $X \subset (T_a)^2$ contain the pairs $i,j \in T_a$ which can be
completed by one pair $k,l \in T_b$. Let $Y \subset (T_b)^2$ contain
the pairs $k,l \in T_b$ which can be completed by $i,j\in T_a$. There
is a bijection between $X$ and $Y$; let it be $f : X \mapsto Y$.
We can now apply the Cauchy-Schwarz inequality:
\begin{eqnarray*}
\sum_{(i,j,k,l) \in B \,\cap\, (T_a \times T_b)} w_i w_j w_k w_l 
&=& \sum_{(i,j) \in X,~ (k,l) = f(i,j)} (w_i w_j)\cdot (w_k w_l) \\
&\le& \sqrt{ \left( \sum_{(i,j) \in X} (w_i w_j)^2 \right) 
             \left( \sum_{(k,l) \in Y} (w_k w_l)^2 \right) }
\end{eqnarray*}
But $\sum_{(i,j) \in X} w_i^2 w_j^2 \le \big( \sum_{i\in T_a} w_i^2
\big)^2$. Thus, the equation is further bounded by $\big( \sum_{i\in
  T_a} w_i^2 \big) \big( \sum_{k\in T_b} w_k^2 \big)$.

Summing up over all $T_a$ and $T_b$, we obtain:
\[
\sum_{(i,j,k,l) \in B \cap (S_t)^4} w_i w_j w_k w_l
~\le~ \sum_{a,b} \left( \sum_{i\in T_a} w_i^2 \right) 
                \left( \sum_{k\in T_b} w_k^2 \right)
~\le~ \left( \sum_{i\in S_t} w_i^2 \right)^2
\]
This completes the proof of the claim.
\end{proof}

The bound of Claim~\ref{clm:onetype} is multiplied by $4^q$, the
number of matrix types. We have thus shown \eqref{eq:moment-simple}.

\subsection{Fourth Moment of a Query-Dependent Bin}  \label{sec:query-bin}

We now aim to bound the 4\th moment of a bin chosen as a function $F$
of $h(q)$, where $q$ is a designated query ball. This requires dealing
with 5 keys ($i,j,k,l$ and the query $q$). Even though simple
tabulation is only 3-independent, we will prove the following
intriguing independence guarantee in \S\ref{sec:five-indep}:

\begin{theorem}\label{thm:5th-indep} 
With simple tabulation, in any fixed set of 5 distinct keys, there is
a key whose hash is independent of the other 4 hash codes.
\end{theorem}

As a side note, we observe that this theorem essentially implies that
{\em any\/} 4-independent tabulation based scheme is also
5-independent. In particular, this immediately shows the
5-independence of the scheme from \cite{thorup04kwise} (which augments
simple tabulation with some derived characters). This fact was already
known~\cite{thorup10kwise}, albeit with a more complicated proof.

In the remainder of this section, we use Theorem~\ref{thm:5th-indep}
to derive the 4\th moment bound \eqref{eq:moment-simple}. As before,
we want to bound terms $w_i w_j w_k w_l \cdot \E[ Y_i Y_j Y_k Y_l]$
for all possible configurations of $(i,j,k,l)$. Remember that $q\notin
S$, so $q \notin \{i,j,k,l\}$. These terms can fall in one of the
following cases:
\begin{itemize}
\item All keys are distinct, and $q$ hashes independently. Then, the
  contribution of $i,j,k,l$ to bin $F(h(q))$ bin is the same as to any
  fixed bin.

\item All keys are distinct, and $q$ is dependent. Then, at least one
  of $\{i,j,k,l\}$ must be independent of the rest and $q$; say it is
  $i$. But then we can factor $i$ out of the product: $\E[ Y_i Y_j Y_k
    Y_l ] = \E[Y_i] \cdot \E[ Y_j Y_k Y_l ]$. The term is thus zero,
  since $\E[Y_i] = 0$.

\item Three distinct keys, $\big| \{i,j,k,l\} \big| = 3$. This case is
  analyzed below.

\item One or two distinct keys: $\big| \{i,j,k,l\} \big| \le 2$. By
  3-independence of simple tabulation, all hash codes are independent,
  so the contribution of this term is the same as in the case of a
  fixed bin.
\end{itemize}

To summarize, the 4\th moment of bin $F(h(q))$ is the same as the 4\th
moment of a fixed bin, plus an additional term due to the case $\big|
\{i,j,k,l\} \big| = 3$. The remaining challenge is to understand terms
of the form $w_i^2 w_j w_k \E\big[ Y_i^2 Y_j Y_k \big]$. We first prove
the following, which is similar to Claim~\ref{clm:bcube}:

\begin{claim}   \label{clm:new-bcube}
For distinct $i,j,k$, $\E[ Y_i^2 Y_j Y_k ] ~=~ O(\frac{1}{m^2})$.
\end{claim}

\begin{proof}
By 3-independence of simple tabulation, $Y_i$ and $Y_j$ are
independent (these involve looking at the hashes of $i,j,q$). For an
upper bound, we can ignore all outcomes $Y_i^2 Y_j Y_k < 0$, i.e.~when
$Y_j$ and $Y_k$ have different signs. On the one hand, $Y_j = Y_k =
1-\frac{1}{m}$ with probability $O(\frac{1}{m^2})$. On the other hand,
if $Y_j = Y_k = -\frac{1}{m}$, the contribution to the expectation is
$O(\frac{1}{m^2})$.
\end{proof}

Assume $w_j \ge w_k$ by symmetry. If $k$ hashes independently of
$\{i,j,q\}$, the term is zero, since $\E[Y_k] = 0$ can be factored
out. Otherwise, the term contributes $O(w_i^2 w_j^2 / m^2)$ to the
sum. 

\begin{claim}
For any distinct $i,j,q$, there is a unique key $k$ such that $h(k)$
depends on $h(i), h(j), h(k)$.
\end{claim}

\begin{proof}
We claim that if any of $\{i,j,k,q\}$ has a unique position-character,
all keys are independent. Indeed, the key with a unique
position-character is independent of the rest, which are independent
among themselves by 3-independence.

Thus, any set $\{i,j,q\}$ that allows for a dependent $k$ cannot have
3 distinct position-characters on one position. In any position where
$i,j,$ and $q$ conincide, $k$ must also share that
position-character. If $i,j,$ and $q$ contain two distinct distinct
characters on some positon, $k$ must contain the one that appears
once. This determines $k$.
\end{proof}

For any $i$ and $j$, we see exactly one set $\{i,j,k\}$ that leads to
bad tuples. By an infinitesimal perturbation of the weights, each such
set leads to $\binom{4}{2} = 6$ tuples: we have to choose two
positions for $i$, and then $j$ is the remaining key with larger
weight. Thus, the total contribution of all terms $(i,j,k,l)$ with 3
distinct keys is $O\big(\sum_{i,j} \frac{w_i^2 w_j^2}{m^2} \big) =
O(\frac{1}{m^2}) (\sum_i w_i^2)^2$. This completes the proof of
\eqref{eq:moment-simple}.

\subsection{Independence Among Five Keys}   \label{sec:five-indep}

The section is dedicated to proving Theorem~\ref{thm:5th-indep}. We
first observe the following immediate fact:

\begin{fact}\label{fact:restrict} 
If, restricting to a subset of the characters (matrix columns), a key
$x\in X$ hashes independently from $X \setminus \{x\}$, then it also
hashes independently when considering all characters.
\end{fact}

If some key contains a unique character, we are done by
peeling. Otherwise, each column contains either a single value in all
five rows, or two distinct values: one appearing in two rows, and one
in three rows. By Fact~\ref{fact:restrict}, we may ignore the columns
containing a single value. For the columns containing two values,
relabel the value appearing three times with 0, and the one appearing
twice with 1. 
By Fact~\ref{fact:restrict} again, we may discard any duplicate
column, leaving at most $\binom{5}{2}$ distinct columns.

Since the columns have weight 2, the Hamming distance between two
columns is either 2 or 4.

\begin{lemma}\label{lem:hamming4} 
If two columns have Hamming distance 4, one hash value is independent.
\end{lemma}

\begin{proof}
By Fact~\ref{fact:restrict}, we ignore all other columns.  Up to
reordering of the rows, the matrix is:
{\tiny\[\left[\begin{array}{cc}
    0 & 1 \\
    0 & 1 \\
    1 & 0 \\
    1 & 0 \\
    1 & 1 \\
    \end{array}\right].\]}
By 3-independence of character hashing, keys 1, 3, and 5 are
independent. But keys 2 and 4 are identical to keys 1 and 3. Thus, key
5 is independent from the rest.
\end{proof}

We are left with the case where all column pairs have Hamming distance
2. By reordering of the rows, the two columns look like the matrix in
(a) below. Then, there exist only two column vectors that are at
distance two from both of the columns in (a):

{\tiny\[\textrm{\normalsize (a)} 
\left[\begin{array}{ccc}
    0 & 0\\
    0 & 1\\
    1 & 0\\
    1 & 1\\
    1 & 1\\
    \end{array}\right]
\textrm{\normalsize \qquad (b)}
\left[\begin{array}{ccc}
    0\\
    1\\
    1\\
    0\\
    1\\
\end{array}\right]
\textrm{\normalsize \qquad (c)}
\left[\begin{array}{ccc}
    1\\
    0\\
    0\\
    1\\
    1\\
\end{array}\right]
\]}

If the matrix does not contain column (b), then keys 4 and 5 are
identical, a contradiction. Thus, the matrix must contain columns (a)
and (b), with (c) being optional. If (c) appears, discard it by
Fact~\ref{fact:restrict}. We are left with the matrix:
{\tiny\[
\left[\begin{array}{ccc}
    0 & 0 & 0\\
    0 & 1 & 1\\
    1 & 0 & 1\\
    1 & 1 & 0\\
    1 & 1 & 1\\
    \end{array}\right]
\]}

Now, observe that the hash code of row 1 is just the xor of the hash
codes of rows 2--4, $h(1)=h(2)\oplus h_C(3)\oplus h_C(4)$. Indeed, the
codes of the one characters in each column cancel out, leaving us with
an xor of the zeros in each column. We claim that row 5 is independent
of rows 2--4. This immediately implies that row 5 is independent of
all others, since row 1 is just a function of rows 2--4.

Independence of row 5 from rows 2--4 follows by peeling. Each of rows
2, 3, and 4 have a position character not present in 5, so they are
independent of 5. This completes the proof of Theorem
\ref{thm:5th-indep}.

\subsection{Linear Probing with Fourth Moment Bounds}\label{sec:fourth-lp}
As in Section \ref{sec:linprobe} we study linear probing with $n$ stored
keys in a table of size $m$, and a query $q$ not among the stored
keys. We define the fill $\alpha=n/m$ and $\eps=1-\alpha$.
Pagh et al.~\cite{pagh07linprobe} presented a proof
that with 5-independent hashing, the expected number probes is 
$O(1/\eps^{13/6})$. We will here improve this to
the optimal $O(1/\eps^2)$, which is optimal even for a fully random
hash function. For the case of smaller fill, where
$\alpha\leq 1/2$, Thorup \cite{thorup09linprobe} proved that the expected 
number of filled entries probes is $O(\alpha)$ which is optimal even for fully random functions.

As discussed in Section \ref{sec:linprobe}, our goal is to
study the length $L$ of the longest filled interval containing
a point $p$ which may depend on $h(q)$, e.g., $p=h(q)$.  
To bound the probability that an interval $I$ is
full, we study more generally the case how the number $X_I$ of keys hashed to 
$I$  deviates from the mean $\alpha|I|$: if $I$ is full,
the deviation is by more than $\eps\alpha|I|$.

As we mentioned earlier, as an initial step Pagh et al.~\cite{pagh07linprobe} 
proved that if we consider the number of keys $X_I$ in an
interval $I$ which may depend on the hash of a query key,
then we have the following 4\th unweighted moment bound:
\begin{equation}\label{eq:fourth-base}
\Pr[X_I\geq \Delta+\alpha|I|]=
O\left(\frac{\alpha|I|+(\alpha|I|)^2}{\Delta^4}\right)
\end{equation}
This is an unweighted version of \req{eq:moment-simple} so \req{eq:fourth-base}
holds both with 5-independent hashing and with simple tabulation hashing.

As with our simple tabulation hashing, for each $i$, we consider
the event $\cC_{i,\delta,p}$ that for some point $p$ that may depend
on the hash of the query key, there is some an interval $I\ni p$, $2^i\leq |I|<2^{i+1}$ with relative deviation $\delta$. In perfect generalization
of \req{eq:fourth-base}, we will show that \req{eq:fourth-base} implies
\begin{equation}\label{eq:fourth-C}
\Pr[\cC_{i,\delta,p}]=
O\left(\frac{\alpha2^i+(\alpha2^i)^2}{(\delta\alpha2^i)^4}\right)
\end{equation}
First consider the simple case where $\delta\geq 1$. 
We apply Claim \ref{claim:large-delta}. Since $\delta\geq 1$, we
get $j=i-3$. The event $\cC_{i,\delta,p}$ implies that one
of the $2^5+1$ relevant $j$-bins has $(1+\frac{\delta}{2})\alpha 2^j$ 
keys. By \req{eq:fourth-base}, the probability of this
event is bounded by 
\[(2^5+1)O\left(\frac{\alpha2^j+(\alpha2^j)^2}{(\frac{\delta}2\alpha2^j)^4}\right)
=O\left(\frac{\alpha2^i+(\alpha2^i)^2}{(\delta\alpha2^i)^4}\right).\]
This completes the proof of \req{eq:fourth-C} when $\delta\geq 1$. 

Now consider the case where $\delta\leq 1$. If
$\alpha2^i\leq 1$, \req{eq:fourth-C} does not give a probability bound below
$1$, so we can assume $\alpha2^i>1$. 
Then \req{eq:fourth-C} simplifies to 
\begin{equation}\label{eq:fourth-C-small-delta}
\Pr[\cC_{i,\delta,p}]=
O(1/(\delta^4(\alpha2^i)^2)
\end{equation}
This time
we will apply Claim \ref{claim:small-delta}. Recall that a $j$-bin is dangerous
for level $i$ if its absolute deviation 
is $\Delta_{j,i} ~=~ \tfrac{\delta\alpha2^i}{24} / 2^{(i-j)/5}$. 
We defined $j_0$ be the smallest non-negative integer satisfying 
$\Delta_{j_0,i} \le \alpha 2^{j_0}$. If
$\cC_{i,\delta,p}$ happens, then for some $j\in\{j_0,...,i\}$,
one of the $2^{i-j+2}+1$ relevant $j$-bins is
dangerous for level $i$. By \req{eq:fourth-base}, the probability
of this event is bounded by $\sum_{j=j_0}^{i} O(P_j)$ where
\[P_j=2^{i-j} \frac{\alpha2^j+(\alpha2^j)^2}{\Delta_{i,j}^4}=O\left(2^{i-j} \frac{\alpha2^j+(\alpha2^j)^2}{\left(\delta\alpha2^i / 2^{(i-j)/5}\right)^4}\right)\]
Let $j_1=\lceil \log_2(1/\alpha)\rceil$. Note that $j_1\leq i$. For $j\geq j_1$,
we have $\alpha2^j+(\alpha2^j)^2=O((\alpha2^j)^2)$, so
\[P_j=O\left(2^{i-j} \frac{(\alpha2^j)^2}{\left(\delta\alpha2^i / 2^{(i-j)/5}\right)^4}\right)=O\left(\frac1{\delta^4(\alpha2^i)^22^{\frac15(i-j)}}\right).\]
We see that for $j\geq j_1$, the bound decreases exponentially with
$j$, so 
\begin{equation}\label{eq:top}
\sum_{j=j_1}^i P_j=O\left(1/(\delta^4(\alpha2^i)^2)\right).
\end{equation}
This is the desired bound from \req{eq:fourth-C-small-delta} for 
$\Pr[\cC_{i,\delta,p}]$, so we
are done if $j_1\leq j_0$.  However, suppose $j_0<j_1$. By
definition, we have $\alpha 2^{j_1}\leq 1$ and $\Delta_{i,j_1}\leq 1$,
so 
\[P_{j_1}=2^{i-j_1} \frac{\alpha2^{j_1}+(\alpha2^{j_1})^2}{\Delta_{i,{j_1}}^4}=\Omega(1).\]
This means that there is nothing
to prove, for with \req{eq:top}, we conclude that 
$\left(1/(\delta^4(\alpha2^i)^2)\right)=\Omega(1)$. Therefore
\req{eq:fourth-C-small-delta} does not promise any probability below 1.
This completes the proof of \req{eq:fourth-C}. As in Theorem \ref{thm:linprobe},
we can consider the more general event $\cD_{\ell,\delta,p}$ that there
exists an interval $I$ containing $p$ and of length at least $\ell$
such that the number of keys $X_I$ in $I$ deviates at least $\delta$ 
from the mean. Then, as a perfect generalization of \req{eq:fourth-C}, we
get
\begin{equation}\label{eq:fourth-lp}
\Pr[\cD_{\ell,\delta,p}]=\sum_{i\geq\log_2 \ell}\cC_{i,\delta,p}=
\sum_{i\geq\log_2 \ell}O\left(\frac{\alpha2^i+(\alpha2^i)^2}{(\delta\alpha2^i)^4}\right)
=O\left(\frac{\alpha\ell+(\alpha\ell)^2}{(\delta\alpha\ell)^4}\right)
\end{equation}
In the case of linear probing with fill $\alpha=1-\eps$, we 
worry about filled intervals. 
Let $L$ be the
length of the longest full interval containing the hash of a query
key. For $\eps\leq 1/2$ and $\alpha\geq 1/2$, we use $\delta=\eps$ and 
\[\Pr[\cD_{\ell,\eps,p}]=O(1/(\ell^2\eps^4))\]
so
\[\E[L]\leq \sum_{\ell=1}^m \Pr[\cD_{\ell,\eps,p}]
=\sum_{\ell=1}^m \min\left\{1,
O(1/(\ell^2\eps^4))\right\}
=O(1/\eps^2)\textnormal,\]
improving the $O(1/\eps^{\frac{13}6})$ bound from \cite{pagh07linprobe}.
However, contrasting the bounds with simple tabulation, the concentration
from \req{eq:fourth-lp} does not work well for higher moments, e.g.,
the variance bound we get with $\eps=1/2$ is $O(\log n)$,
and for larger moment $p\geq 3$, we only get a bound of 
$O(n^p/(n^2))=O(n^{p-2})$.

Now consider $\alpha\leq 1/2$. We use $\delta=1/(2\alpha)$ noting
that $(1+\delta)\alpha|I| =(\alpha+1/2)|I|<|I|$. Then
\[\Pr[\cD_{\ell,1/(2\alpha),p}]=O\left(\frac{\alpha\ell+(\alpha\ell)^2}{\ell^4}\right)\]
In particular, as promised in \req{eq:non-empty}, we have 
\[\Pr[L>0]=\Pr[\cD_{1,1/(2\alpha),p}]=O(\alpha+\alpha^2)=O(\alpha)\] 
 More generally for the mean,  
\begin{eqnarray*}
\E[L]&\leq &\sum_{\ell=1}^m \Pr[\cD_{\ell,1/(2\alpha),p}]\\
&=&\sum_{\ell=1}^m O\left(\frac{\alpha\ell+(\alpha\ell)^2}{\ell^4}\right)\\
&=&O(\alpha).
\end{eqnarray*}
This reproves the bound from Thorup \cite{thorup09linprobe}.

\appendix

\section{Experimental Evaluation}
\label{sec:exp-hash}
In this section, we make some simple experiments comparing simple
tabulation with other hashing schemes, both on their own, and in
applications.  Most of our experiments are the same as those in
\cite{thorup10kwise} except that we here include simple tabulation
whose relevance was not realized in \cite{thorup10kwise}. We will also
consider Cuckoo hashing which was not considered in \cite{thorup10kwise}.

Recall the basic message of our paper that simple tabulation in
applications shares many of the strong mathematical properties
normally associated with an independence of at least 5. For example,
when used in linear probing, the expected number of probes is constant
for any set of input keys. With sufficiently random input, this
expected constant is obtained by any universal hashing scheme
\cite{mitzenmacher08hash}, but other simple schemes fail on simple
structured inputs like dense intervals or arithmetic progressions,
which could easily occur in practice \cite{thorup10kwise}.

Our experiments consider two issues:
\begin{itemize*}
\item How fast is simple tabulation compared with other realistic
  hashing schemes on random input? In this case, the quality of the
  hash function doesn't matter, and we are only comparing their speed.

\item What happens to the quality on structured input. We consider the
  case of dense intervals, and also the hypercube which we believe
  should be the worst input for simple tabulation since it involves
  the least amount of randomness.
\end{itemize*}

We will now briefly review the hashing schemes considered in our
experiments. The focus will be on the common cases of 32 and 64 bit
keys.  If the initial keys are much bigger, we can typically first
apply universal hashing to get down to a smaller domain, e.g.,
collision free down to a domain of size $n^2$. To achieve
expected $O(1)$ time for linear probing, it suffices to map
universally to a domain of just $O(n)$ \cite{thorup09linprobe}.

\subsection{Multiplication-shift Hashing}
The fastest known hashing schemes are based on a multiplication followed
by a shift.

\paragraph{Univ-mult-shift} 
If we are satisfied with plain universal hashing, then as
shown in \cite{dietzfel97closest}, we pick a random odd number $a$ from 
the same $\ell$-bit domain as the keys. If the desired output
is $\ell_{out}$-bit keys, we compute the universal hash function:
\[h_{a}(x)= (a\texttt{*}x)\texttt{>>} (\ell-\ell_{out}).\]
This expression should be interpreted according to the C programming
language.  In particular, \texttt{*} denotes standard computer
multiplication where the result is truncated to the same size as that
of its largest operand. Here this means multiplication modulo
$2^\ell$.  Also, \texttt{>>} is a right shift taking out least
significant bits.  Mathematically, this is integer division by
$2^{\ell-\ell_{out}}$. Note that this scheme is far from
2-independent, e.g., if two keys differ in only their least
significant bit, then so does their hash values. However, the scheme
is universal which suffices, say, for expected constant times in
chaining.

\paragraph{2-indep-mult-shift}
For 2-independent hashing, we use the scheme from
\cite{dietzfel96universal}. We pick a random $2\ell$-bit multiplier
$a$ (which does not need to be odd), and a $2\ell$ bit number $b$.
Now we compute:
\[h_{a,b}(x)=(a\texttt* x\texttt+ b)\texttt{>>} (2\ell-\ell_{out}).\]
This works fine with a single 64-bit multiplication when
$\ell=32$. For $\ell=64$, we would need to simulate $128$-bit
multiplication. In this case, we have a faster alternative used for
string hashing \cite{thorup09linprobe}, viewing the key $x$ as
consisting of two $32$-bit keys $x_1$ and $x_2$. For a 2-independent
$32$-bit output, we pick three random 64-bit numbers $a_1$ and $a_2$
and $b$, and compute
\[h_{a_1,a_2,b}(x_1x_2)=((a_1\texttt +x_2)\texttt * (a_2\texttt +x_1)\texttt+ b)
\texttt{>>}32.\]
Concatenating two such values, we get a 64-bit 2-independent hash value using
just two 64-bit multiplications. 

\subsection{Polynomial Hashing}

For general $k$-independent hashing, we have the classic
implementation of Carter and Wegman \cite{wegman81kwise} by a degree
$k-1$ polynomial over some prime field:
\begin{equation}\label{eq:mod-p}
h(x)= \left( \sum_{i=0}^{k-1} a_i x^i \bmod p \right) \bmod 2^{\ell_{out}}
\end{equation}
for some prime $p\gg 2^{\ell_{out}}$ with each $a_i$ picked randomly
from $[p]$. If $p$ is an arbitrary prime, this method is fairly slow
because `$\bmod\;p$' is slow. However, Carter and Wegman
\cite{carter77universal} pointed out that we can get a fast
implementation using shifts and bitwise Boolean operations if $p$ is
a so-called Mersenne prime of the form $2^i-1$.

\paragraph{5-indep-Mersenne-prime} We use the above scheme for
5-independent hashing. For $32$-bit keys, we use $p=2^{61}-1$, and for
$64$-bit keys, we use $p=2^{89}-1$.

For the practical implementation, recall that standard 64-bit
multiplication on computers discards overflow beyond the 64 bits. For
example, this implies that we may need four 64-bit multiplications
just to implement a full multiplication of two numbers from
$[2^{61}-1]$. This is why specialized 2-independent schemes are much
faster. Unfortunately, we do not know a practical generalization for
higher independence.

\subsection{Tabulation-Based Hashing}

The basic idea in tabulation based schemes is to replace multiplications
with lookups in tables that are small enough to fit in fast memory.

\paragraph{simple-table} Simple tabulation is the basic example of
tabulation based hashing.  A key $x=x_1\cdots x_c$ is divided into $c$
characters. For $i=1 \twodots c$, we have a table $T_i$ providing a
random value $T_i[x_i]$ with a random value, and then we just return
the xor of all the $T_i[x_i]$.  Since the tables are small, it is easy
to populate them with random data (e.g.~based on atmospheric noise
\url{http://random.org}).  Simple tabulation is only 3-independent.

We are free to chose the size of the character domain,
e.g., we could use 16-bit characters instead of 8-bit characters, but
then the tables would not fit in the fast L1 cache.  The experiments
from \cite{thorup10kwise} indicate that 8-bit characters give much
better performance, and that is what we use here.

\paragraph{5-indep-TZ-table} 
To get higher independence, we can compute some additional ``derived
characters'' and use them to index into new tables, like the regular
characters.  Thorup and Zhang \cite{thorup04kwise, thorup10kwise}
presented a fast such scheme for 5-independent hashing. With $c=2$
characters, they simply use the derived character $x_1+x_2$. For
$c>2$, this generalizes with $c-1$ derived characters and a total of
$2c-1$ lookups for 5-independent hashing. The scheme is rather
complicated to implement, but runs well.

\subsection{Hashing in Isolation}
Our first goal is to time the different hashing schemes when run in
isolation. We want to know how simple tabulation compares in speed to
the fast multiplication-shift schemes and to the 5-independent schemes
whose qualities it shares. We compile and run the same C code on two
different computers:
\begin{description}
\item[32-bit computer:] Single-core Intel Xeon 3.2~GHz $32$-bit
  processor with 2048KB cache, 32-bit addresses and libraries.
\item[64-bit computer:] Dual-core 
  Intel Xeon 2.6~GHz $64$-bit processor with 4096KB cache, 64-bit
addresses and libraries.
\end{description}
Table \ref{results} presents the timings for the different
hashing schemes, first mapping 32-bit keys to 32-bit values, second mapping
64-bit keys to 64-bit values. 
\begin{table}[tb]
\begin{center}
\begin{tabular}{|c|l|r|r|}\hline
\multicolumn{2}{|c|}{Hashing random keys} & 32-bit computer & 64-bit computer\\ \hline
bits & hashing scheme &\multicolumn{2}{|c|}{hashing time (ns)}\\\hline
32 & univ-mult-shift  & 1.87 & 2.33\\
32 & 2-indep-mult-shift & 5.78 & 2.88\\
32 & 5-indep-Mersenne-prime &99.70& 45.06\\
32 & 5-indep-TZ-table & 10.12&12.66\\
32 & simple-table & 4.98 & 4.61\\\hline
64 & univ-mult-shift  & 7.05 & 3.14\\
64 & 2-indep-mult-shift & 22.91 & 5.90 \\
64 & 5-indep-Mersenne-prime &241.99& 68.67 \\
64 & 5-indep-TZ-table & 75.81 & 59.84 \\
64 & simple-table & 15.54 & 11.40\\\hline
\end{tabular}
\end{center}
\caption{Average time per hash computation for 10 million hash
  computations.}
\label{results}
\end{table}
Not surprisingly, we
see that the 64-bit computer benefits more than the 32-bit computer when 64-bit multiplications is critical; namely in univ-mult-shift for 
64 bits, 2-indep-mult-shift, and 5-indep-Mersenne-prime. 

As mentioned, the essential difference between our experiments and
those in \cite{thorup10kwise} is that simple tabulation is included,
and our interest here is how it performs relative to the other
schemes. In the case of 32-bits keys, we see that in both computers,
the performance of simple tabulation is similar to
2-indep-mult-shift. Also, not surprisingly, we see that it is more
than twice as fast as the much more complicated 5-indep-TZ-table.

When we go to 64-bits, it may be a bit surprising that simple
tabulation becomes more than twice as slow, for we do exactly twice as
many look-ups.  However, the space is quadrupled with twice as many
tables, each with twice as big entries, moving up from 1KB to 8KB, so
the number of cache misses may increase.

Comparing simple tabulation with the 2-indep-mult-shift, we see that
it is faster on the 32-bit computer and less than twice as slow on the
64-bit computer. We thus view it as competitive in speed. 

The competitiveness of simple tabulation compared with
multiplication-shift based methods agrees with the 
experiments of Thorup \cite{thorup00universal} from more than 10 years
ago on older computer architectures. The experiments from
 \cite{thorup00universal}  did not include schemes of higher independence.

The competitiveness of our cache based simple tabulation with
multiplication-shift based methods is to be expected both now and in
the future. One can always imagine that multiplication becomes faster
than multiplication, and vice versa. However, most data processing
involves frequent cache and memory access. Therefore, even if it was
technically possible, it would normally wasteful to
configure a computer with much faster multiplication than
cache. Conversely, however, there is lot of data processing that
does not use multiplication, so it is easier to imagine
real computers configured with faster cache than multiplication.

Concerning hardware, we note that simple tabulation is ideally
suited for parallel lookups of the characters of a key.
Also, the random data in the character tables
are only changed rarely in connection with a rehash. Otherwise
we only read the tables, which means that we could
potentially have them stored in simpler and faster EEPROM or flash
memory. This would also avoid conflicts with other applications in
cache.

\subsection{Linear Probing}
We now consider what happens when we use the different hashing schemes
with linear probing. In this case, the hash function needs good random
properties are need for good performance on worst-case input. We
consider $2^{20}$ 32-bit keys in a table with $2^{21}$ entries. The
table therefore uses 8MB space, which does not fit in the cache of
either computer, so there will be competition for the cache. Each
experiment averaged the update time over 10 million insert/delete
cycles. For each input type, we ran 100 such experiments on the same
input but with different random seeds for the hash functions.

\paragraph{Random input}
First we consider the case of a random input, where the randomization
properties of the hash function are irrelevant. This means that the
focus is on speed just like when we ran the experiments in
isolation. Essentially, the cost per update should be that of
computing the hash value, as in Table \ref{results}, plus a common
additive cost: a random access to look up the hash entry plus a number
of sequential probes. The average number of probes per update was
tightly concentrated around 3.28 for all schemes, deviating by less than
0.02 over the 100 experiments.

An interesting new issue is that the different schemes 
now have to compete with the linear probing table for the cache. 
In particular, this could hurt the tabulation based schemes. Another
issue is that when schemes are integrated with an application, the
optimizing compiler may have many more opportunities for pipelining etc.
The results for random input are presented in Table \ref{lin-probe-results}.
Within the 100 experiments, the deviation for each data point was
less than 1\%, and here we just present the mean.
\begin{table}[tb]
\begin{center}
\begin{tabular}{|l|r|r|}\hline
Linear probing with random keys & 32-bit computer & 64-bit computer\\ \hline
hashing scheme&\multicolumn{2}{|c|}{update time (nanoseconds)}\\\hline
univ-mult-shift  & 141 & 149\\
2-indep-mult-shift & 151 & 157\\
5-indep-Mersenne-prime &289& 245\\
5-indep-TZ-table & 177 &211\\
simple-table & 149 & 166\\\hline
\end{tabular}
\end{center}
\caption{Linear probing with random 32-bit keys. The time
  is averaged over 10 million updates to set with 1 million keys in linear
  probing table with 2 million entries.}
\label{lin-probe-results}
\end{table}

Compared with Table \ref{results}, we see that our 32-bit computer
performs better than the 64-bit computer on linear probing.  In Table
\ref{results} we had that the 64-bit processor was twice as fast at
the hash computations based on 64-bit multiplication, but in Table
\ref{lin-probe-results}, when combined with linear probing, we see
that it is only faster in the most extreme case of
5-indep-Mersenne-prime. One of the more surprising outcomes is that
5-indep-Mersenne-prime is so slow compared with the tabulation based
schemes on the 64-bit computer. We had expected the tabulation based
schemes to take a hit from cache competition, but the effect appears
to be minor.

The basic outcome is that simple tabulation in linear probing with random
input is competitive with the fast multiplication-shift based scheme and
about 20\% faster than the fastest 5-independent scheme (which is much
more complicated to implement). We note that we cannot hope for a 
big multiplicative gain in this case, since the cost is dominated by
the common additive cost from working the linear probing table itself.

\paragraph{Structured input}
We now consider the case where the input keys are structured in the
sense of being drawn in random order from a dense interval: a commonly
occurring case in practice which is known to cause unreliable
performance for most simple hashing schemes \cite{pagh07linprobe,
  patrascu10kwise-lb,thorup10kwise}. The results are shown in Figure
\ref{fig:perm}. For each hashing scheme, we present the average number
of probes for each of the 100 experiments as a cumulative distribution
function. We see that simple tabulation and the 5-independent schemes
remain tightly concentrated while the multiplication-shift schemes
have significant variance, as observed also in
\cite{thorup10kwise}. This behavior is repeated in the timings on the
two computers, but shifted due to difference in speed for the hash
computations.
\begin{figure}[hptb]
\centering  
\subfigure[Probe]{
  \includegraphics[width=\figurewidthA]{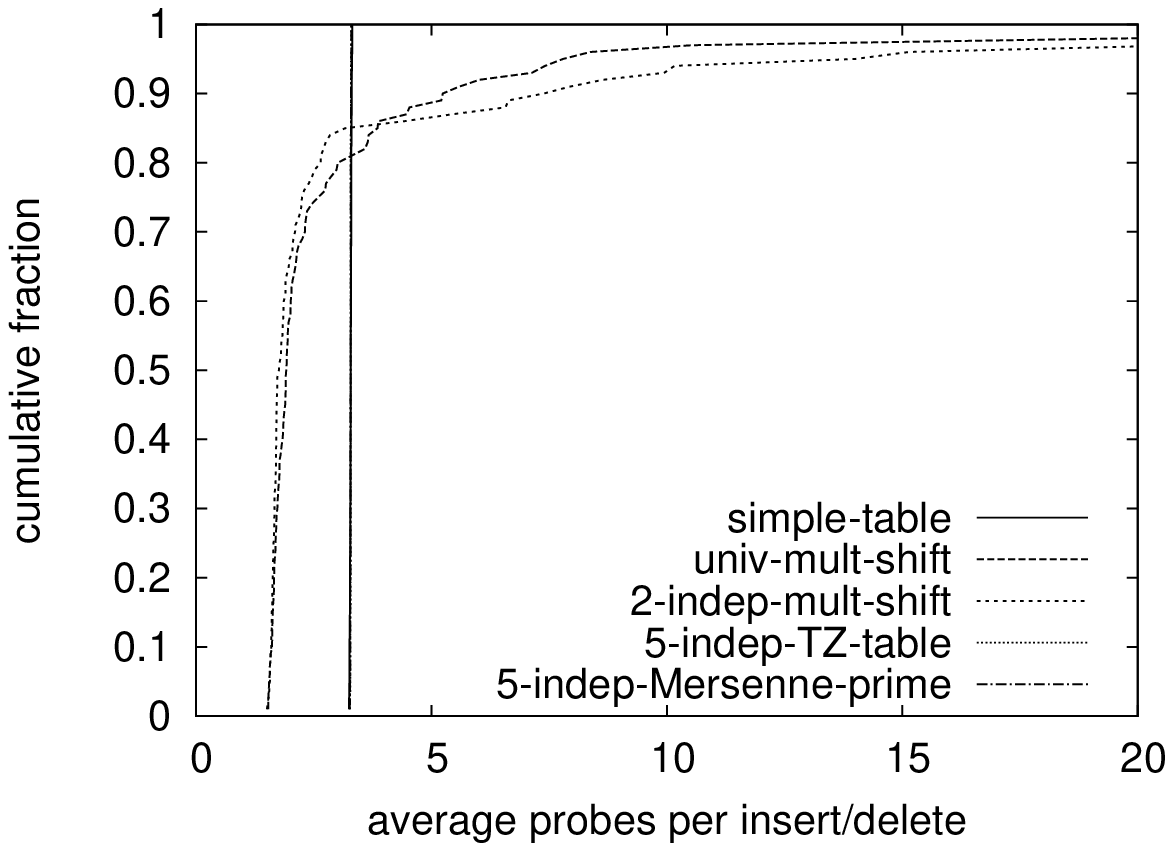}}\\
\subfigure[Time 32-bit computer]{
  \includegraphics[width=\figurewidthA]{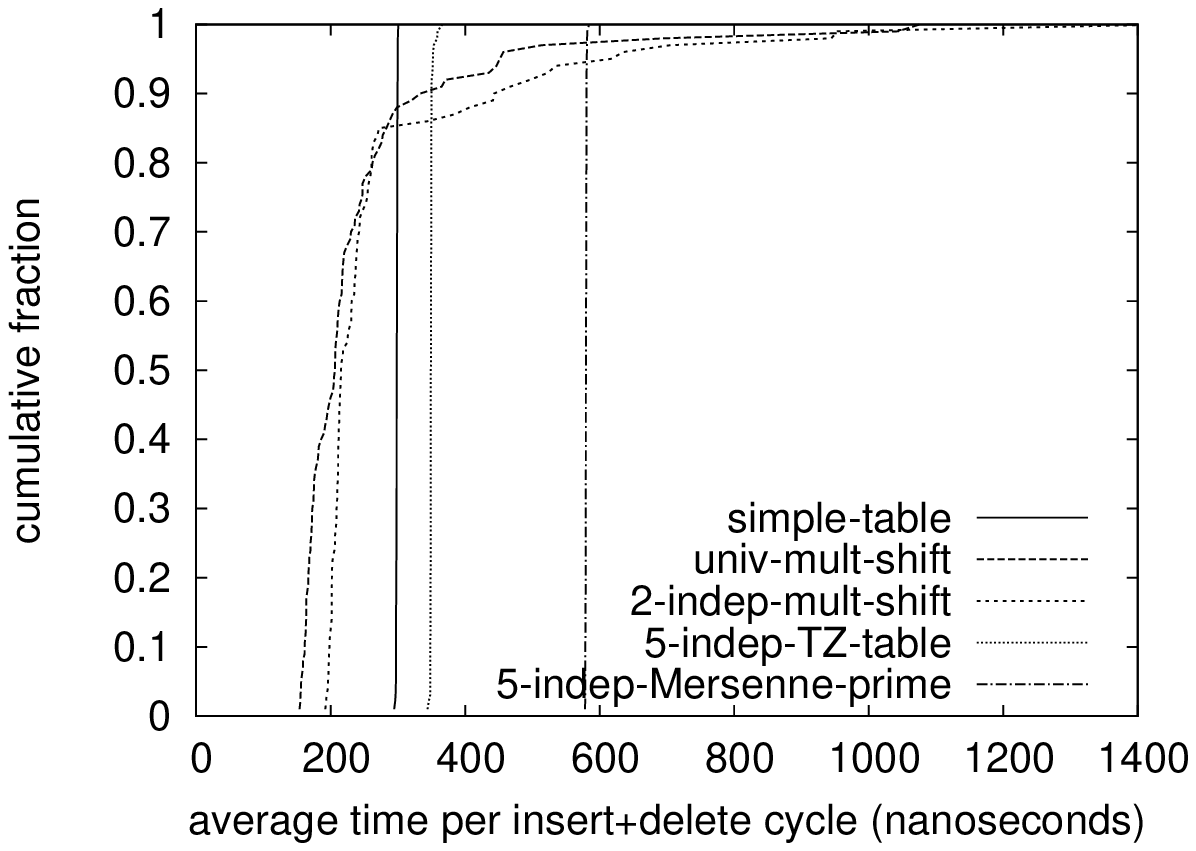}}
\subfigure[Time 64-bit computer]{
  \includegraphics[width=\figurewidthA]{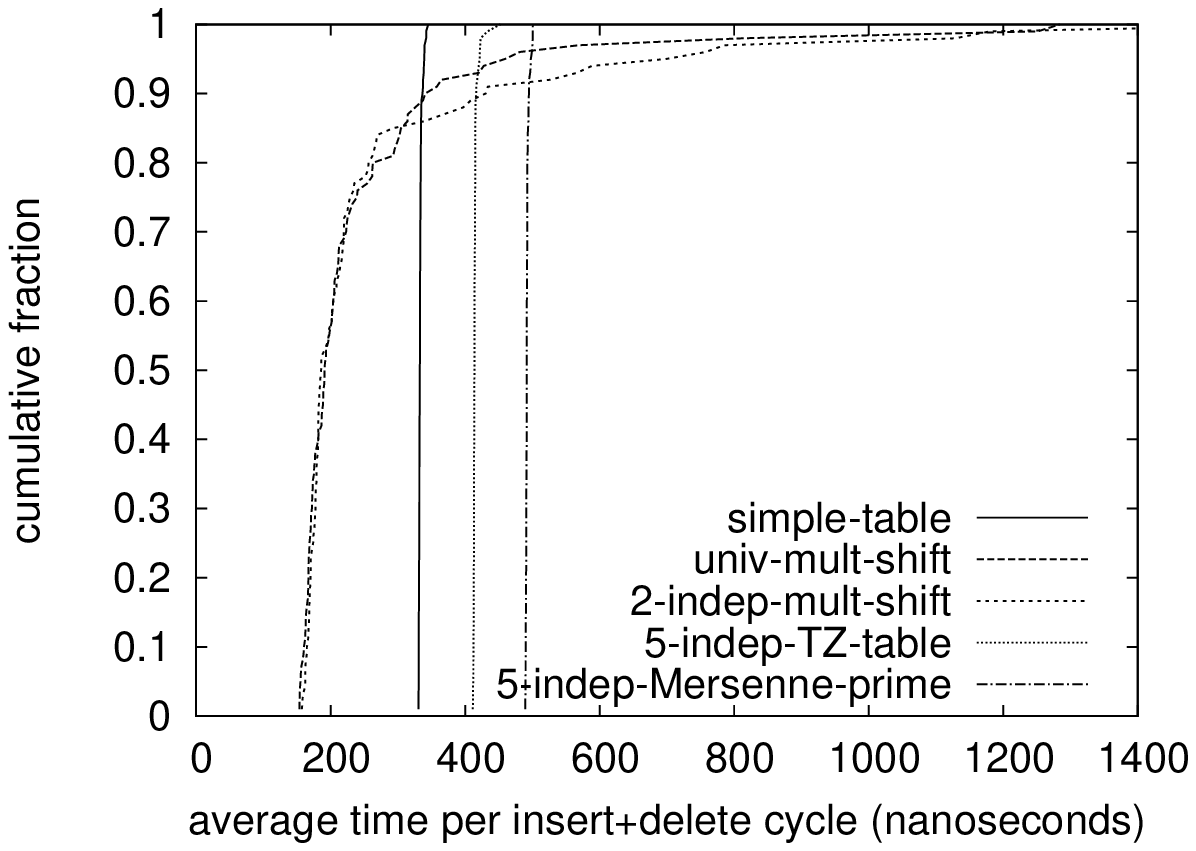}}
\caption{Keys from dense interval. The multiplication-shift schemes
sometimes use far more probes, which also shows in much longer running
times.} \label{fig:perm}
\end{figure}

Thus, among simple fast hashing schemes, simple tabulation stands out
in not failing on a dense interval. Of course, it might be that simple
tabulation had a different worst-case input. A plausible guess is that
the worst instance of simple tabulation is the hypercube, which
minimizes the amount of random table entries used. In our case, for
$2^{20}$ keys, we experimented with the set $[32]^4$, i.e., we only
use $32$ values for each of the $4$ characters. The results for the
number of probes are presented in Figure \ref{fig:hcube}.
\begin{figure}[hptb]
\centering
  \includegraphics[width=\figurewidthA]{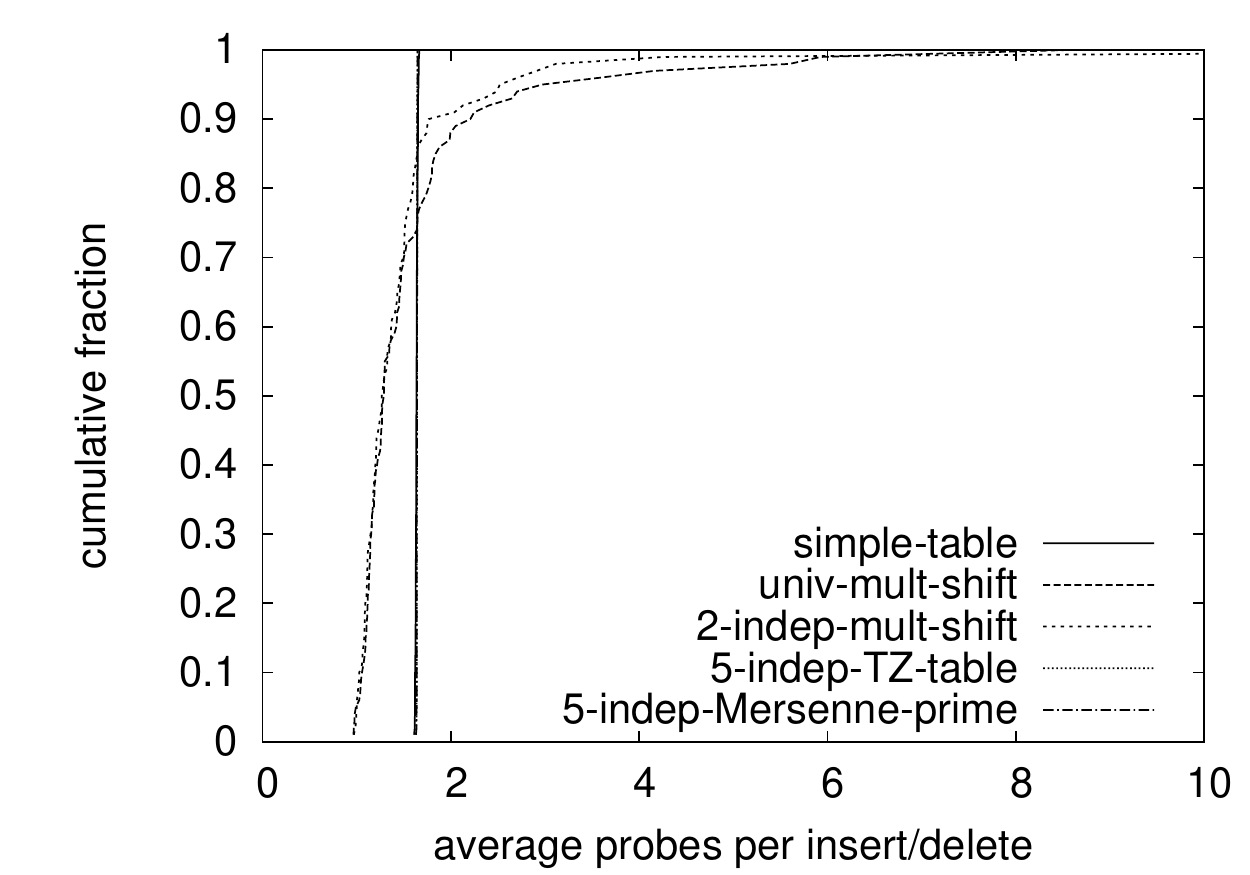}
\caption{Keys from hyper cube.}\label{fig:hcube}
\end{figure}

Thus, simple tabulation remains extremely robust and tightly
concentrated, but once again the multiplication-shift schemes fail
(this time more often but less badly). The theoretical explanation
from \cite{patrascu10kwise-lb} is that multiplication-shift fails on
arithmetic sequences, and in the hypercube we have many different but
shorter arithmetic sequences. It should be said that although it
cannot be seen in the plots, the structured inputs did lead to more
deviation in probes for simple tabulation: the deviation from 3.28
over 100 independent runs grew from below 0.5\% with random input to
almost 1\% with any of the structured input.

Obviously, no experiment can confirm that simple tabulation is robust
for all possible inputs. Our theoretical analysis implies strong
concentration, e.g., in the sense of constant variance, yet the hidden
constants are large.  Our experiments suggest that the true constants
are very reasonable.

\paragraph{Cuckoo hashing}
Our results show that the failure probability in constructing a cuckoo
hashing table is $O(n^{-1/3})$. A pertinent question is whether the
constants hidden by the $O$-notation are too high from a practical
point of view. Experiments cannot conclusively prove that this
constant is always small, since we do not know the worst
instance. However, as for linear probing, a plausible guess that the instance eliciting
the worst behavior is a hypercube: $S = A^c$, for $A \subset
\Sigma$. We made $10^5$ independent runs with the following input
instances:
\begin{description*}
\item[32-bit keys:] Tabulation uses $c=4$ characters. We set $A =
  [32]$, giving $32^4 = 2^{20}$ keys in $S$. The empirical success
  probability was $99.4\%$.

\item[64-bit keys:] Tabulation uses $c=8$ characters. We set $A=8$,
  giving $8^8 = 2^{24}$ keys in $S$. The empirical success probability
  was $97.1\%$.
\end{description*}

\noindent
These experiments justify the belief that our scheme is effective in
practice. 

It has already been shown conclusively that weaker multiplication
schemes do not perform well. Dietzfelbinger and Schellbach
\cite{dietzfel09cuckoo} show analytically that, when $S$ is chosen
uniformly at random from the universe $[n^{12/11}]$ or smaller, cuckoo
hashing with 2-independent multiplicative hashing fails with
probability $1-o(1)$. This is borne out in the experiments of
\cite{dietzfel09cuckoo}, which give failure probability close to 1 for
random sets that are dense in the universe. On the other hand, the
more complicated tabulation hashing of Thorup and Zhang
\cite{thorup04kwise} will perform at least as well as simple
tabulation (that algorithm is a superset of simple tabulation).

A notable competitor to simple tabulation is a tailor-made tabulation
hash function analyzed by Dietzfelbinger and Woelfel
\cite{dietzfel03tabhash}. This function uses two arrays of size $r$
and four $d$-independent hash functions to obtain failure probability
$n/r^{d/2}$.

Let us analyze the parameters needed in a practical implementation.
If we want the same space as in simple tabulation, we can set
$r=2^{10}$ (this is larger than $\Sigma=256$, because fewer tables are
needed). For a nontrivial failure probability with sets of $2^{20}$
keys, this would require 6-independence. In principle, the
tabulation-based scheme of \cite{thorup04kwise} can support
6-independence with $5c-4$ tables (and lookups). This scheme has not
been implemented yet, but based on combinatorial complexity is
expected to be at least twice as slow as the 5-independent scheme
tested in Table~\ref{results} (i.e.~4-8 times slower than simple
tabulation). Alternatively, we can compute four 6-independent hash
functions using two polynomials of degree 5 on 64-bit values
(e.g.~modulo $2^{61}-1$). Based on Table~\ref{results}, this would be
two orders of magnitude slower than simple tabulation. With any of the
alternatives, the tailor-made scheme is much more complicated to
implement.

\section{Chernoff Bounds with Fixed Means}\label{sec:chernoff-appendix}

We will here formally establish that the standard Chernoff bounds
hold if when each variable have a fixed mean even if the 
of the variables are not independent.
Below shall use the notation that if we have variables $x_1,x_2,...$, then
$x_{<i}=\{x_j\}_{j<i}$. In particular, $\sum x_{<i}=\sum_{j<i} x_j$.

\begin{proposition}\label{prop:chernoff-fix-mean}
Consider $n$ possibly dependent random variables
$X_1, X_2, \dots, X_n\in [0,1]$. Suppose for
each $i$ that $\E[X_i]=\mu_i$ is fixed no matter
the values of $X_1,...,X_{i-1}$, that is, for any values
$x_1,...,x_{i-1}$, $\E[X_i|X_{<i}=x_{<i}]=\mu_i$.
Let $X=\sum_i X_i$ and $\mu=\E[X]=\sum_i \mu_i$. Then
for any $\delta>0$, the bounds are:
\[\Pr[X\ge (1+\delta)\mu] \le \left(
   \frac{e^\delta}{(1+\delta)^{(1+\delta)}} \right)^{\mu}
\qquad\qquad
\Pr[X\le (1-\delta)\mu] \le \left(
   \frac{e^{-\delta}} {(1-\delta)^{(1-\delta)}} \right)^{\mu}
\]
\end{proposition}
\begin{proof}
The proof is a simple generalization over the standard proof when
the $X_i$ are independent.
We wish to bound the probability of $X\ge (1+\delta)\mu$.
To do this we will prove that
\[\E[(1+\delta)^X]\leq e^\mu.\]
The proof will be by induction on $n$. Let 
\begin{eqnarray*}
\E[(1+\delta)^X]&=&\sum_{ x_{<n}}\left\langle
\Pr[ X_{<n}= x_{<n}]\times\E\left[(1+\delta)^X\,|\,
X_{<n}=x_{<n}\right]\right\rangle\\
&=&\sum_{x_{<n}}\left\langle
\Pr[X_{<n}=x_{<n}]\times (1+\delta)^{\sum x_{<n}}\times\E\left[(1+\delta)^{X_n}\,|\,
 X_{<n}= x_{<n}\right]\right\rangle.\\
\end{eqnarray*}
Now, for any random variable $Y\in[0,1]$, by convexity,
\[\E\left[(1+\delta)^Y\right]\leq \E[Y](1+\delta)+1-\E[Y]=
1+\delta\E[Y]\leq e^{\delta \E[Y]}\]
Therefore, since $E[X_n |  X_{<n}= x_{<n}]=\mu_n$ for any value $ x_{<n}$ of $ X_{<n}$,
\[\E\left[(1+\delta)^{X_n}\,|\, X_{<n}= x_{<n}\right]\leq
e^{\delta \mu_n}.\]
Thus
\begin{eqnarray*}
\E[(1+\delta)^X]&=&\sum_{x_{<n}}\left\langle
\Pr[ X_{<n}= x_{<n}]\times (1+\delta)^{\sum x_{<n}}\times\E\left[(1+\delta)^{X_n}\,|\,
 X_{<n}= x_{<n}\right]\right\rangle.
\\
&\leq &\sum_{x_{<n}}\left\langle
\Pr[X_{<n}=x_{<n}]\times (1+\delta)^{\sum x_{<n}}\times e^{\delta \mu_n}\right\rangle\\
&=&\E\left[(1+\delta)^{\sum X_{<n}}\right]\times e^{\delta \mu_n}\\
&\leq&e^{\delta\sum \mu_{<n}}\times e^{\delta \mu_n}=e^{\delta\mu}.
\end{eqnarray*}
The last inequality followed by induction. Finally, by Markov's inequality,
we conclude that
\[\Pr[X\ge (1+\delta)\mu] \leq \frac{\E\left[(1+\delta)^X\right]}{(1+\delta)^{(1+\delta)\mu}} \leq \frac{e^{\delta\mu}}{(1+\delta)^{(1+\delta)\mu}}
=\left(
   \frac{e^\delta}{(1+\delta)^{(1+\delta)}} \right)^{\mu}.\]
The case $X\leq (1-\delta)\mu$ follows by a symmetric argument.
\end{proof}

{
\bibliographystyle{alpha} 
\bibliography{../../general}
}

\end{document}